\documentclass[11pt,a4paper,final,twoside]{article}
\usepackage[utf8]{inputenc}
\usepackage[english]{babel}
\usepackage{amsfonts}
\usepackage{amssymb}
\usepackage{graphicx}

\RequirePackage[OT1]{fontenc}
\RequirePackage{graphicx}
\RequirePackage{amsthm}
\RequirePackage{natbib}
\RequirePackage[cmex10]{amsmath}
\RequirePackage[colorlinks,citecolor=blue,urlcolor=blue]{hyperref}
\RequirePackage{hypernat}

\usepackage[utf8]{inputenc}
\usepackage{amsfonts}
\usepackage{amssymb}
\usepackage{mathrsfs}
\usepackage{bbm}

\usepackage{algorithmic}
\usepackage{algorithm}
\usepackage{booktabs}
\usepackage{multirow}
\usepackage{color}
\usepackage{url}
\usepackage[table]{xcolor}
\usepackage{picture}
\usepackage{tikz}
\usetikzlibrary{automata,positioning,shapes.geometric}

\usepackage{dsfont}
\usepackage{color}

\usepackage{ucs}
\usepackage{bm}

\usepackage{mathtools}
\usepackage{fancyhdr}
\usepackage{enumerate}
\usepackage{url}

\usepackage{soul}

\newtheorem{theorem}{Theorem}
\newtheorem{corollary}[theorem]{Corollary}

\newtheorem{proposition}[theorem]{Proposition}

\newtheorem{example}{Example}

\newtheorem{assumpM}{Assumption}

\newtheorem{assumpN}{Assumption}

\newcommand{\eps}{\varepsilon}

\newcommand{\bY}{\boldsymbol Y}
\newcommand{\bZ}{\bm{Z}}

\newcommand{\bz}{\bm{z}}
\newcommand{\bU}{\bm{U}}
\newcommand{\bV}{\bm{V}}

\newcommand{\bu}{\bm{u}}

\newcommand{\bss}{\bm{s}}

\newcommand{\bN}{\bm{N}}

\newcommand{\zero}{\bm{0}}

\newcommand{\bvartheta}{\bm{\vartheta}}

\newcommand{\bth}{\bm{\theta}}

\newcommand{\brho}{\bm{\rho}}
\newcommand{\bnu}{\bm{\nu}}

\newcommand{\bBeta}{\bm{\eta}}

\newcommand{\E}{\mathsf{E}}

\newcommand{\Var}{\mathsf{Var}}

\newcommand{\dist}{\mathsf{D}}

\newcommand{\prob}{\mathsf{P}}

\newcommand{\ud}{\mbox{d}}

\DeclareMathOperator{\tr}{tr}

\floatname{algorithm}{Procedure}

\usepackage{mathpazo}
\usepackage{eucal}
\usepackage[left=2.5cm, right=2.5cm, top=2.5cm, bottom=2.5cm]{geometry}
\usepackage[doublespacing]{setspace}

\usepackage{authblk}
\usepackage{fancyhdr}
\author[1]{Mat\'{u}\v{s} Maciak}%\thanks{\texttt{matus.maciak@mff.cuni.cz}}}
\author[2]{Ostap Okhrin}%\thanks{\texttt{ostap.okhrin@tu-dresden.de}}}
\author[1]{Michal Pe\v{s}ta\thanks{Corresponding author; Address: Sokolovsk\'{a}~49/83, 18675 Prague~8, Czech Republic; Email: \texttt{michal.pesta@mff.cuni.cz}}}
\affil[1]{Charles University, Prague, Czech Republic, Faculty of Mathematics and Physics, Department of Probability and Mathematical Statistics}
\affil[2]{Technische Universit\"at Dresden, Germany, Chair of Statistics, Institute of Economics and Transport, School of Transportation}

\title{Infinitely Stochastic Micro Forecasting}
\date{\small\today}

\allowdisplaybreaks[1]

\begin{document}
	
\maketitle

\begin{abstract}
Forecasting costs is now a~front burner in empirical economics. We propose an~unconventional tool for stochastic prediction of future expenses based on the individual (micro) developments of recorded events. Consider a~firm, enterprise, institution, or state, which possesses knowledge about particular historical events. For each event, there is a~series of several related subevents: payments or losses spread over time, which all leads to an infinitely stochastic process at the end. Nevertheless, the issue is that some already occurred events do not have to be necessarily reported. The aim lies in forecasting future subevent flows coming from already reported, occurred but not reported, and yet not occurred events. Our methodology is illustrated on quantitative risk assessment, however, it can be applied to other areas such as startups, epidemics, war damages, advertising and commercials, digital payments, or drug prescription as manifested in the paper. As a~theoretical contribution, inference for infinitely stochastic processes is developed. In particular, a~non-homogeneous Poisson process with non-homogeneous Poisson processes as marks is used, which includes for instance the Cox process as a~special case.
\end{abstract}

\medskip

\begin{small}
\medskip

%\noindent \emph{JEL codes:} C53, C13, C32, C33, C63, G22, G32, I13, M13
\end{small}

\section{Introduction}\label{sec:intro}
Human as well as monetary losses and uncertainty about their extent are one of the main sources of risk. A~probabilistic prediction of the future monetary losses lies on the cutting edge of quantitative risk assessment, for instance, valuation of operational risk in banking or reserving risk in insurance, while the study of human losses is of particular interest in conflict solution and epidemics modeling. %Moreover, we exemplify our loss prediction method on the latter mention type of risk.
%Although, the proposed prediction approach together with the underlying stochastic procedures can be easily transferred to other areas as demonstrated by several case examples later on.
We propose a~general prediction methodology, which together with the underlying stochastic procedures is applicable to various areas as demonstrated by case examples later on.

Let us define the general structure of our model on the basis of a~financial example. The event's \emph{`lifetime'} can be described as follows: The $i$th loss occurs at the \emph{occurrence time}, which is denoted by $T_i$. Such a~loss is often reported (e.g., to a~financial company) not immediately after the event, but for various reasons, after the \emph{reporting delay} (waiting time) $W_i$, which is the time difference between the \emph{occurrence epoch} (event time) and the \emph{observation epoch} (reporting time). Furthermore, $Z_i=T_i+W_i$ stands for the $i$th \emph{reporting (notification) time}. The contemplated cash flows are visualized in Figure~\ref{fig:illustration}, which elucidates the whole framework behind the loss reporting process together with the time developments of the losses. 
%\vspace{.5cm}
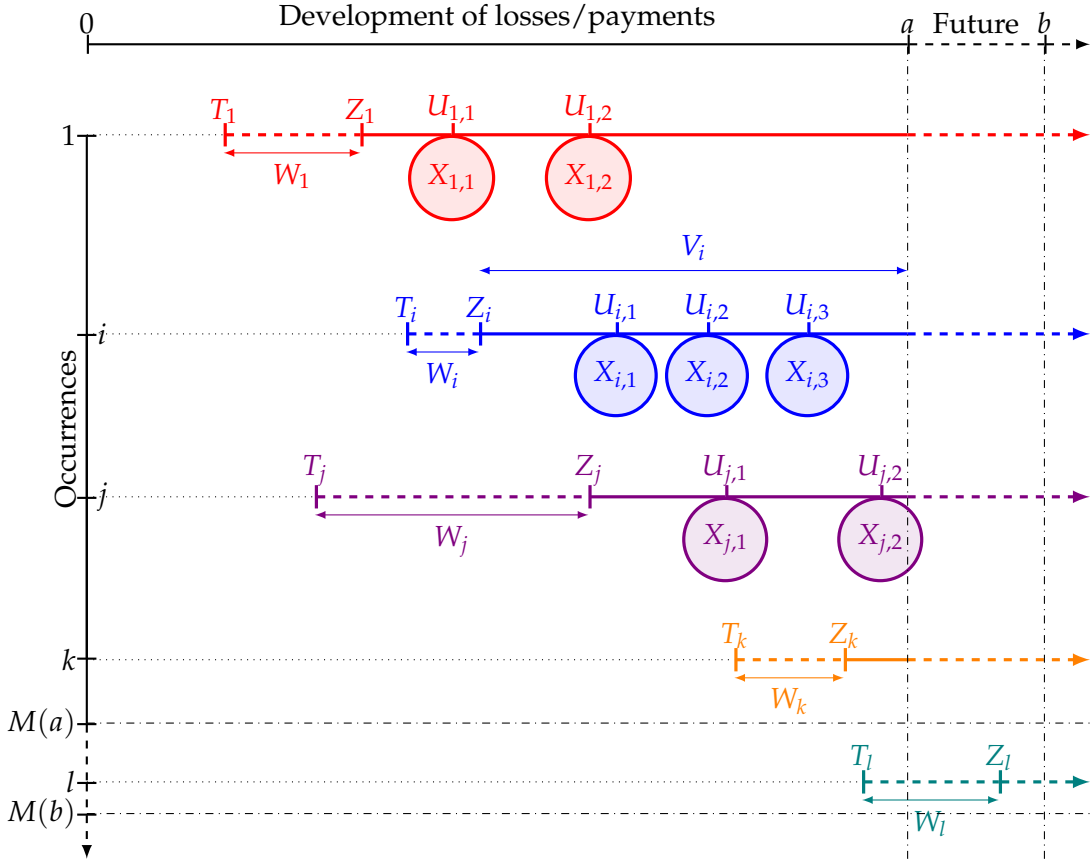
\begin{figure}[!ht]
	\centering
	\begin{tikzpicture}[>=latex, xscale=12, yscale=12]
	
	\draw[|-, thick] (-0.1,0.8) node[left]{$1$} -- (-0.1,0.15) node[midway, above, rotate=90]{Occurrences};
	
	\draw[|-, thick] (-0.1,0.58) node[right]{$i$} -- (-0.1,0.40) node[right]{$j$};
	
	\draw[|-|, thick] (-0.1,0.40) -- (-0.1,0.22) node[left]{$k$};
	
	\draw[|-, dashed, thick] (-0.1,0.15) node[left]{$M(a)$} -- (-0.1,0.085);
	
	\draw[|-, dashed, thick] (-0.1,0.085) node[left]{$l$} -- (-0.1,0.05);
	
	\draw[|->, dashed, thick] (-0.1,0.05) node[left]{$M(b)$} -- (-0.1,0.0);
	
	\draw[|-, thick] (-0.1,0.9) node[left, above]{$0$} -- (0.8,0.9) node[midway, above]{Development of losses/payments};
	
	\draw[|-, dashed, thick] (0.8,0.9) node[left, above]{$a$} -- (0.95,0.9) node[midway, above]{Future};
	
	\draw[|->, dashed, thick] (0.95,0.9) node[left, above]{$b$} -- (1,0.9);
	
	%%%\node[left] at (-0.05,0.8) {$i=1$};
	
	%red
	\draw[dotted] (-0.1,0.8) -- (0.05,0.8);
	
	\draw[|-, dashed, color=red, very thick] (0.05,0.8) node[above]{$T_1$} -- (0.20,0.8);
	\draw[<->, color=red] (0.05,0.78) -- (0.20,0.78) node[midway, below]{$W_1$};
	\draw[|-, color=red, very thick] (0.20,0.8) node[above]{$Z_1$} -- (0.30,0.8);
	\draw[|-, color=red, very thick] (0.30,0.8) node[above]{$U_{1,1}$} node[draw,circle,below,fill=red!10] {$X_{1,1}$} node {} -- (0.45,0.8);
	\draw[|-, color=red, very thick] (0.45,0.8) node[above]{$U_{1,2}$} node[draw,circle,below,fill=red!10] {$X_{1,2}$} node {} -- (0.8,0.8);
	
	\draw[->, dashed, color=red, very thick] (0.8,0.8) -- (1,0.8);
	
	%blue
	\draw[dotted] (-0.1,0.58) -- (0.25,0.58);
	
	\draw[|-, dashed, color=blue, very thick] (0.25,0.58) node[above]{$T_i$} -- (0.33,0.58);
	\draw[<->, color=blue] (0.25,0.56) -- (0.33,0.56) node[midway, below]{$W_i$};
	\draw[|-, color=blue, very thick] (0.33,0.58) node[above]{$Z_i$} -- (0.48,0.58);
	\draw[|-, color=blue, very thick] (0.48,0.58) node[above]{$U_{i,1}$} node[draw,circle,below,fill=blue!10] {$X_{i,1}$} node {} -- (0.58,0.58);
	\draw[|-, color=blue, very thick] (0.58,0.58) node[above]{$U_{i,2}$} node[draw,circle,below,fill=blue!10] {$X_{i,2}$} node {} -- (0.69,0.58);
	\draw[|-, color=blue, very thick] (0.69,0.58) node[above]{$U_{i,3}$} node[draw,circle,below,fill=blue!10] {$X_{i,3}$} node {} -- (0.8,0.58);
	
	\draw[->, dashed, color=blue, very thick] (0.8,0.58) -- (1,0.58);
	
	\draw[<->, color=blue] (0.33,0.65) -- (0.8,0.65) node[midway, above]{$V_i$};
	
	%violet
	\draw[dotted] (-0.1,0.4) -- (0.15,0.4);
	
	\draw[|-, dashed, color=violet, very thick] (0.15,0.4) node[above]{$T_j$} -- (0.45,0.4);
	\draw[<->, color=violet] (0.15,0.38) -- (0.45,0.38) node[midway, below]{$W_j$};
	\draw[|-, color=violet, very thick] (0.45,0.4) node[above]{$Z_j$} -- (0.60,0.4);
	\draw[|-, color=violet, very thick] (0.60,0.4) node[above]{$U_{j,1}$} node[draw,circle,below,fill=violet!10] {$X_{j,1}$} node {} -- (0.77,0.4);
	\draw[|-, color=violet, very thick] (0.77,0.4) node[above]{$U_{j,2}$} node[draw,circle,below,fill=violet!10] {$X_{j,2}$} node {} -- (0.8,0.4);
	
	\draw[->, dashed, color=violet, very thick] (0.8,0.4) -- (1,0.4);
	
	%orange
	\draw[dotted] (-0.1,0.22) -- (0.61,0.22);
	
	\draw[|-, dashed, color=orange, very thick] (0.61,0.22) node[above]{$T_k$} -- (0.73,0.22);
	\draw[<->, color=orange] (0.61,0.20) -- (0.73,0.20) node[midway, below]{$W_k$};
	\draw[|-, color=orange, very thick] (0.73,0.22) node[above]{$Z_k$} -- (0.8,0.22);
	
	\draw[->, dashed, color=orange, very thick] (0.8,0.22) -- (1,0.22);
	
	%teal
	\draw[dotted] (-0.1,0.085) -- (0.75,0.085);
	
	\draw[|-, dashed, color=teal, very thick] (0.75,0.085) node[above]{$T_l$} -- (0.90,0.085);
	\draw[<->, color=teal] (0.75,0.065) -- (0.90,0.065) node[midway, below]{$W_l$};
	
	\draw[|->, dashed, color=teal, very thick] (0.9,0.085) node[above]{$Z_l$} -- (1,0.085);

	% dash dot
	\draw[dash dot] (0.8,0.0) -- (0.8,0.9);
	
	\draw[dash dot] (0.95,0.0) -- (0.95,0.9);
	
	\draw[dash dot] (-0.1,0.05) -- (1.0,0.05);
	
	\draw[dash dot] (-0.1,0.15) -- (1.0,0.15);
	
	\end{tikzpicture}
	\caption{Scheme of the event occurrence process and the event development processes.}
	\label{fig:illustration}
\end{figure}

Our observation history for the reported losses is a~time interval $[0,a]$, where~$a$ is the present time. The main aim is to predict the losses, which are going to be reported in the future time horizon $(a,b]$, and simultaneously to predict the development of the losses within the time interval $(a,b]$, that have already occurred before time~$a$, but are not settled yet. Some of them are already incurred (i.e., occurred within $[0,a]$, but will be reported after time point~$a$). The observed loss data are \emph{truncated} in the way that we observe only the reported losses, i.e., $Z_i\leq a$. Without loss of generality, we assume that the reporting times are chronologically ordered such that $Z_{i_1}\leq Z_{i_2}$ for $i_1<i_2$. After some internal process is carried out, the company pays $N_i(t)$ payments till time point $t$, or $N_i(\infty)$ payments in order to fully settle the~$i$th loss. The amount of the $k$th payment within the $i$th loss paid at time $U_{i,k}$ is represented by $X_{i,k}$, for $k=1,\ldots,N_i(a)$. The time window from the reporting time $Z_i$ up to the last observed (available) time~$a$ for the $i$th loss has length $V_i$, i.e., $V_i=a-Z_i$. One can think of the reporting times $Z_i$'s as the arrival times of the counting process $\{M(t)\}_{t\geq 0}$ and the payment times $U_{i,k}$'s as the arrival times of the counting processes $\{N_i(t)\}_{t\geq 0}$ for $i\in\mathbb{N}$. Assuming that there are $i=1,\ldots,M(a)$ losses already reported, we observe a~collection
\[
\{T_i,Z_i,\{U_{i,j}\}_{j=1,\ldots,N_i(a)},\{X_{i,j}\}_{j=1,\ldots,N_i(a)}\}_{i=1,\ldots,M(a)}
\]
or, alternatively and equivalently, $\{Z_i,W_i,\{N_i(t)\}_{t\in[0,a]},\{X_{i,j}\}_{j=1,\ldots,N_i(a)}\}_{i=1,\ldots,M(a)}$.

\subsection{Motivation and applications}
The proposed class of models---\emph{infinitely stochastic processes}---is a~very rich and general class that nests, for examples, doubly stochastic (Cox) processes. Our approach and results are motivated in the context of several applications taken from the empirical economics literature.

\paragraph{Case~1: Operational risk}
Banks and other financial institutions have to face operational risk covering fraud, system failures, security, privacy protection, terrorism, legal risks, employee compensation claims, physical (e.g., infrastructure shutdown) or environmental risks. As pointed out in~\cite{OR2007}, some large banks prefer to use their own formal definition of operational risk. For instance, Deutsche Bank~(2017) defines operational risk as ``the risk of loss resulting from inadequate or failed internal processes, people and systems or from external events, and includes legal risk.'' Recent developments for operational risk prediction---comprehensively summarized by~\cite{BLM2018}---reveal that challenges like truncated data and non-homogeneous processes have to be handled in operational risk modeling. The empirical literature, for instance \cite{Cohen2018}, suggests that operational risk capital models can be based on the loss distribution approach. Generally, a~loss~$i$ corresponding to operational risk is occurred at~$T_i$, but is internally reported later at~$Z_i$. Consequently, compensations need to carried out by the bank to the affected side. For the loss~$i$, the compensations $X_{i,k}$'s are going to be paid at the times $U_{i,k}$'s. The bank is then required (e.g., by the Third Basel Accord) to quantify the future distribution of losses (measured through their compensations) belonging to operational risk.

\paragraph{Case~2: War damages}
Modeling of the evolution of national and international conflicts is a~long standing strand of research~\citep{GleditschMetternichRuggeri2014}, although data collection is a~very challenging task, see \cite{Arnold2019} and \cite{Cressey2008}. Based on the comprehensive databases as COPDAB \citep{Azar1980}, MIDLOC \citep{Braithwaite2010}, or PRIO \citep{Hallberg2012}, main approaches still remain to be classic econometric linear ones with a list of exogenous factors or those based on the hazard models, cf.~\cite{CollierHoefflerSoderbom2004}, \cite{Schrodt2014}, \cite{Clauset2014} and \cite{BakkeGreenhillWard2010}, \cite{HarrisonWolf2012}. As the list of current approaches was criticized by \cite{Schrodt2014}, namely ``garbage can models that ignore the effect of collinearity'', ``complex models without understanding the underlying assumptions'' or ``linear statistical monoculture'', we believe that our model can make a useful step forward in the conflict prediction. Using the proposed model, one considers each point~$T_i$ in the~$M$ process as the beginning of the tension between regions/countries. The followed up~$Z_i$ is the official beginning of the conflict, through the official notice or the first armed intrusion. This point (as the mark) starts a~\emph{spread the armed conflict} over a~series of battles~$k$ at time points~$U_{i,k}$ that take~$X_{i,k}$ lifes. The main assumptions of the process are fully in-line with the nature of the war: \ref{ass:Zs} ``there will always be another war'' and \ref{ass:Us} ``each war has an end''.

\paragraph{Case~3: Epidemics}
Proper modeling and prediction of the spread of epidemic is a~very important strand of literature, in particular in the view of recent H1N1 and Ebola epidemics. Classical models arising from modeling the online diffusions are those based on Susceptible-Infected-Recovered model introduced by \cite{Kermack27}, \cite{LindaAllen2008}, often extended to the stochastic case as in \cite{bobashev2007}, \cite{PingYan2008}, and others. Newly, these models were linked to the Hawkes processes in which spread of the disease in one population has been investigated, see \cite{SIRHawkes2018}. Considering \emph{several populations} (neighbor regions, countries, flight connections, etc.), the proposed model is a~natural flexible extension, in which each point~$T_i$ of the process is the infection of an~individual in the region~$i$. The delay~$W_i$ is thus the incubation period, after which the process of infection individuals in the population~$i$ starts, with individuals being infected at the points~$U_{i,k}$. The values~$X_{i,k}$ may be considered as the severity of the illness also converted to monetary quantities.

\paragraph{Case~4: Drug prescription}
Health care expenditures have become one of the most serious issues of the modern society and \emph{prescribed medicaments} seem to form one of the fastest growing component of the health insurance expenses. Managed care organizations encourage physicians to be more cost-conscious, see, e.g., \cite{miller}. They use financial incentives to induce physicians to reduce expenses while maintaining the quality of medical care. \emph{General practitioners (GP)} comprise a~significant part of the health care system and influence importantly the total amount of insurance money spent. Therefore, the \emph{prescription behaviour} of GPs is of utmost importance. It has been studied from the point of view of the pharmaceutical firms in several studies, see~\cite{gonul}, \cite{manchanda}, and other references therein; or from the point of view of the health insurance companies, e.g., \cite{HPH2017}. Furthermore, the prescription patterns and the influencing factors have been analyzed in~\cite{ekedahl}, \cite{rokstad}, \cite{watkins}, or~\cite{Caldbick2015}. One may think of a~spread of disease/illness as the occurrence time $T_i$ of event and, correspondingly, a~visit to the GP as the reporting time $Z_i$. Expenses for the prescribed drugs---sometimes more than one medical examination by the GP is needed (at times $U_{i,k}$)---are then the event payments $X_{i,k}$. After all, the responsible health care financing organization is interested to know the future expenditures for prescribed medicaments by the GPs within a~predetermined time horizon.

\paragraph{Case~5: Startups}
A~recent entrepreneurship bloom prompts for another straightforward application of the proposed micro forecasting method. Many well-known multinational companies leading the global market these days have begun their business in terms of small and locally based startups with only very limited human, social, and financial capital. These are, although, crucial factors for the future startup performance and its ability to survive~\citep{BPTW2004}. Especially the last one---the financial capital---turns out to play the most significant role for establishing an~entrepreneur on the global market~\citep{CG2008}. The initial financial capital is, however, usually not sufficient to start operations at the desired scale as the credit constraints for bank loans are too strict~\citep{HJR1994}. Therefore, the startup team looks for additional external sources of equity capital (such as external support, collaboration, fund raising, etc.), which is credited later over time. When modeling the overall impact of the startup in terms of its frontier production, this additional capital should be also considered~\citep{AKS1977}. Thus, the \emph{future cash flow} is of the main interest. Using our terminology, a~new entrepreneur~$i$ starts with its business after some waiting time~$W_i$ and additional capital amounts~$X_{i,k}$'s arrive at the times~$U_{i,k}$'s. The whole scenario can be analogously also adapted, for instance, for modeling the human capital, social investments, or business performance of the startups.

%%%%%abc ... Each mark of M is the start of new branch (AR games, VR games, Cryptocurrency, etc.) Each mark of the N will be the starting capital (or personnel, etc) of each founded company/start-up.

\paragraph{Case~6: Advertising and commercials}
In recent years, television and internet advertising have become increasingly tailored to individuals. Television commercials simply rely on so-called contextual advertising, where ads are chosen based upon the broadcast contents \citep{LTW2015}. More sophisticated ad placement techniques use online behavioral advertising or targeting advertising, which are typical for internet commercials and are based on the browsing history, online activities, or web searches \citep{GT2011}. \emph{Shopping behavior and shopping patterns} have recently been analyzed by economists in order to better understand the process through which consumers search for their preferred options \citep{BURDA2012,XIAO2018}. Let us concentrate on the perspective of a~company selling a~specific product over the Internet. A~moment, when some ad is placed over the Internet by another company running some website, can be thought of the occurrence time~$T_i$. The owner of the website is directly or indirectly paid for the advertisement by the product-selling company. Then, the reporting time~$Z_i$ naturally corresponds to the time when the ad is displayed or when it is recognized by some user (for example, the first click on it). If the user proceeds with a~purchase in the online store advertised by the ad, the payment time and the payment amount represent the mark of an~event development process. The product-selling company can consequently predict their future income based on the advertising and, thus, judge the efficiency of their commercial product placement.

\paragraph{Case~7: Apple Pay}
%Digital payment platforms make use of boundary resources to be highly integrative or integratable, which supports the intended conjoint commercialization efforts
Digital payment platforms are multi-sided and layered modular artifacts that primarily mediate payment transactions between payers and payees~\citep{Kazan2015}. Apple Pay as a~payment method has become increasingly popular in everyday life during recent years. As remarked by~\cite{LIU2019268}, it simply boosts satisfaction through elevated coolness in a~successful encounter. \cite{LIU2015372} examined recent changes in the payment sector in financial services, specifically related to mobile payments that enable new channels for consumer payments for goods and services purchases, and other forms of economic exchange. Although, Apple Pay serves merely as a~proxy and mediator between cardholders' (card issuer) and merchants' (acquirer) bank accounts, a~card issuer (e.g., a~bank) is obliged to pay a~portion from the payment for such a~service. Therefore, from the \emph{card issuer's perspective}, it is of interest to determine the amount of charges for the Apple Pay service during the future time period. Here, the date of issuing a~payment card can be considered as the occurrence time $T_i$ of event, the date of registering the Apple Pay is the reporting time $Z_i$, and the corresponding purchases are the payment amounts $X_{i,k}$ realized at payment times $U_{i,k}$.

%%%%%'claim' = banks client sets up Apple Pay and starts to pay = 'claim payments') or something similar like Google Pay, credit card from Amazon or registration in a huge eshop (registration = claim, purchase = claim payment

\paragraph{Case~8: Actuarial claims reserving}
An~\emph{insurance company} needs to predict future claims with corresponding payments and, additionally, future payments coming from already occurred claims, which do not have to be necessarily reported, and this is due to the current regulatory framework for insurance supervision (e.g., Solvency~II Directive or Swiss Solvency Test). A~\emph{lifetime} of a~claim can be characterized by the following variables that are driving the claim process: A~claim~$i$ (i.e., a~loss) occurs at the \emph{accident time}~$T_i$, however the insurance company is notified with some delay (heavy injuries that did not allow insured person report the claim, too light damages of the vehicle that allow an~insured person to postpone a~report, etc.) at the \emph{reporting time}~$Z_i$; the corresponding claim payments $X_{i,k}$'s are going to be paid by the insurance company to the insured at the payment times $U_{i,k}$'s. Finally, the distribution of cumulative payments within a~predetermined time window has to be predicted in order to settle the required claim reserves (e.g., Value at Risk or Expected Shortfall at $99.5\%$). Below, we concentrate in more details on the claims reserving task and \emph{exemplify the proposed methodology} through analyses of two insurance lines of business in order to demonstrate practical efficiency of our prediction method. % Although, one should bear in mind that the whole modeling approach together with the underlying stochastic procedures can be easily transferred to other areas.

\subsection{Outline}
This paper is structured as follow: Next section introduces the data and main stochastic objects we intend to model. Section~\ref{sec:theory} provides the assumptions and theory for occurrence and reporting times; number of payments; reporting delays; and payment amounts in different subsections. Section~\ref{sec:practical} contains a~practical application to the actuarial data. Afterwards, our conclusion follows. The proofs of our theoretical results are put in the Appendix.

\section{Granular loss reserving}\label{sec:granular}
A~classical actuarial problem called claims reserving is elaborated from an~emerging perspective. In contrast to the traditional claims reserving techniques based on aggregated information from historical data, our approach relies on granular individual claim-by-claim data and contributes to increase in the prediction's precision.

\emph{Claims (loss) reserving} in insurance determines a~sufficient amount of money, that needs to be put aside from the premium, to cover future claim (loss) payments. The main issue is to \emph{estimate/predict} these \emph{claims reserves}, which should be held by the insurer in order to meet all future claims arising from policies currently in force and policies written in the past. Claims reserving is a~classical problem in \emph{non-life insurance}, sometimes also called general insurance (in UK) or property and casual insurance (in USA). A~non-life insurance policy is a~contract between the insurer and the insured. The insurer receives a~deterministic amount of money, known as premium, from the insured in order to obtain a~financial coverage against well-specified randomly occurred events. If such an event (claim) occurs, the insurer is obliged to pay in respect of the claim a~claim amount, also known as a~loss amount. In layman's terms, if an~accident happens to an~insured person, he or she goes to the insurance company to request a~claim payment. The insurance company pays this claim amount from the loss reserves. In many cases several payments are performed for a~single accident, for instance, further health problems, hidden damages of the car that were not visible by the first inspection, etc.

Claims reserving methods based on \emph{aggregated data} from so-called run-off triangles are predominantly used to calculate the claims reserves, see~\cite{EV2002} or~\cite{wutrich_kniha} for an overview. Such models are not based on the particular claims or accidents, but rather on the aggregated overall payments through some predefined period, typically one year. These conventional reserving techniques have series of disadvantages: loss of information from the policy and the claim's development due to the aggregation; usually small number of observations in the aggregated data; only few observations for recent accident years; various assumptions of independence, which can sometimes be unrealistic or at least questionable; and sensitivity to the most recent paid claims, see \cite{PH2013} or \cite{PO2014} for some recent developments. In order to overcome the above mentioned deficiencies or imperfections, \emph{micro (granular) loss reserving methods for individual claim-by-claim data} need to be derived. Moreover, estimation of the whole distribution of the total future payments is a~crucial part of the risk valuation process. % Since there is a~growing demand for prediction of total reserves for different types of claims or even multiple lines of business, a~\emph{copula framework} for the granular reserving seems to be suitable to handle this multidimensional problem.

\subsection{Current status}
To estimate the distribution of reserves means to predict future cash flows and their uncertainty. On the top of that, this is becoming compulsory by the law due to the introduction of new supervisory guidelines. % (Solvency~II).

The loss reserving approaches based on individual/micro-level/granular/claim-by-claim data do not represent the mainstream in the reserving field. First attempts within the reserving framework of incorporating the claim information for reporting delays were using a~Bayesian approach~\citep{Jewell1989,Jewell1990} or an~empirical-Bayes approach~\citep{WTC1984}. Substantial branch of the individual loss reserving methods, that are based on a~position dependent marked Poisson process, involves work of \cite{Arjas}, \cite{Norberg,Norberg1999}, and \cite{HaastrupArjas}. A~Markov model for granular loss reserving was proposed by~\cite{Hesselager1994}. Including claims features to specify the model components within the setup of the marked point processes was revisited by~\cite{Larsen}. Empirical investigation by~\cite{AntonioPlat} indicates that individual reserving provides a~better accuracy compared to some selected aggregated models. A~discrete time formulation instead of the continuous time point process description was suggested by~\cite{GodecharleAntonio} and~\cite{PigeonAntonioDenuit}. Besides that, \cite{ZhaoZhouWang} and \cite{ZhaoZhou} proposed semiparametric techniques from survival analysis. Several case studies of the individual reserving approaches can be found in~\cite{TaylorMcGuireSullivan}. Machine learning techniques in the individual claims reserving were elaborated by~\cite{W2016}. Furthermore, \cite{VW2016} pointed out a~gain in using the individual methods by employing non-stationarity. Cox processes were utilized by~\cite{BLT2016}.

%%%%%%%%\crossoutred{Evaluation of the reserving risk lies on the cutting edge in actuarial science. Regardless of its importance, practical reserving techniques often forfeit diagnostics of the theoretical models' assumptions. }\cite{PH2012} \crossoutred{derived sufficient and necessary conditions for appropriateness of the widely used chain-ladder reserving method. Moreover, the majority of the classical reserving approaches are based on various assumptions of independence, which can sometimes be unrealistic or at least questionable} \citep{PH2013}\crossoutred{. Another relaxation of some restrictive assumptions in the traditional reserving methods---independent claim amounts and large number of parameters depending on the number of observations---was performed by~}\cite{PO2014}\crossoutred{. Overcoming previously mentioned pitfalls contributes to the increase in precision of the reserves' prediction (e.g., smaller variability, lighter tails) and allows straightforward claims forecasting beyond the last observed development period.}

Practical loss prediction techniques often forfeit diagnostics of the theoretical models' assumptions. This is overcome by our approach, where we deal with nonlinear continuous time Markov environments~\citep{HS2009}. Generally, times of events together with accompanying measures can be analyzed as marked point processes, e.g., in case of ultra-high-frequency data, see~\cite{Engle2000}. Stochastic methods for modeling the total claim amount via marked Poisson cluster models have been recently proposed by~\cite{BWZ2018}. Here, the marks can take values only in a~finite-dimensional space and have a~common distribution. Our approach allows for point processes as marks, which is very suitable for a~practical application of unrestricted number of payments, and enables time-varying distributions for, e.g., payment dates, reporting delays, or payment amounts.

\subsection{Main goals}
This paper contributes to the literature by aiming at using \emph{all} the available information in the data. The proposed model motivated by the data controls for dependencies (between different payment amounts, between payments amounts and reporting delays, between reporting delay and accident date, etc.) in a~simple and natural way. We assume a~marked non-stationary Poisson process for the time ordered reporting dates; flexibly parametrized conditional distribution of the reporting delay and payment sizes given the accident date; or a~non-homogeneous Poisson process in the role of a~process' mark for the number of payments. All the models in this paper are supported with the asymptotic theory, and are combined in an~\emph{omnibus model}. Application to the true data strongly outperforms classical models in both point and interval forecast of the reserved losses. Up to the best of our knowledge, this is the first time where all the possible cross and temporal dependencies of the claim data are taken into account.

\subsection{Data}\label{sec:data}
Nowadays, modern databases and computer facilities provide a~foundation for loss reserving based on individual data. There is no more reason to rely on the reserving techniques using aggregated data only. We possess the unique database from the Guarantee Fund of the Czech Insurer's Bureau for car insurance which consist of claims developments from the beginning of 2004 up to the end of 2016. Each record in the data set contains:
\begin{itemize}
\item Claim ID (if one claim is associated with more payments, each payment has a~separate entry);
\item Type of claim, which can be either \emph{bodily injury} or \emph{material damage};
\item Accident time (occurrence);
\item Reporting time (notification);
\item Date of payment, when the payment is credited to the client's bank account;
\item Amount of payment.
\end{itemize}
%In return for providing the data, the Czech Insurers' Bureau (\v{C}KP) will certainly benefit from our results. For example, the member insurance companies of the bureau will acquire more precise methodology for loss reserving pre-tested on case studies.
%
%This dataset has been cleaned for typos (f.e. reporting date prior to accident date) as well as payments of 2500 CZK or 2783 CZK have been removed, as they correspond to internal movements or first lump-sum payment to the owner of the non Czech vehicle. For the current study we restrict ourselves to claims being reported at years 2004-2016 as a very big portion of the claims happened in 2017 are not yet reported nor credited.
All in all we have 4450 claims comprising of 10820 payments for bodily injuries and 30545 claims distributed into 35642 payments for material injuries within the investigated time interval. For \emph{back-testing} purposes, we only use the data up to the end of 2015 to construct the prediction. The data from 2016 are only employed for comparison purposes with the obtained results. %Having this dataset at hand, let us define main variables that are driving the claim process.

%Because of the reporting delays and possibly more than one claim payment per claim, the claims cannot be settled immediately. Therefore, insurance companies need to hold reserves to cover these losses.

%We generally distinguish between two types of claims reserves. \emph{RBNS} stands for Reported But Not Settled claims. These are the claims that had occurred and were reported before the present moment, but their settlement would occur in the future. Hence in Figure~\ref{fig:claim-dev}, the present moment is somewhere between $t_2$ and $t_6$ or $t_{10}$. \emph{IBNR} stands for Incurred But Not Reported claims. These are the claims that occurred before the present moment, but will be reported in the future. Therefore, the present moment is somewhere between $t_1$ and $t_2$.

%%%%%%%%Let us remind that all dates (accident, reporting, and payment) are recorded in a~format year-month-day-hour-minute-second. Therefore, these dates (times) are considered to have some \emph{continuous distribution}.

%\subsection{Structure of the paper}
%a

\section{Theoretical framework}\label{sec:theory}
The size of the claims reserves protects insurance company against the future losses. Bearing this practical issue in mind, we propose a~series of theoretical models describing each component of the claim's chain. We assume that the reporting dates $Z_i$'s follow a~non-homogeneous Poisson process with a~parametric intensity function; the reporting delays $W_i$'s follows a~time-varying continuous parametric distribution conditional on the reporting dates; the payment dates for each claim~$i$ are represented by arrival times of a~non-homogeneous Poisson process $N_i(t)$, where $N_i$ is a~mark of $M$; and the payment amounts $X_{i,j}$'s are modeled similarly to the reporting delays via a~time-varying parametric conditional distribution. All the models are then brought together under one umbrella in the empirical study.

Recently, \cite{GIESECKE2018} have discussed marked point processes with applications in finance and economics to model the timing of defaults, corporate bankruptcies, market transactions, unemployment spells, births, and mortgage delinquencies. They developed likelihood estimators for the parameters of a~marked point process and incompletely observed explanatory factors that influence the arrival intensity and mark distribution, although they presumed only finite dimensional marks. We go beyond therein investigated models by assuming infinite dimensional marks via different stochastic framework. We establish an approximation to the likelihood and analyze the convergence and large-sample properties of the associated estimators. Numerical results illustrate the behavior of our estimators.

Recall that our \emph{primary practical goal} is to model and, consequently, to simulate a~distribution of the sum of future payments within the time period $(a,b]$. Besides that, the \emph{secondary practical goal} is to back-predict claims that have already occurred, but are still not reported. As a~theoretical by-product, we investigate marked non-homogeneous Poisson processes with infinite dimensional marks, which has not been done yet.

\subsection{Occurrence and reporting times}
As for the insurance company it is not so relevant, when the accident has happened, but rather when it has been reported, as on that date the whole procedure of the claim payments starts. 

We proceed to the assumptions on the reporting dates $Z_i$'s that are needed for showing existence, uniqueness, and asymptotic properties of the proposed estimators. Taking into account, that different claims are supposed to be unrelated, it is natural to assume that the time differences $\{Z_i-Z_{i-1}\}_{i\in\mathbb{N}}$ are independent ($Z_0\equiv 0$). Thus, the reporting dates $Z_i$'s can be viewed as the arrival times of a~counting process with independent increments. A~reasonable and parsimonious representative would be a~non-homogeneous Poisson process.

\begin{assumpM}\label{ass:Zs}
The time ordered reporting times $\{Z_i\}_{i\in\mathbb{N}}$ are arrival times of a~non-homoge\-neous Poisson process $\{M(t)\}_{t\geq 0}$ with a~parametric intensity $\psi(t;\brho)>0$ such that $M(t)=\sum_{i=1}^{\infty}\mathbbm{1}\{Z_i\leq t\}$, $\brho\in{\boldsymbol R}\subseteq\mathbb{R}^q$, and ${\boldsymbol R}$ is an~open convex set.
\end{assumpM}

The reporting epochs $\{Z_i\}_{i\in\mathbb{N}}$ are reversely determined by the counts $\{M(t)\}_{t\geq 0}$ such that $Z_i=\inf_{t\geq 0}\{M(t)\geq i\}$. The intensity $\psi(t;\brho)$ can be considered as a~\emph{risk exposure} for an~accident reporting (not occurring) in time~$t$. Although, one may still argue that it should be more convenient to assume that the occurrence (accident) times, and not the reporting times, should form the arrival times of some non-homogeneous Poisson process. This is indeed in concordance with the parametric time-varying conditional density~$f_W$ for the reporting delay~$W_i$ (given $Z_i=z$) defined later on, which results in the fact that $\{T_i\}_{i\in\mathbb{N}}$ are arrival times of another non-homogeneous Poisson process having the intensity
\begin{equation}\label{eq:displacement}
\mu(t;\brho,\bvartheta)=\int_{\mathbb{R}}\psi(z;\brho)f_{W}\{t;w(z,\bvartheta)\}\ud z,
\end{equation}
because of the \emph{displacement theorem} \citep[p.~61]{Kingman1993}. Thus, $\mu(t;\brho,\bvartheta)$ is just a~risk exposure for an~accident occurring in time~$t$.

We should emphasize that similarly to the below derived statistical inference for the non-homogeneous Poisson process, many other authors dealt with consistent estimation of the process intensity. To mention at least some of them, we refer to~\cite{Konecny1987}, \cite{SCHOENBERG2005}, \cite{Waag2007}, \cite{WG2009}, \cite{coeurjolly2014}, and~\cite{PDV2017}, although sometimes in a~more general setup. There are two main reasons why we derive consistency and asymptotic normality of the intensity estimator in a~different fashion: First, it is of a~practical interest to require simple assumptions, which are easily verifiable and allowing for a~huge class of parametric intensities. Second, our theoretical results and ways of proving them serve as an~intermediate product for developing a~suitable statistical inference for the marked non-homogeneous Poisson process with marks being non-homogeneous Poisson processes (discussed in Subsection~\ref{subsec:NumberOfPayments}).

Since we consider a~fully parametric approach, it is firstly necessary to estimate the unknown parameter~$\brho$. We employ the maximum likelihood (ML) approach for the arrival times. The unconditional likelihood in case of $M(t)$, when the last observable (deterministic) time is~$t$, has the form
\begin{equation}\label{eq:likelihood}
L(\brho;Z_1,\ldots,Z_{M(t)},t)=\exp\{-\Psi(t;\brho)\}\prod_{i=1}^{M(t)} \psi(Z_i;\brho),\quad 0<Z_1<\ldots<Z_{M(t)}<t,
\end{equation}
where $\Psi(t;\brho):=\int_{0}^t\psi(z;\brho)\ud z=\E M(t)$ is a~\emph{cumulative intensity} function.

%the logarithm of the likelihood function~$L$ provides the maximum likelihood estimator $\widehat{\brho}$. 
Maximizing the log-likelihood function
\begin{equation}\label{eq:llWZ}
\ell(\brho;\bZ,t)=\sum_{i=1}^{M(t)}\log\psi(Z_i;\brho)-\int_0^t\psi(z;\brho)\ud z
\end{equation}
with respect to $\brho$ provides an~ML estimator $\widehat{\brho}$. The true value $\brho_0\in{\boldsymbol R}$ of the unknown parameter $\brho$ is supposed to uniquely maximize $\E_{\brho}\ell(\brho;\bZ,t)$. The uniqueness of~$\brho_0$ is essential for identifiability and, consequently, for consistency and asymptotic normality of the ML estimator~\citep{White1982}. These properties are proved later on. Although, one has to realize that we are dealing with \emph{not independent and not identically distributed} (n.i.n.i.d.) random variables (i.e., arrival times). % These two desired properties are proved by copying the arguments of \citet[Theorem~2.1]{HjortPollard2011} mutatis mutandis, where the asymptotic results provided by Theorem~2.1 from~\cite{HjortPollard2011} are for the i.i.d.~random variables.
Let us define
\[
h(Z_i;\brho,t):=\frac{1}{M(t)}\Psi(t;\brho)-\log \psi(Z_i;\brho),
\]
which means that
\begin{equation}\label{eq:MLErho}
\widehat{\brho}=\arg\min_{\brho\in{\boldsymbol R}} \sum_{i=1}^{M(t)} h(Z_i;\brho,t).
\end{equation}
Assume the following to hold with respect to the $h(Z_i;\brho,t)$ functions in order to obtain a~sensible estimator (i.e., consistent and asymptotically normal).

\begin{assumpM}\label{ass:convex}
$h(z;\brho,t)$ is convex in $\brho\in{\boldsymbol R}$ for all $0<z<t$.
\end{assumpM}

The convexity of~$h$ from Assumption~\ref{ass:convex} contributes to assurance that there exists a~unique optimum, namely some function $\widehat{\brho}\equiv\widehat{\brho}(t)$.

For simplicity of further notations, let us denote $[\cdot][\cdot]^{\top}\equiv[\cdot]^{\otimes 2}$, $\partial_{\brho_0}\equiv\frac{\partial}{\partial\brho}\left[\cdot\right]_{\brho=\brho_0}$, $\partial^2_{\brho_0}\equiv\frac{\partial^2}{\partial\brho\partial\brho^{\top}}\left[\cdot\right]_{\brho=\brho_0}$, $\partial_{\brho_0,i}\equiv\frac{\partial}{\partial\rho_i}\left[\cdot\right]_{\brho=\brho_0}$, and $\partial^2_{\brho_0,i,j}\equiv\frac{\partial^2}{\partial\rho_i\partial\rho_j}\left[\cdot\right]_{\brho=\brho_0}$ for $\brho=[\rho_1,\ldots,\rho_q]^{\top}$. Symbol~$\zero$ stands for a~zero vector and ${\boldsymbol I}$ means an~identity matrix with a~suitable dimension. Furthermore for $t>0$, let us define an~information matrix $\mathcal{I}(t;\brho_0):=\int_0^t\frac{\{\partial_{\brho_0}\psi(z;\brho)\}^{\otimes 2}}{\psi(z;\brho_0)}\ud z$ and a~matrix $\mathcal{K}(t,\brho_0):=\frac{1}{\sqrt{\psi(t;\brho_0)}}\left[\partial^2_{\brho_0}\psi(t;\brho)-\frac{\{\partial_{\brho_0}\psi(t;\brho)\}^{\otimes 2}}{\psi(t;\brho_0)}\right]$.

\begin{assumpM}\label{ass:interchange1}
$\frac{\partial^2}{\partial\brho\partial\brho^{\top}}\psi(\cdot;\brho):\,\mathbb{R}^+\to\mathbb{R}^{q\times q}$ is continuous for all~$\brho\in{\boldsymbol R}$ and there exist Lebesgue-integrable functions $m_{1,i}$ and $m_{2,i,j}$ such that
\[
\left|\frac{\partial}{\partial\rho_i}\psi(t;\brho)\right|\leq m_{1,i}(t)\quad\mbox{and}\quad\left|\frac{\partial^2}{\partial\rho_i\partial\rho_j}\psi(t;\brho)\right|\leq m_{2,i,j}(t)
\]
for all $\brho\in{\boldsymbol R}$, almost every $t>0$, and $i,j=1,\ldots,q$.
\end{assumpM}

The definition of $\mathcal{I}(t;\brho_0)$, the uniqueness of the true value $\brho_0\in{\boldsymbol R}$, and the differentiability of $\psi(t;\cdot)$ from Assumption~\ref{ass:interchange1} ensure that $\mathcal{I}(t;\brho_0)$ is positive definite. Moreover, the integrable majorants from Assumption~\ref{ass:interchange1} together with the smoothness of~$\psi$ allow to interchange the integral and the derivative of~$\psi$.

\begin{assumpM}\label{ass:infmat}
%There exists a~symmetric non-random matrix function $\mathcal{K}(\cdot,\cdot):\,\mathbb{R}_0^+\times\mathbb{R}^q\to\mathbb{R}^{q\times q}$ such that,
As $t\to\infty$,
\begin{enumerate}
\item $M(t)\mathcal{I}^{-1}(t,\brho_0)$ converges in probability to a~positive semidefinite matrix;
%\item $\mathcal{K}^{-1/2}(t,\brho_0)\mathcal{J}(t;\brho_0)\mathcal{K}^{-1/2}(t,\brho_0)\to\mathcal{I}(\brho_0)$, which is positive definite;
\item $\int_0^{t}\left\{\mathcal{I}^{-1/2}(t,\brho_0)\mathcal{K}(z,\brho_0)\mathcal{I}^{-1/2}(t,\brho_0)\right\}^2\ud z\to\zero$.
\end{enumerate}
\end{assumpM}

To check whether Assumption~\ref{ass:infmat}~(i) holds, one needs, for instance, to verify the following two relations: $\mathcal{I}^{-1}(t,\brho_0)\E M(t)=\mathcal{I}^{-1}(t,\brho_0)\Psi(t;\brho_0)$ converges to a~positive semidefinite matrix, which may also be a~zero matrix; and $\Var \big\{\left(\mathcal{I}^{-1}(t,\brho_0)\right)_{i,j}M(t)\big\}=\left(\mathcal{I}^{-1}(t,\brho_0)\right)_{i,j}^2\Psi(t;\brho_0)\to 0$ as $t\to\infty$ for all $i,j=1,\ldots,q$. By the Cauchy-Schwarz inequality and equations~\eqref{eq:jkelement}--\eqref{eq:EjlEkm} from the proof of the consequent Theorem~\ref{thm:consistency}, Assumption~\ref{ass:infmat}~(ii) is satisfied if for all $j,k,\ell,m=1,\ldots,q$ holds
\begin{align*}
&\left(\mathcal{I}^{-1/2}(t,\brho_0)\right)_{j,k}^2\left(\mathcal{I}^{-1/2}(t,\brho_0)\right)_{\ell,m}^2\\
&\ \times\int_0^{t}\frac{1}{\psi(z;\brho_0)}\left\{\partial_{\brho_0,k,\ell}^2\psi(z;\brho)-\frac{\partial_{\brho_0,k}\psi(z;\brho)\partial_{\brho_0,\ell}\psi(z;\brho)}{\psi(z;\brho_0)}\right\}^2\ud z\to 0,\quad t\to\infty.
\end{align*}
%%%%%%%%{\color{red}IS IN THE PREVIOUS FORMULA MEANT THE FOLLOWING?}
%%%%%%%%{\color{red}\begin{align*}
%%%%%%%%&\left(\mathcal{I}^{-1/2}(t,\brho_0)\right)_{j,k}^2\left(\mathcal{I}^{-1/2}(t,\brho_0)\right)_{\ell,m}^2\\
%%%%%%%%&\quad \times\int_0^{t}\frac{1}{\psi(z;\brho_0)}\left\{\partial_{\brho_0,k,\ell}^2\psi(z;\brho)-\frac{\partial_{\brho_0,k}\psi(z;\brho)\partial_{\brho_0,\ell}\psi(z;\brho)}{\psi(z;\brho_0)}\right\}^2\ud z \to 0, \  t\to\infty
%%%%%%%%\end{align*}}

At first sight, technical Assumptions~\ref{ass:interchange1} and~\ref{ass:infmat} have also a~practical impact on admissibility of the intensity function~$\psi$ and on the amount of information about the parameter~$\brho$ contained in the process~$M$. Basically, the intensity~$\psi$ has to be sufficiently smooth and adequately regular with respect to the information matrix function~$\mathcal{I}$. In practice, we do not allow for too `wild' and too quickly changing behavior of the process of reporting times.

\begin{theorem}[Consistency I]\label{thm:consistency}
Under Assumptions~\ref{ass:Zs}--\ref{ass:infmat},
\[
\mathcal{I}^{1/2}(t,\brho_0)\left(
\widehat{\brho}-\brho_0\right)=-\mathcal{I}^{-1/2}(t,\brho_0)\sum_{i=1}^{M(t)}\partial_{\brho_0}h\left(Z_i;\brho_0,t\right)+o_{\prob}(1),\quad t\to\infty.
\]
\end{theorem}

Let us discretize the `continuous' time $t\in\mathbb{R}_0^+$ for the process $\{M(t)\}_{t\geq 0}$ in a~way that one observes~$M$ only at all discrete time points $a\in\mathbb{N}$. This is indeed in concordance with the nature of our practical problem, where we evaluate the number of reported claims at the end of the calendar year represented by a~discrete value of~$a$.

Additionally, the next \emph{Lindeberg condition} can extend the assertion of Theorem~\ref{thm:consistency}.
\begin{assumpM}\label{ass:Lind}
$\lim_{a\to\infty}\sum_{i=1}^{a}\E\left({\boldsymbol d}^{\top}\bY_i\right)^2\mathbbm{1}\{|{\boldsymbol d}^{\top}\bY_i|\geq \eps\|{\boldsymbol d}\|_2\}=0$ for all ${\boldsymbol d}\in\mathbb{R}^{q}$ and $\eps>0$, 
%\begin{multline*}
%\lim_{a\to\infty}\E\sum_{i=1}^{M(a)}\bigg[\left({\boldsymbol d}^{\top}\mathcal{I}^{-1/2}(a,\brho_0)\partial_{\brho_0}h\left(Z_i;\brho_0,a\right)\right)^2\bigg.\\
%\bigg.\mathbbm{1}\left\{\left|{\boldsymbol d}^{\top}\mathcal{I}^{-1/2}(a,\brho_0)\partial_{\brho_0}h\left(Z_i;\brho_0,a\right)\right|\geq\eps\right\}\bigg]=0,
%\end{multline*}
where $\bY_i:=\mathcal{I}^{-1/2}(a,\brho_0)\int_{i-1}^{i}\left\{\partial_{\brho_0}\log\psi(z;\brho)\right\}\left(\ud M(z)-\psi(z;\brho_0)\ud z\right)$.
\end{assumpM}

For practical verification purposes, one can assume a~version of the \emph{Lyapunov condition} instead of the Lindeberg one. For instance, for all ${\boldsymbol d}\in\mathbb{R}^{q}$, there exists $\delta>0$ such that $\lim_{a\to\infty}\sum_{i=1}^{a}\E\left|{\boldsymbol d}^{\top}\bY_i\right|^{2+\delta}=0$. On one hand, the Lyapunov condition is more restrictive than the Lindeberg one, on the other hand, it is easier to verify.

\begin{corollary}[Asymptotic normality I]\label{cor:asnorm}
Under Assumptions~\ref{ass:Zs}--\ref{ass:Lind},
\[
\mathcal{I}^{1/2}(a,\brho_0)\left(
\widehat{\brho}-\brho_0\right)\xrightarrow[a\to\infty]{\dist}\mathsf{N}_{q}\left(\zero,{\boldsymbol I}\right).
\]
\end{corollary}

Let us consider cases of the constant and exponential intensity function.
\begin{example}
Intensity $\psi(z;\brho)=\rho$, which corresponds to a~homogeneous Poisson process. Here, ${\boldsymbol R}=(0,\infty)$ and the cumulative intensity is $\Psi(t,\rho)=\rho t$. The log-likelihood function is $\ell(\rho;\bZ,t)=M(t)\log\rho-\rho t$. The ML estimator is $\hat{\rho}=M(t)/t$, the information number becomes $\mathcal{I}(t;\rho_0)=t/\rho_0$, and $\mathcal{K}(t;\rho_0)=-\rho_0^{-3/2}$. Assumption~\ref{ass:infmat} is easily verifiable. For the Lyapunov condition, the choice of $\delta=1$ leads to $\sum_{i=1}^{a}\left(\frac{\rho_0}{a}\right)^{3/2}\E\left|\int_{i-1}^i\frac{1}{\rho_0}\ud M(z)-\int_{i-1}^i\ud z\right|^3=\sqrt{\frac{\rho_0^{3}}{a}}\E\left|\frac{Z_{i}-Z_{i-1}}{\rho_0}-1\right|^3=\sqrt{\frac{\rho_0^{3}}{a}}(12\mathrm{e}^{-1}-2)\to 0$ as $a\to\infty$, because the interarrival times $\{Z_i-Z_{i-1}\}_i$ have the exponential distribution with parameter~$1/\rho_0$ (i.e., its expectation equals $\rho_0$). Hence, $\sqrt{a/\rho_0}\left(
\widehat{\rho}-\rho_0\right)\xrightarrow[a\to\infty]{\dist}\mathsf{N}\left(0,1\right)$.
\end{example}

%One of the simplest and practically useful examples satisfying the above stated assumptions is the following.
\begin{example}
Intensity $\psi(z;\brho)=\exp\{\rho_1+\rho_2z\}$. Here, ${\boldsymbol R}=(0,\infty)\times (0,\infty)$. The log-likelihood function is $\ell(\brho;\bZ,t)=\rho_1M(t)+\rho_2\sum_{i=1}^{M(t)}Z_i-\mathrm{e}^{\rho_1}\left(\mathrm{e}^{\rho_2 t}-1\right)/\rho_2$. The ML estimator of the parameter $\rho_2$ can be obtained as a~solution of $\sum_{i=1}^{M(t)}Z_i+M(t)/\hat{\rho}_2-tM(t)/\left(1-\mathrm{e}^{-\hat{\rho}_2t}\right)=0$ and the ML estimator of the parameter $\rho_1$ comes from $\hat{\rho}_1=\log\left\{\hat{\rho}_2M(t)/\left(\mathrm{e}^{\hat{\rho}_2t}-1\right)\right\}$. Consequently,
\[
\mathcal{I}(t;\brho)=\left(\begin{matrix}
\mathrm{e}^{\rho_1} \left(\mathrm{e}^{\rho_2 t}-1\right)/\rho_2 & \mathrm{e}^{\rho_1} \left\{\mathrm{e}^{\rho_2 t} \left(\rho_2 t-1\right)+1\right\}/\rho_2^2\\
\mathrm{e}^{\rho_1} \left\{\mathrm{e}^{\rho_2 t} \left(\rho_2 t-1\right)+1\right\}/\rho_2^2 & \mathrm{e}^{\rho_1} \left[\mathrm{e}^{\rho_2 t} \left\{\rho_2 t \left(\rho_2 t-2\right)+2\right\}-2\right]/\rho_2^3
\end{matrix}\right),
\]
which can be easily proved to be positive definite for all $t>0$ and any $\brho\in{\boldsymbol R}$. Assumption~\ref{ass:infmat} together with the Lyapunov condition can be checked as well. Hence, $\mathcal{I}^{1/2}(a;\brho_0)\left(\widehat{\brho}-\brho_0\right)\xrightarrow[a\to\infty]{\dist}\mathsf{N}_2\left(\zero,{\boldsymbol I}\right)$.
\end{example}

An~example of the intensity function directly used in the consequent practical analysis of our data for modeling the reporting times of bodily injury claims, see also Figure~\ref{fig:Reporting} (left panel, green line), is given below.
\begin{example}\label{ex:Zbodily}
Intensity $\psi(z;\brho)=\exp\{\rho_1+\rho_2\log z+\rho_3\cos\left(2\pi z/\rho_5\right)+\rho_4\sin\left(2\pi z/\rho_5\right)\}$. The above formulated assumptions are satisfied for a~particular open convex ${\boldsymbol R}\subseteq\mathbb{R}^5$. The defined entities are not presented here due to their voluminous forms.
\end{example}

\begin{figure}[!ht]
	\begin{center}
		\includegraphics[width=0.5\textwidth]{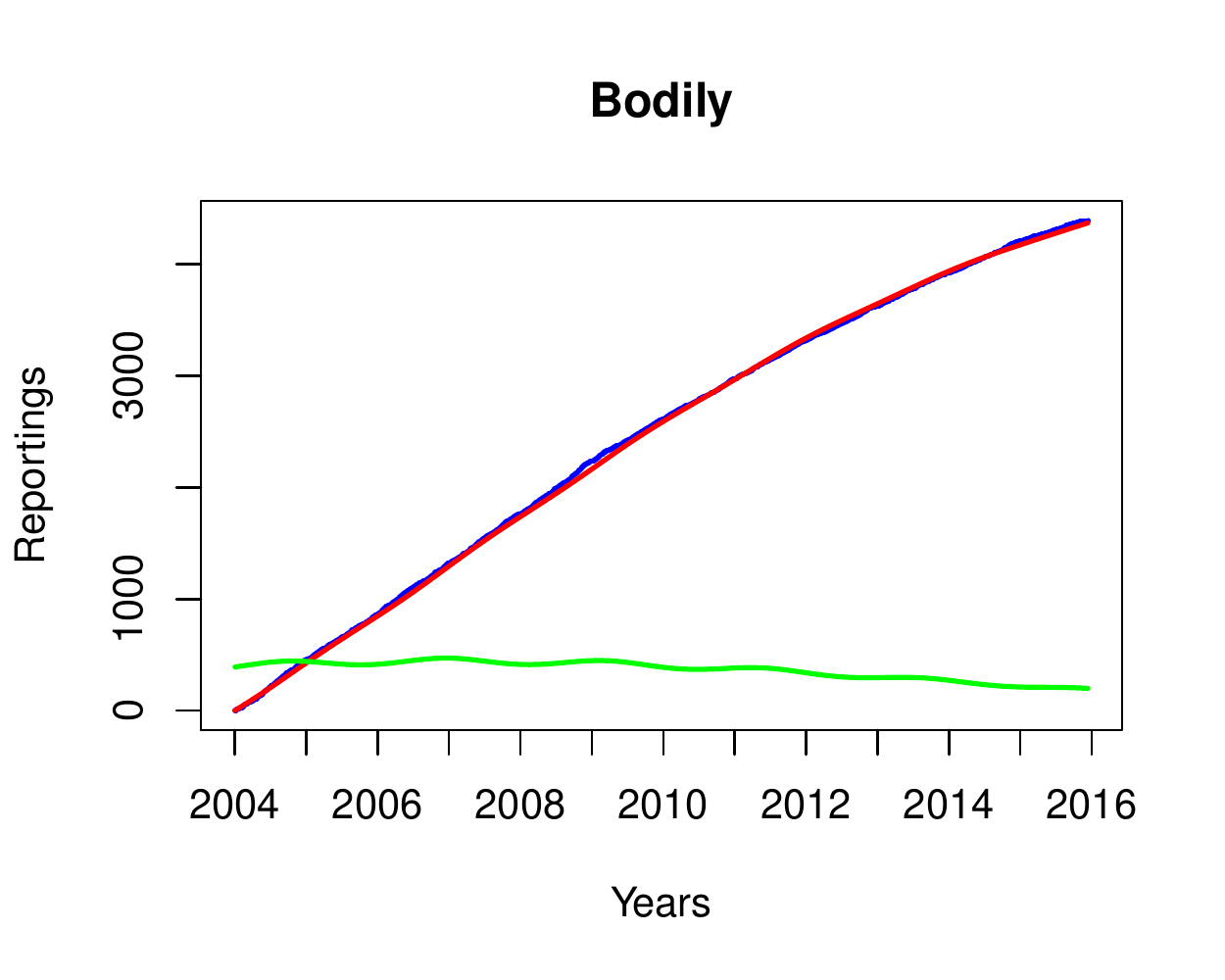}\includegraphics[width=0.5\textwidth]{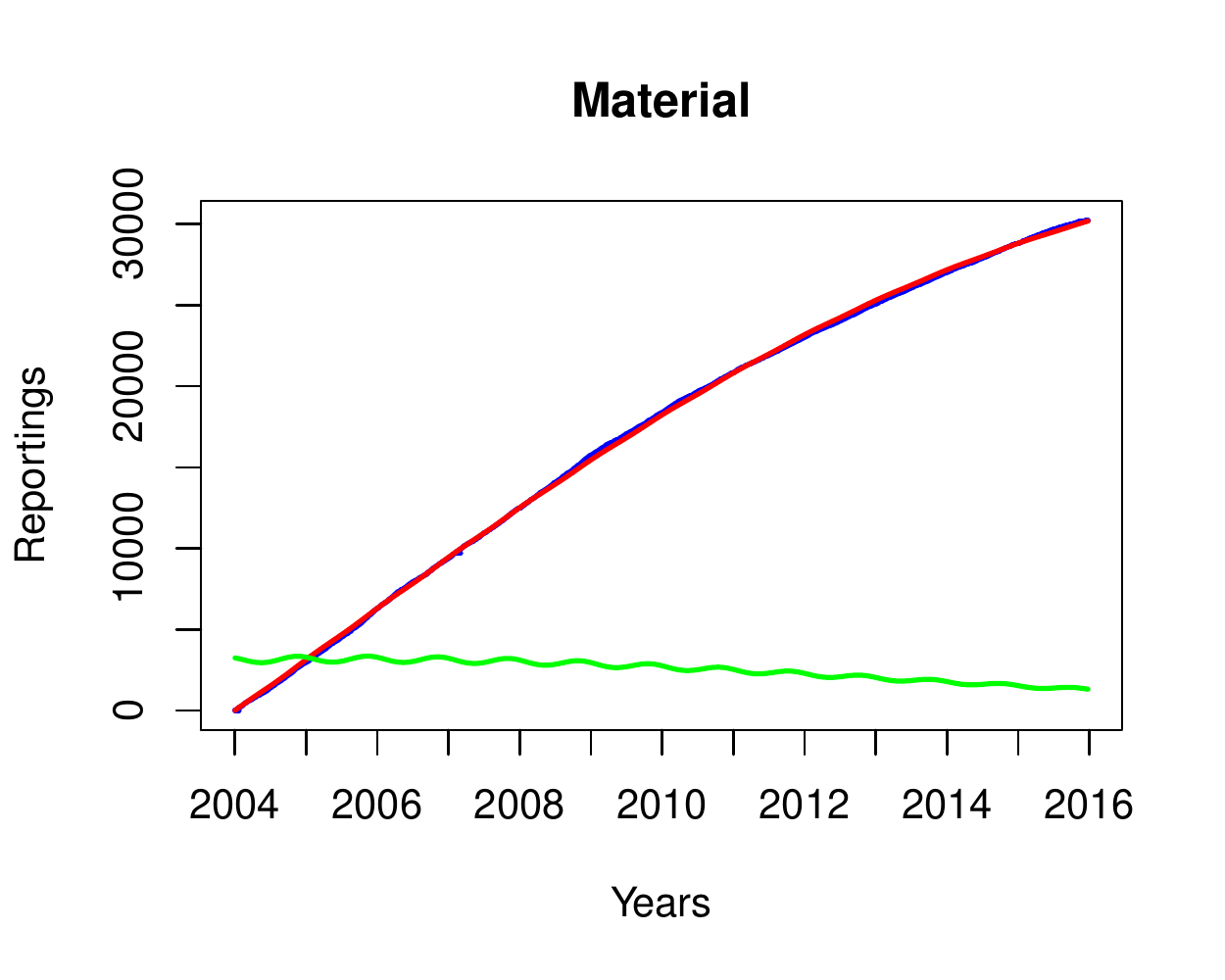}
		\caption{Number of reported claims---empirical (observed) cumulative intensity in blue, estimated cumulative intensity in red, and estimated intensity in green (prediction uses only data up to the end of 2015).}
		\label{fig:Reporting}
	\end{center}
\end{figure}

Besides that, the next example is used in the data analysis for the reporting times of material damage claims, cf.~Figure~\ref{fig:Reporting} (right panel, green line).
\begin{example}\label{ex:Zmaterial}
Intensity $\psi(z;\brho)=\exp\{\rho_1+\rho_2 z+\rho_3 z^2+\rho_4\cos\left(2\pi z/\rho_6\right)+\rho_5\sin\left(2\pi z/\rho_6\right)\}$. The required assumptions are again satisfied, but the above defined entities are not presented here due to their complicated and voluminous forms.
\end{example}

%%%%%confidence intervals ... maybe bootstrap

Suitability of Examples~\ref{ex:Zbodily} and~\ref{ex:Zmaterial} for the practical analysis is illustrated in Figure~\ref{fig:Reporting}, where the observed and fitted cumulative intensities (corresponding to the theoretical cumulative intensity $\Psi(t;\brho)$) are compared. The deviations between them are minor. Let us recall that our estimation of the reporting dates' intensity is based on the data up to the end of year 2015. Extrapolation of the estimated cumulative intensity for the `future' year 2016 also nicely mimics the known reality from year 2016 (not used for estimation). The estimated underlying intensity is depicted as well.

%%%%%\begin{figure}[!ht]
%%%%%\centering\includegraphics[width=0.49\textwidth]{Rep_back_Bod.pdf}
%%%%%\includegraphics[width=0.49\textwidth]{Rep_back_Mat.pdf}
%%%%%\caption{Quarterly averaged counts of reported claims -- observed in lighter color and predicted in the darker one (recall that our prediction uses only data up to the end of 2015).\label{fig:Rep_back}}
%%%%%\end{figure}

%%%%%Additionally, one can validate our estimation of the reporting dates' intensity based on the data up to the end of year 2015 by comparing its extrapolation for the `future' year 2016 with the known reality (not used for estimation). Figure~\ref{fig:ReportingDates} shows observed counts of the reporting dates for 2016 together with the extrapolated fitted intensity. The fitted intensity functions are not straight lines, despite it may look like, because they are just zoomed for year 2016. Nevertheless, both estimated intensities nicely mimic the observed counts or reported claims.

%%%%%\begin{figure}[!ht]
%%%%%\begin{center}
%%%%%\includegraphics[width=0.5\textwidth]{Reporting_dates_future_1y_Bod.pdf}\includegraphics[width=0.5\textwidth]{Reporting_dates_future_1y_Mat.pdf}
%%%%%\caption{Time histogram of the counts of reporting dates for the year 2016 with the fitted intensity function estimated based on data up to the end of 2015.}
%%%%%\label{fig:ReportingDates}
%%%%%\end{center}
%%%%%\end{figure}

%%%%%simple assumptions, easily verifiable, allowing for a~huge class of parametric intensities (not too restrictive assumptions at all)

Finally, it would be natural to characterize the number of new arriving claims with respect to the intensity of the process~$M$.

%%%%%\begin{assumpM}\label{ass:Minfty}
%%%%%$\lim_{t\to\infty}\Psi(t;\brho)=\infty$.
%%%%%\end{assumpM}

\begin{proposition}[Infinite number of renewals]\label{prop:infty}
If Assumption~\ref{ass:Zs} holds and $\lim_{t\to\infty}\Psi(t;\brho)=\infty$, then $\prob\left\{\lim_{t\to\infty}M(t)=\infty\right\}=1$.
\end{proposition}

This proposition reveals that the divergent cumulative intensity $\Psi(t;\brho)$ assures that there are still new claims being reported with probability one.

\subsection{Number of payments}\label{subsec:NumberOfPayments}
Let us recall that the number of payments corresponding to the $i$th claim till time point~$t$ is denoted by~$N_i(t)$. So, we possess panels of count processes $\{N_i(t)\}_{t>0}$ for $i=1,\ldots,M(t)$ that can be represented as $\{\bN(t)\}_{t>0}$. Thus, the claim notifications together with the claim payments can be viewed as a~marked Poisson process with Poisson processes as marks
\[
\{\{M(t)\}_{t\geq 0},\{\bN(t)\}_{t\geq 0}\}.
\]
In practice, we observe the counting process $\{N_i(t)\}_{t>0}$ through $\{U_{i,1},\ldots,U_{i,N_i(t)}\}$ for $i=1,\ldots,M(t)$, where the number of payments for the $i$th claim is denoted by $N_i(t)$ and $U_{i,k}$ is the time of the $k$th payment within the $i$th claim for $k=1,\ldots,N_i(t)$. Moreover, the amount of the $k$th payment for the $i$th claim paid at time $U_{i,k}$ is represented by $X_{i,k}$, which is going to be modeled in Subsection~\ref{subsec:payment_amounts}.

\begin{assumpN}\label{ass:Us}
The ordered payment times $\{U_{i,1},U_{i,2},\ldots\}$ of the $i$th loss are arrival times of a~non-homogeneous Poisson process $\{N_i(t)\}_{t\geq 0}$. Processes $\{N_i(t)\}_{t\geq 0}$, $i=1,2,\ldots$ are independent having parametric intensities $\lambda(t,Z_i;\bth)$ such that $N_i(t)=\sum_{k=1}^{\infty}\mathbbm{1}\{U_{i,k}\leq t\}$, $\bth\in{\boldsymbol P}\subseteq\mathbb{R}^p$, and ${\boldsymbol P}$ is an~open convex set.% The cumulative intensity $\Lambda_i(t;\bnu,\bBeta):=\int_{0}^t\lambda_i(v;\bnu,\bBeta)\ud v$ converges if $t\to\infty$.
\end{assumpN}

Since $U_{i,k}\geq Z_i$ (i.e., the payment times come after the reporting time), the corresponding density~$\lambda$ has to be constant zero up to the reporting date~$Z_i$. Alternatively, one can think of a~`restarted' process $\tilde{N}_{Z_i}(\tau)=\sum_{k=1}^{\infty}\mathbbm{1}\{U_{i,k}-Z_i\leq\tau\}$ with an~`internal' time~$\tau$ of the claim~$i$ after its reporting time~$Z_i$.

%%%%%... the same parameters are shared by different intensities $\lambda$'s ... example of decomposition ...

Note that the processes $\{N_i(t)\}_{t\geq 0}$, $i=1,2,\ldots$ do not have to be identically distributed, because of possible different effect of the reporting date. Here, the intensities $\lambda(t,Z_i;\bth)$ can be considered as \emph{payment frequencies}. However, the common parameter~$\bth$ is assumed to be shared by different intensity functions $\lambda$'s. In contrast to~\cite{Fawless1987}, we do not assume a~specific product form of the intensity~$\lambda$ and, moreover, we allow the processes~$N_i$ to depend on the process~$M$ via the reporting times~$Z_i$'s as the marks' locations. To the best of our knowledge, we are not aware of \emph{any previous work} dealing with inference for the marked non-homogeneous Poisson process with marks being non-homogeneous Poisson processes.%%% ???? FAKT ??? -> On contrary to Proposition~\ref{prop:infty}, the convergent  cumulative intensities $\Lambda_i(t;\bth)$ in Assumption~\ref{ass:Us} represent the reality that there should be a~finite number of payments per each claim (with probability one).

In order to estimate the unknown parameter $\bth$, we again use the ML approach for the arrival times, which can be considered as an~extension of the case for a~single realization of the Poisson process. Such a~framework can be extended for several independent non-homogeneous Poisson processes, where the likelihood is as follows
\begin{equation*}
L\{\bth;\bN(t),M(t)\}=\prod_{i=1}^{M(t)}\left[\left\{\prod_{k=1}^{N_i(t)}\lambda(U_{i,k},Z_i;\bth)\right\}\exp\left\{-\int_{Z_i}^{t}\lambda(\tau,Z_i;\bth)\ud \tau \right\}\right],
\end{equation*}
where~$\bZ=(Z_1,\ldots,Z_{M(t)})^{\top}$ can be also viewed as the covariates (regressors) of the intensity~$\lambda$ having corresponding realizations~$\bz$. %and $u_{i,k}$'s are the realizations of $U_{i,k}$'s. 
One should bear in mind that the whole information about the process $\{M(t)\}_{t\geq 0}$ is included in the sequence $\{Z_1,Z_2,\ldots\}$. Furthermore, the intensity function may be \emph{decomposed}, for instance, as
\begin{equation}
\lambda(t,Z_i;\bth)=\lambda_0(t-Z_i;\bnu)\exp\{f(Z_i;\bBeta)\}
\end{equation}
such that $\bth=(\bnu^{\top},\bBeta^{\top})^{\top}$, $\lambda_0(\tau;\bnu)$ is a~baseline intensity function, where $\lambda_0(\tau;\bnu)=0$ for $\tau<0$, and $f(Z_i;\bBeta)$ is a~parametric covariate function introducing the effects of the covariates $Z_i$'s.% Realize that $t-Z_i=V_i$, cf.~Figure~\ref{fig:illustration}.

To obtain the ML estimator of~$\bth$, one has to maximize the log-likelihood function
\begin{equation}\label{eq:llN}
\ell\{\bth;\bN(t),M(t)\}
=\sum_{i=1}^{M(t)}\left\{\sum_{k=1}^{N_i(t)}\log\lambda(U_{i,k},Z_i;\bth)-\int_{Z_i}^{t}\lambda(\tau,Z_i;\bth)\ud \tau \right\}.
\end{equation}
The true value~$\bth_0\in{\boldsymbol P}$ of the unknown vector parameter~$\bth$ is supposed to uniquely maximize $\E_{\bth}\ell\{\bth;\bN(t),M(t)\}$. For $\bU_i:=(U_{i,1},\ldots,U_{i,N_i(t)})^{\top}$, let us define
\[
g_i(\bU_i;\bth,t):=\int_{Z_i}^{t}\lambda(\tau,Z_i;\bth)\ud \tau-\sum_{k=1}^{N_i(t)}\log\lambda(U_{i,k},Z_i;\bth),
\]
which means that
\begin{equation}\label{eq:MLEtheta}
\widehat\bth=\arg\min_{\bth\in{\boldsymbol P}} \sum_{i=1}^{M(t)} g_i(\bU_i;\bth,t).
\end{equation}
%Let us note that the convergence of the integral $\int_{0}^{\infty}\lambda_i(v;\bnu,\bBeta)\ud v$ from Assumption~\ref{ass:Us} assures that the integral $\int_0^{v_i}\lambda_i(\tau;\bnu,\bBeta)\ud \tau$ contained in the function $g_i$ is finite and, hence, the log-likelihood function $\ell(\bnu,\bBeta;\bN|\bW=\bw,\bZ=\bz,a)$ attains only finite values. 
Assume the following to hold with respect to the $g_i(\bU_i;\bth,t)$ functions in order to obtain consistent and asymptotically normal estimators. Next assumption, being analogous to Assumption~\ref{ass:convex}, secures the existence of the unique solution. 

\begin{assumpN}\label{ass:convexN}
$g_i(\bu;\bth,t)$ are convex in $\bth\in{\boldsymbol P}$ for all $0<u_1<\ldots<u_n<t$ and $n\in\mathbb{N}$.
\end{assumpN}

At a~very first sight, further $\mathscr{N}$-assumptions might be considered as the copied $\mathscr{M}$-assump\-tions mutatis mutandis. There is, however,  an~additional layer of randomness present for the marks $N_i$'s. For instance, deterministic integrals become stochastic ones. Furthermore, the assumptions regarding the process~$M$ are not just replaced by the assumptions for the processes~$N_i$'s. There are indeed several additional assumptions regarding the marks~$N_i$'s added to some assumptions for the original underlying process~$M$. Therefore, one can neither simplify nor unify the $\mathscr{M}$-~and $\mathscr{N}$-assumptions. Next assumption, being similar to Assumption~\ref{ass:interchange1}, controls the almost sure boundedness of the derivatives of the intensity functions for all $Z_i$.

\begin{assumpN}\label{ass:interchange1N}
$\frac{\partial^2}{\partial\bth\partial\bth^{\top}}\lambda(\cdot,Z_i;\bth):\,\mathbb{R}^+\to\mathbb{R}^{p\times p}$ are continuous for all~$\bth\in{\boldsymbol P}$ and there exist Lebesgue-integrable functions $m_{1,i,j}$ and $m_{2,i,j,k}$ such that
\[
\left|\frac{\partial}{\partial\theta_j}\lambda(t,Z_i;\bth)\right|\leq m_{1,i,j}(t)\quad\mbox{and}\quad\left|\frac{\partial^2}{\partial\theta_j\partial\theta_k}\lambda(t,Z_i;\bth)\right|\leq m_{2,i,j,k}(t)
\]
almost surely, for all $\bth\in{\boldsymbol P}$, $i\in\mathbb{N}$, almost every $t>0$, and $j,k=1,\ldots,p$.
\end{assumpN}

Let us define a~cumulative intensity $\Lambda(t,Z_i;\bth):=\int_{Z_i}^t\lambda(\tau,Z_i;\bth)\ud\tau$, an~information matrix $\mathcal{J}(t;\bth_0):=\E\sum_{i=1}^{M(t)}\mathcal{J}_i(t;\bth_0)$, where $\mathcal{J}_i(t;\bth_0):=\int_{Z_i}^t\frac{\{\partial_{\bth_0}\lambda(\tau,Z_i;\bth)\}^{\otimes 2}}{\lambda(\tau,Z_i;\bth_0)}\ud \tau$, and a~matrix $\mathcal{L}_i(t,\bth_0):=\frac{1}{\sqrt{\lambda(t,Z_i;\bth_0)}}\left[\partial^2_{\bth_0}\lambda(t,Z_i;\bth)-\frac{\{\partial_{\bth_0}\lambda(t,Z_i;\bth)\}^{\otimes 2}}{\lambda(t,Z_i;\bth_0)}\right]$ for $t>0$. The forthcoming assumption differs from the analogous one \ref{ass:infmat} via averaging over all the payments~$Z_i$.

\begin{assumpN}\label{ass:infmatN}
%There exists a~symmetric non-random matrix function $\mathscr{K}(\cdot,\cdot):\,\mathbb{R}_0^+\times\mathbb{R}^p\to\mathbb{R}^{p\times p}$ such that,
As $t\to\infty$,
\begin{enumerate}
\item $M(t)\mathcal{J}^{-1}(t,\bth_0)$ converges in probability to a~positive semidefinite matrix;
%\item $\mathscr{K}^{-1/2}(t,\bth_0)\mathscr{J}(t;\bth_0)\mathscr{K}^{-1/2}(t,\bth_0)\to\mathscr{I}(\bth_0)$, which is is positive definite;
\item $\E\sum_{i=1}^{M(t)}\int_{Z_i}^{t}\left\{\mathcal{J}^{-1/2}(t,\bth_0)\mathcal{L}_i(\tau,\bth_0)\mathcal{J}^{-1/2}(t,\bth_0)\right\}^2\ud \tau\to\zero$.
\end{enumerate}
\end{assumpN}

Analogous discussions like after Assumptions~\ref{ass:convex}--\ref{ass:infmat} might be carried out regarding Assumptions~\ref{ass:convexN}--\ref{ass:infmatN}. Briefly and informally, the intensity functions~$\lambda$'s representing the behavior of the payment times' processes~$N_i$'s are supposed to be \emph{sufficiently smooth} and \emph{adequately regular} with respect to the amount of information about the parameter~$\bth$ contained in~$N_i$'s.

%%%\begin{assumpN}\label{ass:expansionN}
%%%There exist an~expansion
%%%\[
%%%g_i(\bT_i,V_i,N_i;\bnu_0+\bt_1,\bBeta_0+\bt_2)-g_i(\bT_i,V_i,N_i;\bnu_0,\bBeta_0)=\bC_i(\bT_i,V_i,N_i)^{\top}\bt+S_i(\bT_i,V_i,N_i;\bt),
%%%\]
%%%such that $\E\bC_i(\bT_i,V_i,N_i)=\zero$ and %%%$\left[\bt_1^{\top},\bt_2^{\top}\right]^{\top}\equiv\bt$.
%%%\end{assumpN}

%%%With the previous development in mind, let us denote
%%%\begin{align}
%%%\E S_i(\bT_i,V_i,N_i;\bt)&=\frac{1}{2}\bt^{\top}\bB_i\bt+w_{i,0}(\bt), & \sum_{i=1}^n \Var S_i(\bT_i,V_i,N_i;\bt)&=w_n(\bt),\label{eq:ERiN}\\
%%%\sum_{i=1}^n\bB_i&=\bL_n, & \sum_{i=1}^n \Var\bC_i(\bT_i,V_i,N_i)&=\bM_n.\nonumber
%%%\end{align}

%%%\begin{assumpN}\label{ass:restN}
%%%For all $\bss\in\mathbb{R}^{r+s}$, $\sum_{i=1}^n w_{i,0}(\bss /\sqrt{n})\to 0$ and $w_n(\bss /\sqrt{n})\to 0$ as $n\to\infty$. Moreover, $\bL_n/n\to \bL$ and $\bM_n/n\to \bM$ as $n\to\infty$, where $\bL$ is positive definite.
%%%\end{assumpN}

\begin{theorem}[Consistency II]\label{thm:consistencyN}
Under Assumptions~\ref{ass:Zs} and \ref{ass:Us}--\ref{ass:infmatN},
\[
\mathcal{J}^{1/2}(t,\bth_0)\left(\widehat{\bth}-\bth_0\right)=-\mathcal{J}^{-1/2}(t,\bth_0)\sum_{i=1}^{M(t)}\partial_{\bth_0}g_i\left(\bU_i;\bth_0,t\right)+o_{\prob}(1),\quad t\to\infty.
\]
\end{theorem}

Again, let us discretize the `continuous' time $t\in\mathbb{R}_0^+$ for the processes $\{N_i(t)\}_{t\geq 0}$ in a~way that one observes $N_i$ only at discrete time points $a\in\mathbb{N}$, e.g., status at the closed calendar years. The next Lindeberg condition is used to extend the assertion of Theorem~\ref{thm:consistencyN} in order to derive asymptotic normality of~$\widehat{\bth}$.
\begin{assumpN}\label{ass:lindN}
$\lim_{a\to\infty}\sum_{i=1}^{a}\E\left({\boldsymbol d}^{\top}\mathcal{Y}_i\right)^2\mathbbm{1}\{|{\boldsymbol d}^{\top}\mathcal{Y}_i|\geq \eps\|{\boldsymbol d}\|_2\}=0$ for all ${\boldsymbol d}\in\mathbb{R}^{p}$ and $\eps>0$, where $\mathcal{Y}_i:=\mathcal{J}^{-1/2}(a,\bth_0)\int_{j-1}^{j}\int_{z}^{a}\left\{\partial_{\bth_0}\log\lambda(\tau,z;\bth)\right\}(\ud\tilde{N}_z(\tau-z)-\lambda(\tau,z;\bth_0)\ud\tau)\ud M(z)$.
%For all ${\boldsymbol d}\in\mathbb{R}^{p}$ and $\eps>0$,
%\begin{multline*}
%\lim_{a\to\infty}\E\sum_{i=1}^{M(a)}\bigg[\left({\boldsymbol d}^{\top}\mathcal{J}^{-1/2}(a,\bth_0)\partial_{\bth_0}g_i\left(\bU_i;\bth_0,a\right)\right)^2\bigg.\\
%\bigg.\mathbbm{1}\left\{\left|{\boldsymbol d}^{\top}\mathcal{J}^{-1/2}(a,\bth_0)\partial_{\bth_0}g_i\left(\bU_i;\bth_0,a\right)\right|\geq\eps\right\}\bigg]= 0.
%\end{multline*}
\end{assumpN}

The Lyapunov condition can be assumed as well. For instance, for all ${\boldsymbol d}\in\mathbb{R}^{p}$, there exists $\delta>0$ such that $\lim_{a\to\infty}\sum_{i=1}^{a}\E\left|{\boldsymbol d}^{\top}\mathcal{Y}_i\right|^{2+\delta}=0$.

\begin{corollary}[Asymptotic normality II]\label{cor:asnormN}
Under Assumptions~\ref{ass:Zs} and \ref{ass:Us}--\ref{ass:lindN},
\[
\mathcal{J}^{1/2}(a,\bth_0)\left(
\widehat{\bth}-\bth_0\right)\xrightarrow[a\to\infty]{\dist}\mathsf{N}_{p}\left(\zero,{\boldsymbol I}\right).
\]
\end{corollary}

%%%%%One of the simplest examples ... Possible choices for the decreasing baseline intensity functions are
%%%%%\begin{equation}\label{eq:wei_base}
%%%%%\lambda_0(\tau;\bnu)=\nu_1\nu_2\tau^{\nu_1-1}
%%%%%\end{equation}
%%%%%or
%%%%%\begin{equation}\label{eq:exp_base}
%%%%%\lambda_0(\tau;\bnu)=\exp\{\nu_1+\nu_2\tau\}
%%%%%\end{equation}
%%%%%as stated in~\citet[Subsection 8.5, p.~166]{CKSS1991}.

The simplest situation is that each $\{N_i(t)\}_{t\geq 0}$ is a~homogeneous Poisson process after the reporting time~$Z_i$ having a~constant common intensity $\theta>0$ for all~$i$'s.
\begin{example}
Intensity $\lambda(\tau,Z_i;\theta)=\theta\mathbbm{1}\{\tau\geq Z_i\}$ with ${\boldsymbol P}=(0,\infty)$. The log-likelihood function is $\ell\{\theta;\bN(t),M(t)\}=(\log\theta)\sum_{i=1}^{M(t)} N_i(t)-\theta\sum_{i=1}^{M(t)}(t-Z_i)$. The ML estimator becomes $\hat{\theta}=\frac{\sum_{i=1}^{M(t)}N_i(t)}{tM(t)-\sum_{i=1}^{M(t)}Z_i}$ and the information number is $\mathcal{J}(t;\theta_0)=\left\{t\Psi(t;\brho_0)-\int_0^tz\psi(z;\brho_0)\ud z\right\}/\theta_0$. Assumption~\ref{ass:infmatN} together with the Lyapunov condition can be checked as well. Hence,
\[
\sqrt{\frac{a\Psi(a;\brho_0)-\int_0^az\psi(z;\brho_0)\ud z}{\theta_0}}\left(
\widehat{\theta}-\theta_0\right)\xrightarrow[a\to\infty]{\dist}\mathsf{N}\left(0,1\right).
\]
Moreover, if the underlying intensity of the Poisson process~$M$ is $\psi(t;\rho)=\rho$ for $t\geq 0$ (i.e., homogeneous Poisson process), then $\sqrt{\frac{\rho_0}{2\theta_0}}a\left(
\widehat{\theta}-\theta_0\right)\xrightarrow[a\to\infty]{\dist}\mathsf{N}\left(0,1\right)$.
\end{example}

Another example contains a~baseline intensity, which is motivated by~\citet[Subsection 8.5, p.~166]{CKSS1991}.
\begin{example}
Intensity $\lambda(\tau,Z_i;\bnu,\eta)=\nu_1\nu_2(\tau-Z_i)^{\nu_1-1}\exp\{\eta Z_i\}\mathbbm{1}\{\tau\geq Z_i\}$. Here, ${\boldsymbol P}=(0,\infty)^2\times\mathbb{R}$. The log-likelihood becomes
%%%\begin{multline*}
%%%\ell\{\bnu,\eta;\bN(t),M(t)\}=\sum_{i=1}^{M(t)}\Bigg\{N_i(t)\log(\nu_1\nu_2)+(\nu_1-1)\sum_{k=1}^{N_i(t)}\log(U_{i,k}-Z_i)\\+\eta Z_iN_i(t)-\nu_2\left(t-Z_i\right)^{\nu_1}\exp\{\eta Z_i\} \Bigg\}.
%%%\end{multline*}
\begin{multline*}
\ell\{\bnu,\eta;\bN(t),M(t)\}\\=\sum_{i=1}^{M(t)}\left\{N_i(t)\log(\nu_1\nu_2)+(\nu_1-1)\sum_{k=1}^{N_i(t)}\log(U_{i,k}-Z_i)+\eta Z_iN_i(t)-\nu_2\left(t-Z_i\right)^{\nu_1}\exp\{\eta Z_i\} \right\}.
\end{multline*}
The ML estimator has to be computed numerically. The above formulated assumptions are satisfied for the particular open convex ${\boldsymbol P}\subseteq\mathbb{R}^3$. The defined entities are not presented here due to their complicated forms.
\end{example}

The next example of the intensity function is directly used in the consequent practical analysis of our data for modeling the payment times of bodily injury as well as material damage claims.
\begin{example}\label{ex:N}
$\lambda(\tau,Z_i;\bnu,\bBeta)=\exp\{\nu_1+\nu_2(\tau-Z_i)+\eta_1\cos\left(2\pi Z_i/\eta_3\right)+\eta_2\sin\left(2\pi Z_i/\eta_3\right)\}$. The defined entities are again not presented here due to their voluminous forms.
\end{example}

%%%%%For instance in case of the ``Weibull'' baseline intensity function~\eqref{eq:wei_base}, the log-likelihood can become
%%%%%\begin{align}
%%%%%&\ell(\bnu,\bBeta;\bN|\bW=\bw,\bZ=\bz,a)\nonumber\\
%%%%%&=\sum_{i=1}^n\left[n_i\log(\nu_1\nu_2)+(\nu_1-1)\sum_{k=1}^{n_i}\log\tau_{ik}+\eta_1w_i+\eta_2z_i-\nu_2\left(v_i\right)^{\nu_1}\exp\{\eta_1w_i+\eta_2z_i\} \right].
%%%%%\end{align}
%%%%%Furthermore, in case of the ``exponential'' baseline intensity function~\eqref{eq:exp_base}, the log-likelihood could be rewritten as
%%%%%\begin{align}
%%%%%&\ell(\bnu,\bBeta;\bN|\bW=\bw,\bZ=\bz,a)\nonumber\\
%%%%%&=\sum_{i=1}^n\left[\nu_1n_i+\nu_2\sum_{k=1}^{n_i}\tau_{ik}+\eta_1w_i+\eta_2z_i-\frac{\exp\{\nu_1\}}{\nu_2}\left(\exp\{\nu_2 v_i\}-1\right)\exp\{\eta_1w_i+\eta_2z_i\} \right].
%%%%%\end{align}

\subsection{Reporting delay}\label{subsec:ReportingDelay}
%\crossoutred{assumption on $W_n$'s}%%%, cf.~\cite{Diao2013}
%\ostap{Preliminary analysis on reporting delays, show independence between delays for different contracts, but the distribution of them is changing with respect to the reporting time $Z_i$. In the upper four panels of Figure \ref{fig:numacc} with grey lines we represent estimates of the parameters of the log-normal distribution (location parameter on top panels, and scale parameter on the middle panels) over different weeks of reporting date (left panels for material cases and right panels for bodily cases). This simple analysis shows strong temporal dependence of the parameters on the distribution of $W_i$, therefore next assumption of independency but no identical distribution of $W_i$ is a natural implication of this empirical evidence.}
Since the reporting delays correspond to different claims from different accidents, independence between the reporting delays $W_i$'s for different contracts is assumed. However, the distribution of the reporting delays is allowed to change with respect to the reporting time $Z_i$. The reporting delays seem to become shorter and shorter, which can be explained by a~possibility to report an~accident over the internet and even by a~denser net of the insurance company branches. So, given $Z_i$, $W_i$ has a~parametric conditional density $f_{W}\{\cdot;w(Z_i,\bvartheta)\}$, where $\bvartheta\in\mathbb{R}^r$. Note that $\{W_i\}_{i\in\mathbb{N}}$ are not identically distributed, which allows, for instance, to assume \emph{time-varying distributions} for the $W_i$'s through the function $w(\cdot,\bvartheta)$. Using similar arguments as in~\cite{HjortPollard2011}, one has consistency and asymptotic normality for the ML estimator~$\widehat{\bvartheta}$.

%%%%%\begin{assumpW}\label{ass:Ws}
%%%%%The reporting delays $W_i$'s are independent random variables. Given $Z_i=z$, $W_i$ has a~parametric conditional density $f_{W}(\cdot,z;\bvartheta)$, where $\bvartheta\in\mathbb{R}^r$.
%%%%%\end{assumpW}

%\ostap{Relying on the classical maximum-likelihood theory, assuming, that true parameters $\bth$ are lying in the interior of the parameter space and that the likelihood function is convex, we obtain consistency and asymptotic normality of the $\bth$, see ANYPAPERONMLESTIMATE. Following seasonal shape of the parameters in Figure~\ref{fig:numacc} we set $f_{W}(\cdot,z;\bth) = f_{W}\{\cdot;g(z,\bth)\}$.}

%%%%%%%%\crossoutred{In our study, the most suitable distributions to describe the conditional behavior of the reporting delay are log-normal, Weibull, and Gamma distributions with time-varying parameters. Other distributions with similar shapes were considered as well, but they did not provide reliable performance. For a~brief recall, the densities of each of the three considered distributions are}
A~variety of parametric distributions are suitable for the analysis, particularly those with the shapes similar to the ones provided by log-normal, Weibull, or Gamma distributions. All of them have similar performance, whereas the importance lies in the time-varying parameters. Although in the rest of the study, we concentrate ourselves purely on the log-normal distribution, let us briefly recall all three mentioned densities
\begin{align}
f_{Gam}(x; c, d) &= \frac{1}{d^c\Gamma(c)}x^{c-1}\exp\left(-\frac{x}{d}\right), \quad x\geq0, c>0,d>0;\nonumber\\
f_{Wei}(x; c, d) &= \frac{a}{d^c}x^{c-1}\exp\left\{-\left(\frac{x}{d}\right)^c\right\}, \quad x\geq0, c>0,d>0;\nonumber\\
f_{LN}(x; c, d)  &= \frac{1}{\sqrt{2\pi}xd}\exp\left\{-\frac{(\log x - c)^2}{2d^2}\right\}, \quad x>0,c\in\mathbb{R},d>0.\label{eq:LN}
\end{align}
For the Gamma and Weibull distributions, parameters~$c$ and~$d$ are called the shape and scale, respectively, for the log-normal distribution~$c$ and~$d$ are the mean and standard deviation of the distribution on the log scale, respectively.

%\crossoutred{The reasoning, why the conditional distribution of reporting delays should depend on time, comes from the exploratory analysis of the claims' database. Let us define $j$ the index of the week, starting with $j=1$ being the first full week in the observation sample. For $j=1$, correspond dates from the interval $t_1=\{20000103,\ldots,20000109\}$, for $j = 2$ correspond dates $t_2=\{20000110,\ldots,20000116\}$, etc. In total whole observation sample covers $J = 775$ weeks. The number of observations $|\{W_i(t)|t\in t_j\}|,j\in1,\ldots,J$ over weeks heavily varies, providing small amount of observations in the beginning of the observation sample, with gradual increase and further decrease towards the end of the observation sample, see bottom panel of the Figure 3 for both types of accidents. By estimating the unconditional distributions of $W_i(t)$ separately for the accidents happening in each week we observe the dependency of the parameters of the distribution on the week of the accident date, as depicted with black lines in top and middle panel of Figure 3 for shape and scale of the Weibull distributions. Similar behavior is observed if one considers material and bodily injuries separately.}

Taking into account the dependency between the accident date~$Z_i$ and the reporting delay~$W_i$ and, additionally to that, allowing for possible seasonal behavior, we consider truncated Fourier series for the parameters of the conditional distributions with
\begin{align*}
\bvartheta_1 &= (\alpha_c, \beta_c, \delta_{c,1}, \xi_{c,1}, \gamma_{c,1},\ldots,\delta_{c,L}, \xi_{c,L}, \gamma_{c,L})^\top,\\
\bvartheta_2 &= (\alpha_d, \beta_d, \delta_{d,1}, \xi_{d,1}, \gamma_{d,1},\ldots,\delta_{d,L}, \xi_{d,L}, \gamma_{d,L})^\top
\end{align*}
in the form of 
\begin{align}
c(z, \bvartheta_1) &= \alpha_c + \frac{z}{7}\beta_c + \sum_{\ell = 1}^L\left\{\delta_{c,\ell}\cos\left(\frac{\xi_{c,\ell}\cdot2\pi\cdot z}{52\cdot 7}\right) + \gamma_{c,\ell}\sin\left(\frac{\xi_{c,\ell}\cdot2\pi\cdot z}{52\cdot 7}\right)\right\},\label{eq:time-varying1}\\
d(z, \bvartheta_2) &= \alpha_d + \frac{z}{7}\beta_d + \sum_{\ell = 1}^L\left\{\delta_{d,\ell}\cos\left(\frac{\xi_{d,\ell}\cdot2\pi\cdot z}{52\cdot 7}\right) + \gamma_{d,\ell}\sin\left(\frac{\xi_{d,\ell}\cdot2\pi\cdot z}{52\cdot 7}\right)\right\},\label{eq:time-varying2}
\end{align}
where~$52$ is the number of weeks in one year and~7 is the number of days in one week. Later on for considering models with \emph{increasing flexibility}, we discuss constant models with $L = 0$ and $\beta_c=\beta_d=0$, linear models with $L=0$, and models with one ($L = 1$) or two ($L= 2$) seasonality patterns. The ML estimator of $\bvartheta = (\bvartheta_1^\top, \bvartheta_2^\top)^\top$ is obtained as% has the following form
\begin{align}
&\widehat{\bvartheta}= \arg\max_{\bvartheta}\sum_{i=1}^{M(t)}\log f_{W}\{W_i;c(Z_i,\bvartheta_1), d(Z_i,\bvartheta_2)\},\label{eq:MLEvartheta1}\\
&\mbox{s.t.}\quad c(Z_i,\bvartheta_1)>0,\, d(Z_i,\bvartheta_2)>0, \quad \mbox{for all } i\in\{1,\ldots, M(t)\}\label{eq:MLEvartheta2}.%,
\end{align}
%where~$w_i$ is the observed reporting delay of the $i$th claim with the corresponding realized reporting date~$z_i$. 
The condition~$c(Z_i,\bvartheta_1)>0$ is considered only in case when the assumed distribution~$f_W$ is Weibull or Gamma. As initial values in the iterative maximization procedure, we took the parameter values that at best fit the piecewise constant weekly averaged values in the least squares sense.

%\crossoutred{As too few observations were gathered in the beginning of the observation and in the end of the observation period, the estimation procedure considers only observations lying in the interval $20030106-20131229$,}

Figure~\ref{fig:numacc} shows the estimation results for~$f_W$ being the log-normal distribution. 
\begin{figure}[!ht]
\begin{center}
\includegraphics[width=.97\textwidth]{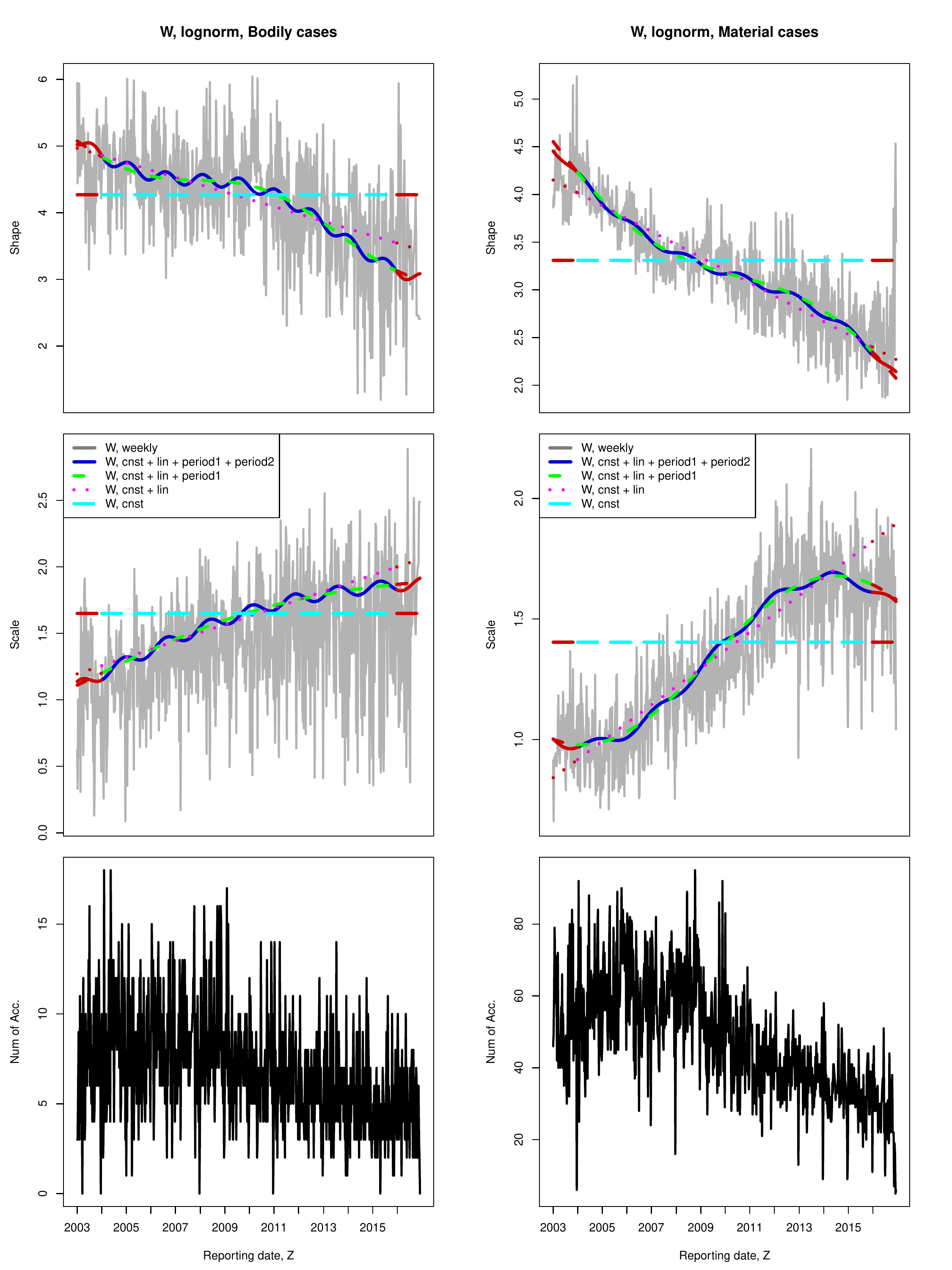}
\caption{Weekly estimates (solid grey) and conditional temporal models: constant (cyan dashed), linear trend (pink dotted), linear trend and one period (green dashed), and linear trend with two periods (blue solid) for shape (top panels) and scale (middle panels) of the log-normal distribution of the reporting delay~$W_i$. The extrapolated periods are depicted in red and the numbers of accidents are in black.}
\label{fig:numacc}
\end{center}
\end{figure}
The left panel corresponds to the bodily injury claims, where the right one to the material damage claims, on the bottom panel we show the number of accidents as the function of the reporting date. In the two upper panels, we depict the location parameters $c(z, \widehat{\bvartheta}_1)$ and two middle panels scale parameters $d(z, \widehat{\bvartheta}_2)$ over different weeks of the reporting date. Different colors represent different complexities of the models used: for cyan we have a~constant unconditional model with $L = 0$ and $\beta_c=\beta_d=0$; pink uses only linear temporal dependency with $L = 0$; green and blue lines have one ($L = 1$) and two ($L = 2$) levels of seasonality. The most flexible density with $L = 2$ has been used in the final omnibus model. With red we depict values extrapolated to the data not used in the estimation, namely years 2003 and 2016. 
%%%In the upper four panels of Figure~\ref{fig:numacc} grey lines represent estimates of the parameters of the log-normal distribution (location parameter on top panels, and scale parameter on the middle panels) over different weeks of reporting date (the left panels for material damage claims and the right panels for bodily injury claims). 
This analysis shows strong temporal dependence of the parameters on the distribution of~$W_i$.%%%%% Therefore, Assumption~\ref{ass:Ws} is a~natural implication of this empirical evidence.
%\crossoutred{thus the resulting estimates presented with red in the Figure~\ref{fig:numacc} on the top and middle plot. The blue lines are just prediction onto the complete time span, namely on the weeks not used in the estimation procedure. The same has been performed separately for material and bodily injuries using three distributions mentioned above.}

%For example, the conditional gamma distribution seems to be suitable, where the conditional density is
%\begin{equation}
%f_{W|t}(w)=\frac{1}{\Gamma(\xi)}\left[-\xi(\beta_0+\beta_1 t)\right]^{\xi}w^{\xi-1}\exp\left\{\xi(\beta_0+%\beta_1 t) w\right\}.
%\end{equation}
%for $w>0$

%via gamma with inverse link

In order to validate our results from the fitted model, Figure~\ref{fig:Delay_back} compares the observed and predicted quarterly averaged reporting (waiting) delays in days for the bodily injury claims as well as for the material damage claims. 

\begin{figure}[!ht]
\begin{center}
\includegraphics[width=0.5\textwidth]{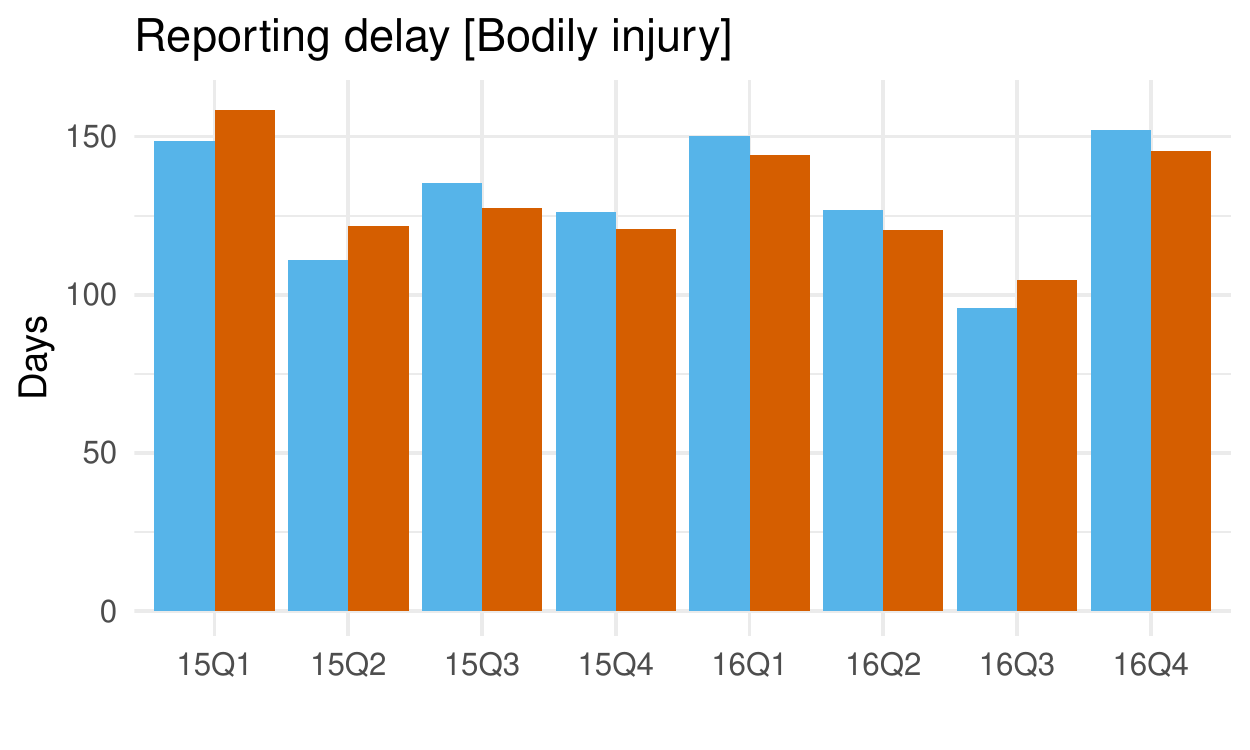}\includegraphics[width=0.5\textwidth]{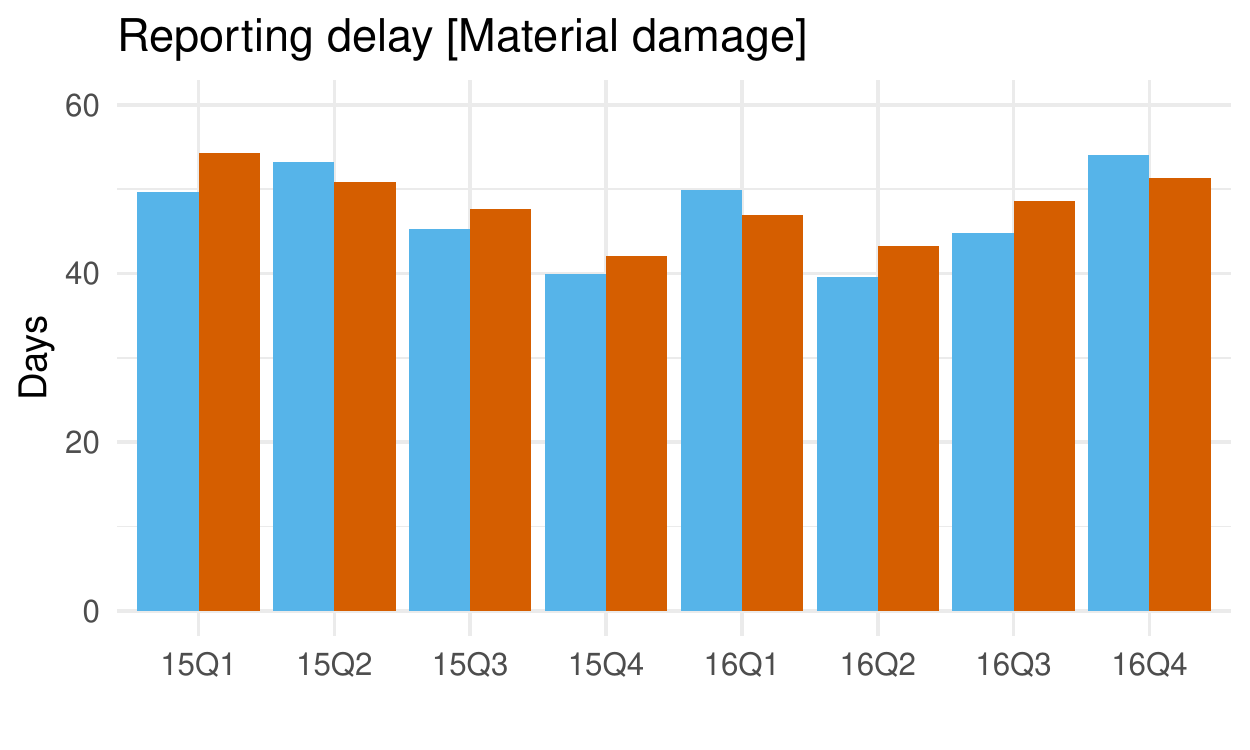}
\caption{Quarterly averaged (with respect to the reporting date) reporting (waiting) delays in days---observed in blue and predicted in orange for the period 2015--2016 (prediction uses only data up to the end of 2015).}
\label{fig:Delay_back}
\end{center}
\end{figure}

\subsection{Payment amounts}\label{subsec:payment_amounts}
Denote the $j$th payment amount for the $i$th claim by $X_{i,j}$, where $j=1,\ldots,N_i(t)$. The $X_{i,j}$'s are independent over all $j$'s as well as all $i$'s. This independency assumption can be easily relaxed and any other times series model (e.g., an~autoregressive model) can be used instead. However, our empirical findings imply independency. Given $Z_i$, $X_{i,j}$ has a~parametric conditional density $f_{X}\{\cdot;v(Z_i,{\boldsymbol\varsigma})\}$, where ${\boldsymbol\varsigma}\in\mathbb{R}^s$ and the function $v(\cdot,{\boldsymbol\varsigma})$ introduces the time-varying effects of $Z_i$'s. Moreover, the reporting delays~$W_i$'s are also supposed to be independent from the payment amounts~$X_{i,j}$'s, as well as all the payment amounts from the same claim are independent among each other. These assumptions are based on the preliminary empirical analysis of the pairwise relationships between the waiting time ($W_i$) and the first ($X_{i, 1}$), second ($X_{i, 2}$), third ($X_{i, 3}$), and fourth ($X_{i, 4}$) claim payment amounts shown in Figure~\ref{fig:WvsX}. For this plot, the data are transformed by the estimated cumulative distribution functions (cdf) $\hat F_X$ and $\hat F_W$ obtained from the plugged-in densities $\hat f_X(\cdot)\equiv f_X\{\cdot,v(Z_i,\widehat{\boldsymbol{\varsigma}})\}$ and $\hat f_W(\cdot)\equiv f_W\{\cdot,w(Z_i,\widehat{\boldsymbol{\vartheta}})\}$, respectively. Further, the data are transformed via the quantile function of the standard normal distribution. If the distributional assumptions are correct and the payments and delays are independent, the bivariate kernel density estimates and scatterplots of the transformed data should be suggestive of circular shapes as clearly visible in Figure~\ref{fig:WvsX}.

%%%%%preliminary statistical analysis, that motivate our models

\begin{figure}[!ht]
\begin{center}
\includegraphics[width=0.5\textwidth]{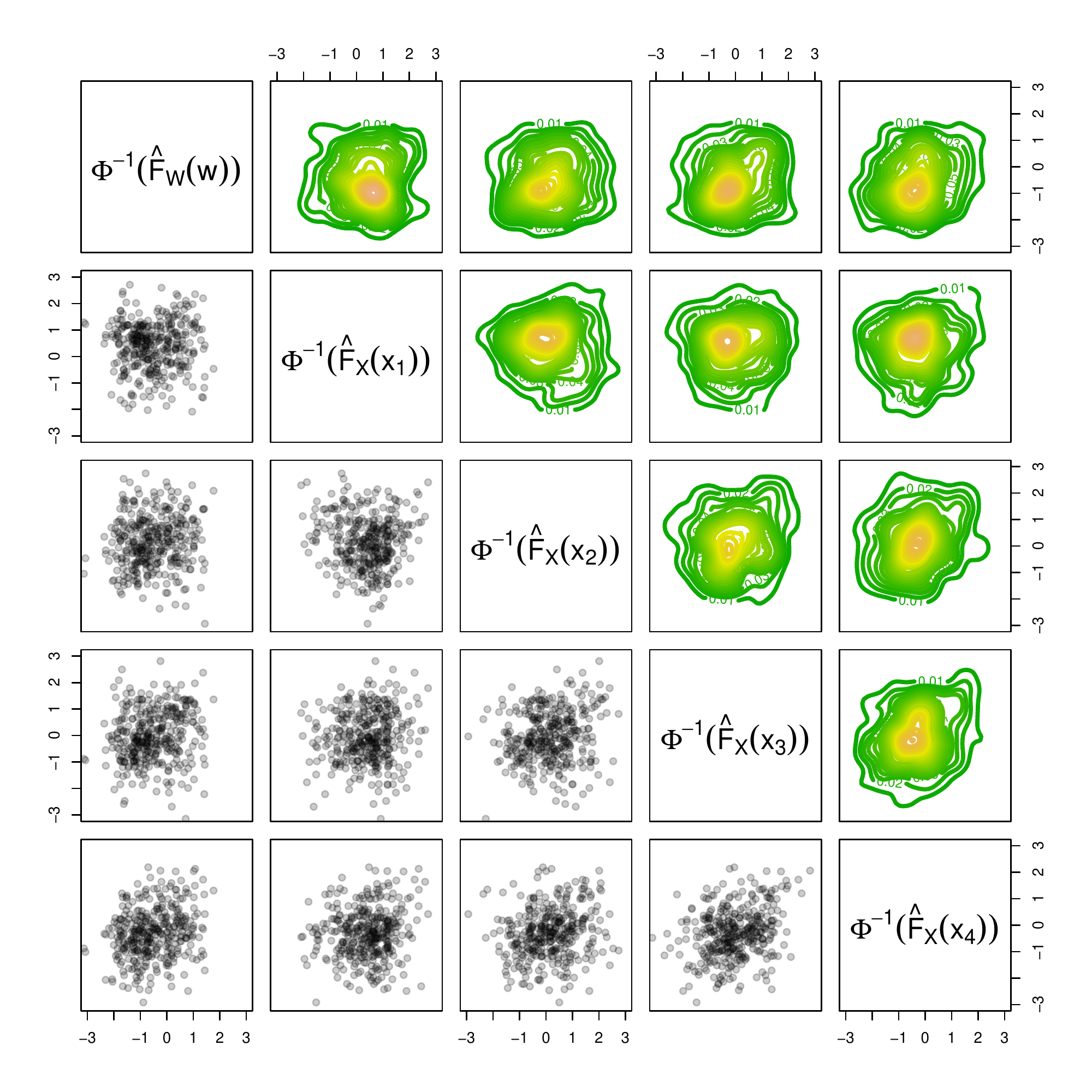}\includegraphics[width=0.5\textwidth]{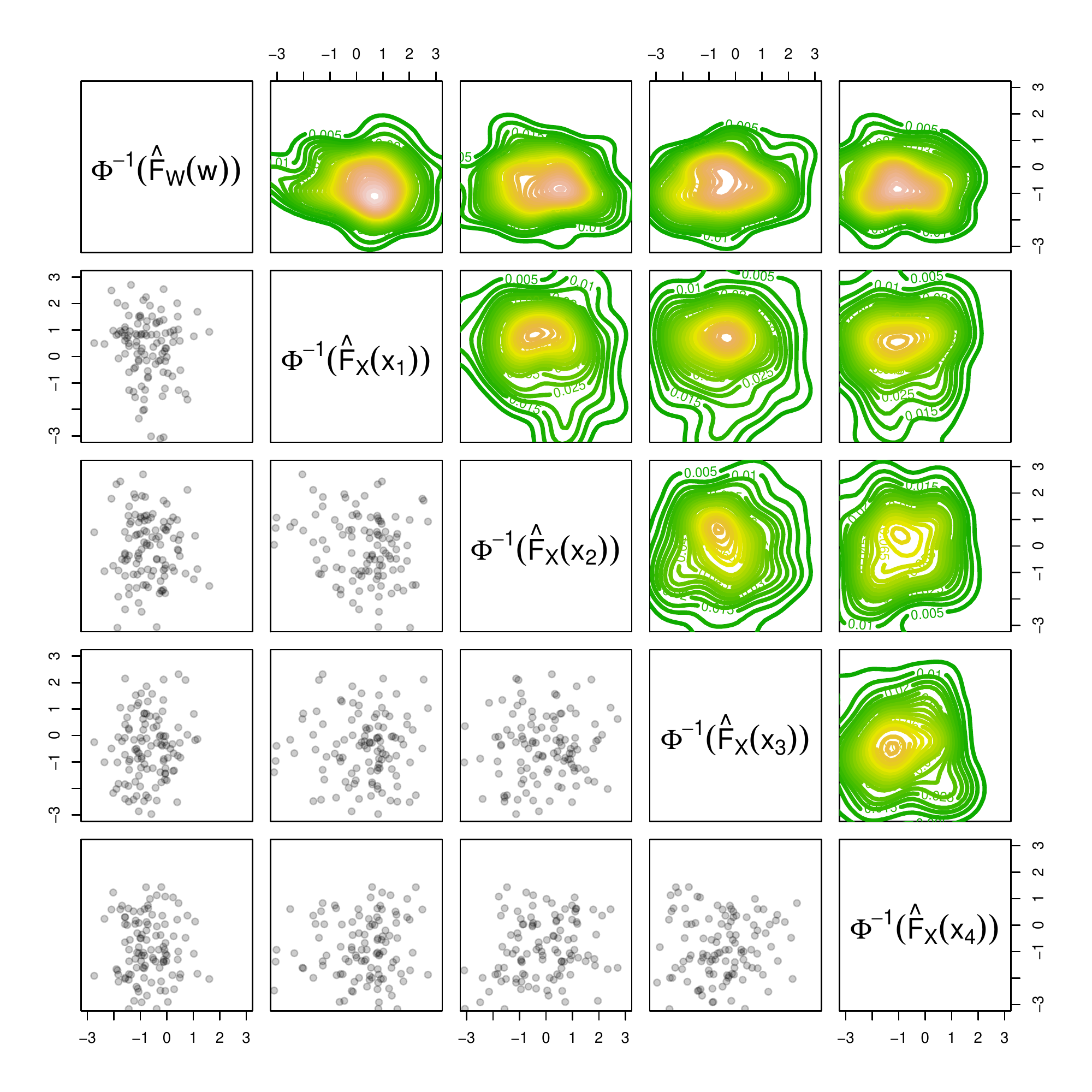}
\caption{Pairwise relationship between the reporting delay and the claim payment amounts (bodily injury claims---left; material damage claims---right). Subfigures below the diagonal show scatter plots for the transformed reporting delays/payment amounts $\Phi^{-1}\{\hat{F}_j(\cdot)\}$ versus $\Phi^{-1}\{\hat{F}_k(\cdot)\}$, where $j,k\in\{W,X_1,X_2,X_3,X_4\}$, $j\neq k$, $\hat{F}_j$ is the corresponding estimated cdf, and $\Phi$ is the cdf of the standard normal distribution. Subfigures above the diagonal display contour plots of $\Phi^{-1}\{\hat{F}_j(\cdot)\}$ against $\Phi^{-1}\{\hat{F}_k(\cdot)\}$.}
\label{fig:WvsX}
\end{center}
\end{figure}

%%%%%\begin{assumpX}\label{ass:Xs}
%%%%%The payment amounts $X_{i,j}$'s are independent random variables. Given $Z_i=z$, $X_{i,j}$ has a~parametric conditional density $f_{X}(\cdot,z;{\boldsymbol\varsigma})$, where ${\boldsymbol\varsigma}\in\mathbb{R}^s$.
%%%%%\end{assumpX}

Using similar arguments as in~\cite{HjortPollard2011}, one can prove consistency and asymptotic normality for the ML estimator~$\widehat{\boldsymbol\varsigma}$. The procedure for modeling the claim payments closely resembles the procedure mentioned when modeling the reporting delays in Subsection~\ref{subsec:ReportingDelay}. For the modeling of the payment amounts, we also considered more complex models, where previous payments were included as exogenous variables. This however did not bring any improvements.

Figure~\ref{fig:numacc1} presents the time-varying parameters of log-normal distribution of bodily injury as well as material damage claims for the first payment in yellow (fully flexible model with two seasonal periods) and in grey (separately for each week). Other curves present parameters for all the payment amounts pooled together. As there is no much difference and results for the pooled models seemed to be more stable, we concentrate in the later only on the pooled ones, namely the most flexible with two periods and a~linear trend.

\begin{figure}[!ht]
\begin{center}
\includegraphics[width=.97\textwidth]{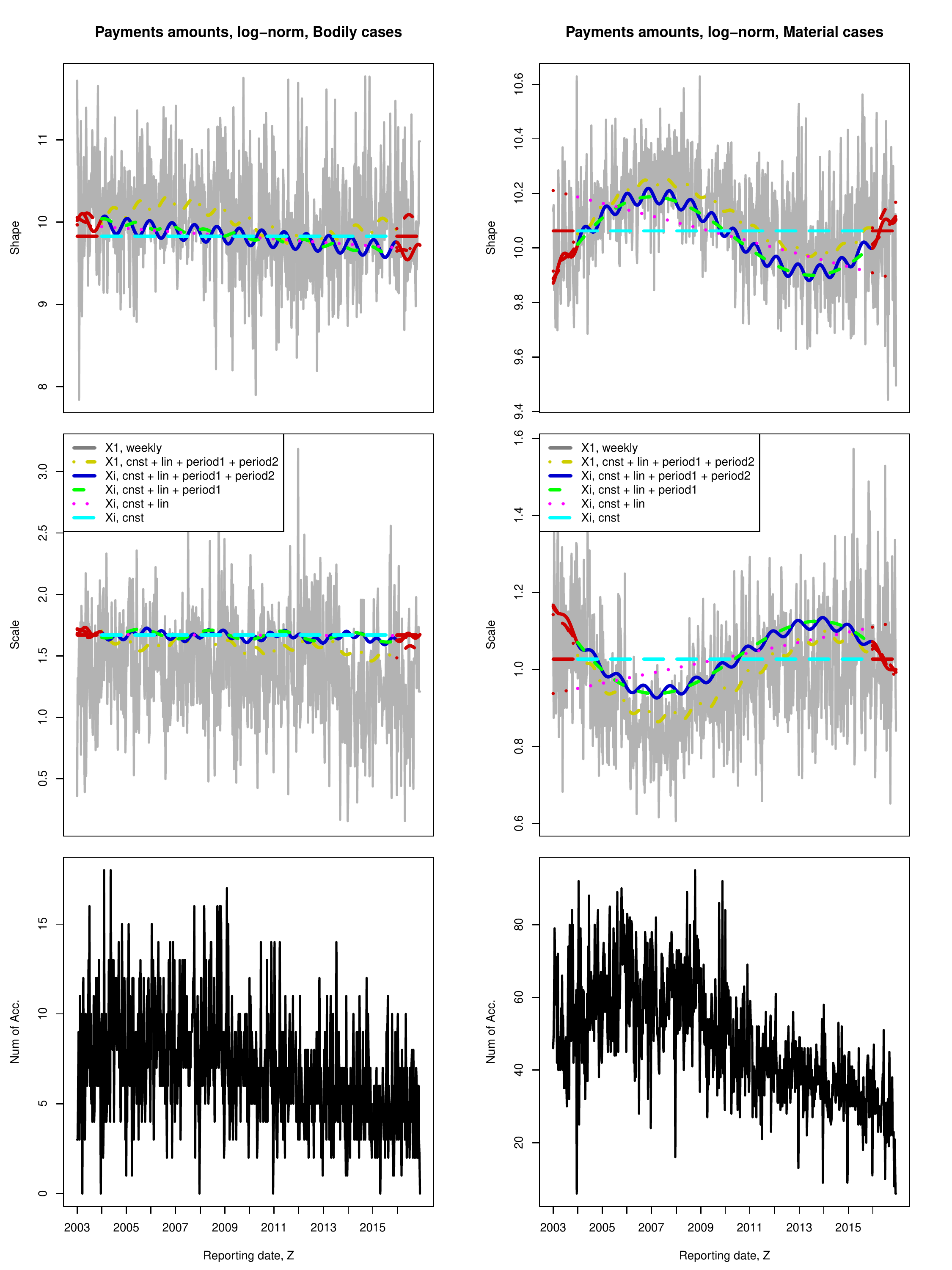}
\caption{Weekly estimates (only $X_{1,t}$ in solid grey) and conditional temporal models: constant (cyan dashed), linear trend (pink dotted), linear trend and one period (green dashed), and linear trend with two periods (blue solid for $X_{i,t}$ and yellow dot-dashed for $X_{1,t}$) for shape (top panel) and scale (middle panel) of the log-normal distribution of the $X_{i,t}$. The extrapolated periods are depicted in red and the numbers of accidents are in black.}
\label{fig:numacc1}
\end{center}
\end{figure}

\section{Practical application and empirical results}\label{sec:practical}
To numerically illustrate the performance of our method, we use two data sets---bodily injury and material damage claims (cf.~motivation and data description in Section~\ref{sec:data}). Let us recall that data from the last available year 2016 are used only for back-testing and comparison with the predicted results. Furthermore, our `micro' (granular, claim-by-claim) approach is also compared with a~traditional standard actuarial technique---\emph{bootstrap chain-ladder} \citep{EV1999}---in combination with linear extrapolation of the reported claims in the next year. This `macro' approach is based on aggregation of data and, hence, it disregards the information about the policy and the claim's development. We refer to it from now on as the aggregated method.

%%%%%results are going to be compared to the traditional standard method based on the chain-ladder approach in combination with predicted next cell? ...

Firstly, the previously described estimation procedures (Section~\ref{sec:theory}) provide parameter estimates of our omnibus model. Secondly, a~Monte Carlo prediction technique is involved in order to generate (simulate) the future claims' developments. In essence, prediction for a~distribution of the total payments in year 2016, for which we also possess the real paid claim amounts.

\subsection{Parameter estimation}
All the estimates are obtained through the ML approach, which guaranties a~proper stochastic inference. For the case of densities, it is widely known that under some regularity conditions the ML estimators are consistent and asymptotically normal. For the case of intensities, we have proved consistency and asymptotic normality of these ML estimators. Consequently, one can plug-in the estimated parameters into the parametric forms of the densities and intensities present in our micro model in order to have predicted (fitted) intensities of the reporting dates/payment dates and densities of the reporting delays/payment amounts. They are going to be used for the simulation of the future payments (dates and amounts).

In particular, the intensity function for modeling the reporting times of bodily injury claims comes from Example~\ref{ex:Zbodily} and in case of material damage claims from Example~\ref{ex:Zmaterial}. The reporting delay and the claim payments are modeled as in relations~\eqref{eq:LN}--\eqref{eq:MLEvartheta2}. And the intensity functions for the payment times of bodily injury as well as material damage claims come from Example~\ref{ex:N}.

\subsection{Monte Carlo predictions}
Basically for each claim being reported up to the future time point~$b$ from Figure~\ref{fig:illustration} (e.g., end of the next calendar year), one needs to simulate payment dates and corresponding payment amounts, which are going to be summed in each simulation's run. These sums of payments give us the simulated (empirical) predictive distribution of the total future payments. Hence, for the next year (in a~general future time window $(a,b]$), we need to simulate the new payment dates for all already reported claims as well as the payment dates corresponding to the incurred but not reported claims. Consequently, it is requisite to generate a~corresponding payment amount for every payment time within the time interval $(a,b]$.

Let us realize that we need to generate many realizations of the non-homogeneous Poisson process for each Monte Carlo simulation's run. To simulate the non-homoge\-neous Poisson process, we suggest to rely on the \emph{thinning} algorithm by~\cite{LS1979}. The main reason for choosing this way to generate enormous number of realizations of the non-homogeneous Poisson process is that this approach can be applied to any rate function without the necessity of numerical integration or simulation of Poisson variables.

\subsubsection{Primary goal: Prediction of distribution of the future cash flows}
Our primary target is to predict the distribution of the total payment amounts within the future time period. Such a~prediction is going to be obtained through Procedure~\ref{alg:primary}.
\begin{algorithm}[!ht]
\caption{Prediction of the distribution of the future payments}
\label{alg:primary}
%\algsetup{indent=2em}
\begin{algorithmic}[1]
\REQUIRE Collection of observations $\{T_i,Z_i,\{U_{i,j}\}_{j=1}^{N_i(a)},\{X_{i,j}\}_{j=1}^{N_i(a)}\}_{i=1}^{M(a)}$ and number of Monte Carlo simulation's runs~$S$
\ENSURE Simulated predictive distribution of the total future payments in $(a,b]$, i.e., the empirical distribution where probability mass $1/S$ concentrates at each of ${}_{(1)}P(a,b),\ldots,{}_{(S)}{}P(a,b)$
\STATE obtain the ML estimator $\widehat{\brho}$ as in~\eqref{eq:MLErho} for the parametric intensity of the reporting dates
%%%%\STATE get the ML estimator $\widehat{\bvartheta}$ as in~\eqref{eq:MLEvartheta1}--\eqref{eq:MLEvartheta2} for the parametric densities of the reporting delays
\STATE compute the ML estimator $\widehat{\bth}$ as in~\eqref{eq:MLEtheta} for the parametric intensities of the payments dates
\STATE calculate the ML estimator $\widehat{\boldsymbol\varsigma}$ in the same manner as in~\eqref{eq:MLEvartheta1}--\eqref{eq:MLEvartheta2} for the parametric densities of the payment amounts
\FOR[repeat in order to obtain the empirical distribution]{$s=1$ to $S$}
\STATE generate a~realization of the non-homogeneous Poisson process $\{{}_{(s)}M(t)\}_{t\geq 0}$ with intensity $\psi(t;\widehat{\brho})$ for the future time window $(a,b]$ as the arrival times $\{{}_{(s)}Z_{M(a)+1},\ldots,{}_{(s)}Z_{{}_{(s)}M(b)}\}$ representing the future reporting dates
\FOR[payments for the (old) already reported claims]{$i=1$ to $M(a)$}
\STATE generate a~realization of the non-homogeneous Poisson process $\{{}_{(s)}N_i(t)\}_{t\geq 0}$ with intensity $\lambda(t,Z_i;\widehat{\bth})$ for the time window $(a,b]$ as the arrival times $\{{}_{(s)}U_{i,N_i(a)+1},\ldots,{}_{(s)}U_{i,{}_{(s)}N_i(b)}\}$ representing the future payments dates
\STATE generate the payment amounts $\{{}_{(s)}X_{i,N_i(a)+1},\ldots,{}_{(s)}X_{i,{}_{(s)}N_i(b)}\}$ independently from the density $f_X\{\cdot;v(Z_i,\widehat{\boldsymbol{\varsigma}})\}$
\ENDFOR
\FOR[payments for the (future) new reported claims]{$i=M(a)+1$ to ${}_{(s)}M(b)$}
\STATE generate a~realization of the non-homogeneous Poisson process $\{{}_{(s)}N_i(t)\}_{t\geq 0}$ with intensity $\lambda(t,{}_{(s)}Z_i;\widehat{\bth})$ for the time window $(a,b]$ as the arrival times $\{{}_{(s)}U_{i,N_i(a)+1},\ldots,{}_{(s)}U_{i,{}_{(s)}N_i(b)}\}$ representing the future payments dates
\STATE generate the payment amounts $\{{}_{(s)}X_{i,N_i(a)+1},\ldots,{}_{(s)}X_{i,{}_{(s)}N_i(b)}\}$ independently from the density $f_X\{\cdot;v({}_{(s)}Z_i,\widehat{\boldsymbol{\varsigma}})\}$
\ENDFOR
\STATE calculate the total future payments ${}_{(s)}P(a,b)=\sum_{i=1}^{{}_{(s)}M(b)}\sum_{j=N_i(a)+1}^{{}_{(s)}N_i(b)}{}_{(s)}X_{i,j}$
\ENDFOR
\end{algorithmic}
\end{algorithm}

The predicted distribution of the forthcoming payment amounts within the next year is graphically displayed in Figure~\ref{fig:main}. Here, our micro approach is compared with the traditional method based on data aggregation. Moreover, the true cumulative amount of the one-year ahead payments is depicted in order to judge the point prediction's precision.

\begin{figure}[!ht]
\begin{center} \includegraphics[width=0.5\textwidth]{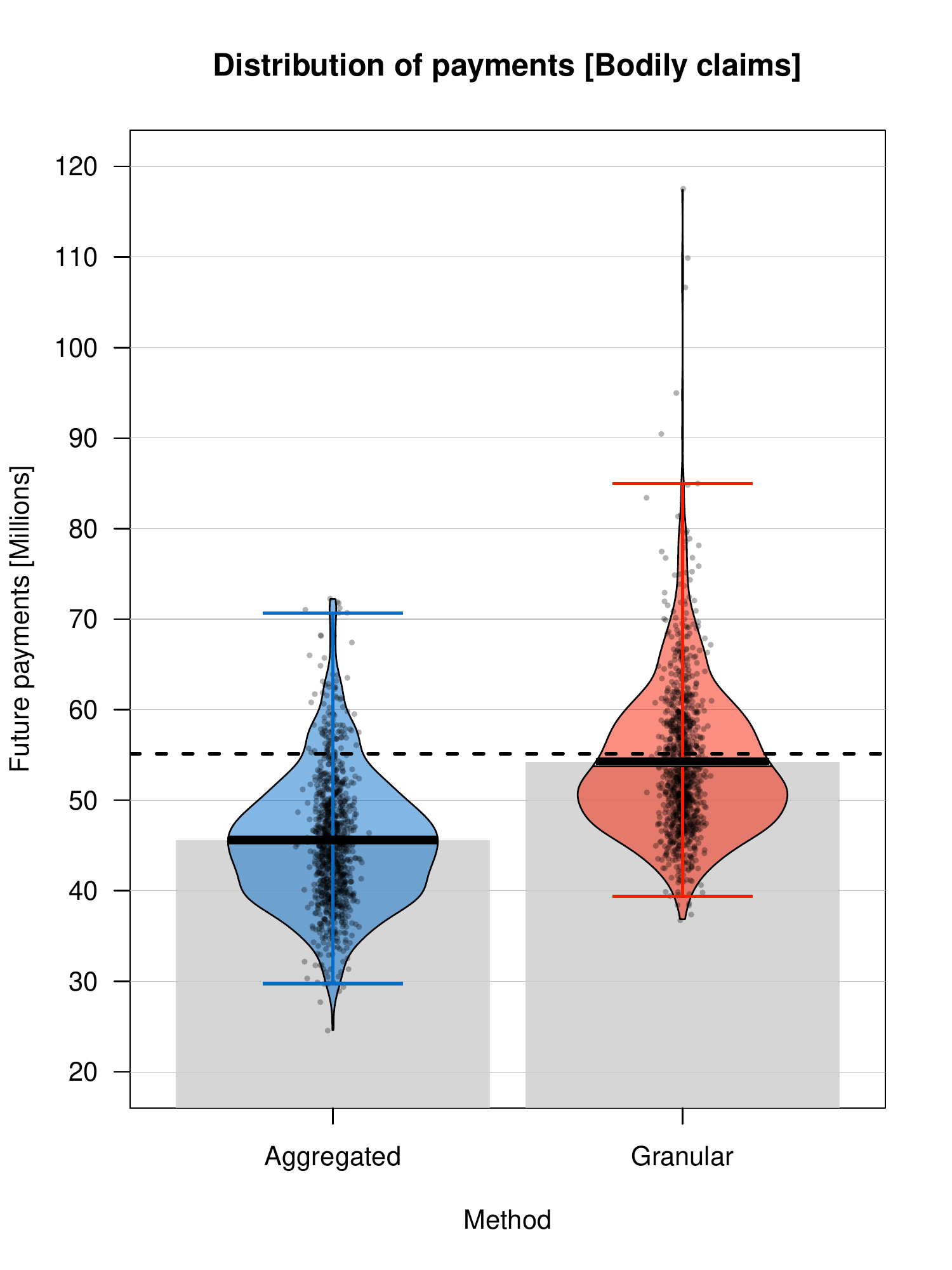}\includegraphics[width=0.5\textwidth]{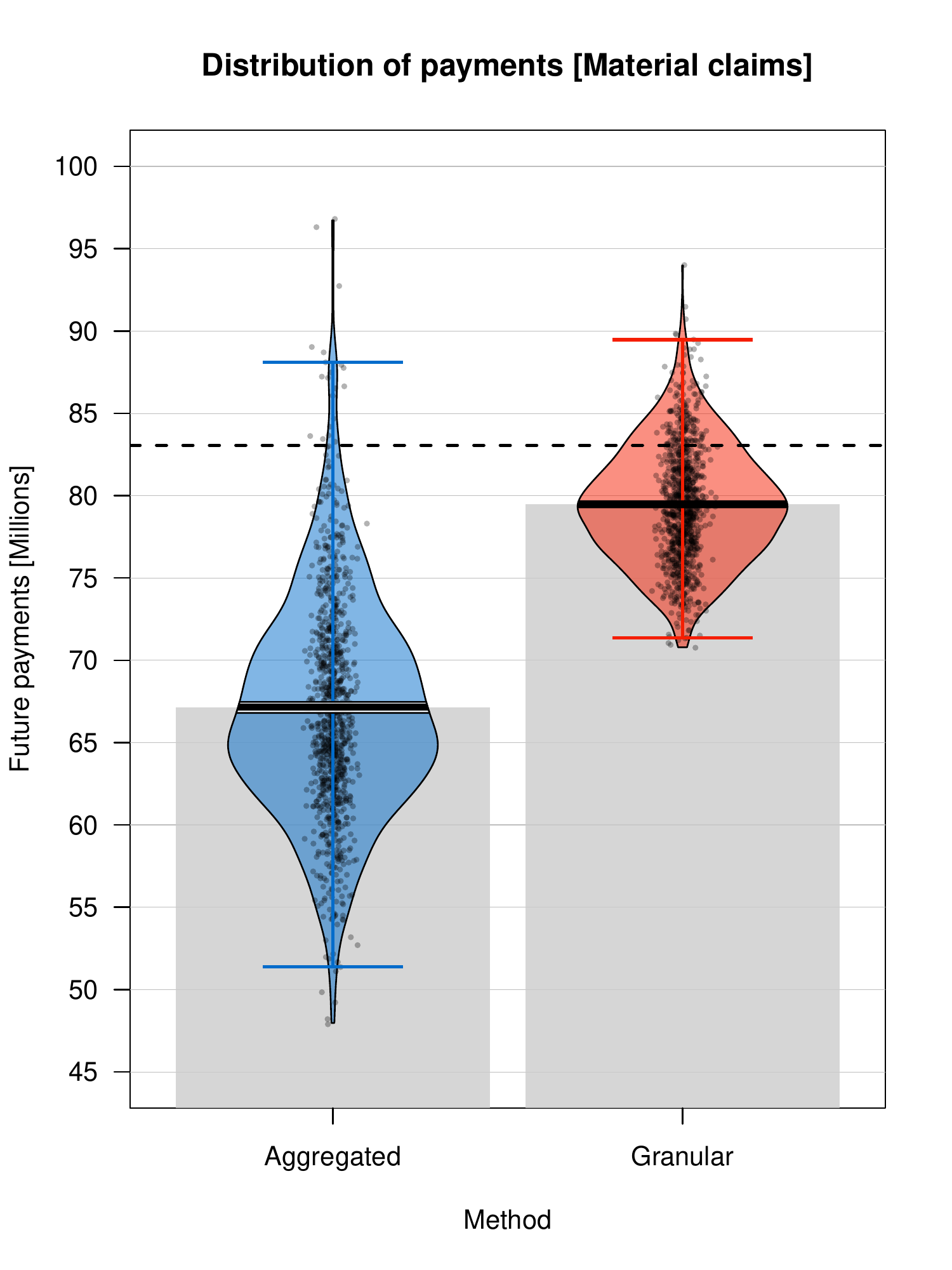}
\caption{Prediction of the distribution of the forthcoming payments for the next year (primary aim)---traditional (aggregated) method in blue, micro (granular) method in red. Bold solid horizontal line corresponds to the median of the predictive distribution. Height of the grey vertical bar corresponds to the mean of the predictive distribution. Colored solid horizontal whiskers represent the $0.5$th and the $99.5$th percentiles of the predictive distribution. Dashed horizontal line stands for the real (true) sum of payments.}
\label{fig:main}
\end{center}
\end{figure}

There are two general and, from a~practical point of view, very important findings with respect to the prediction of the future total payment amounts. First, our claim-by-claim based method is \emph{more precise} in point prediction to the `unknown' true value compared to the traditional technique based on aggregated data. And this holds for both lines of business. Second, our micro approach provides \emph{less volatile} predicted distribution, e.g., in terms of the coefficient of variation.

\subsubsection{Secondary goal: Back-prediction of the truncated occurrence times}
Our secondary practical target is to back-predict the accident dates of the claims, which are \emph{truncated} due to the reporting delay. We are indeed not aware of so-called incurred but not reported claims and the insurance company needs to back-predict these claims, which have already occurred, but have not been reported yet. This can be reached via Procedure~\ref{alg:secondary}.
\begin{algorithm}[!ht]
\caption{Estimation of the intensity of the accident dates}
\label{alg:secondary}
%\algsetup{indent=2em}
\begin{algorithmic}[1]
\REQUIRE Observations $\{T_i,Z_i\}_{i=1,\ldots,M(a)}$
\ENSURE Fitted intensity $\widehat{\mu}$ for the underlying Poisson process of the accident dates
\STATE obtain the ML estimator $\widehat{\brho}$ as in~\eqref{eq:MLErho} for the parametric intensity of the reporting dates
\STATE calculate the ML estimator $\widehat{\bvartheta}$ as in~\eqref{eq:MLEvartheta1}--\eqref{eq:MLEvartheta2} for the parametric densities of the reporting delays
\STATE get the estimator of the intensity $\mu$ as in~\eqref{eq:displacement}, i.e., by plugging-in the corresponding estimates~$\widehat{\brho}$ and~$\widehat{\bvartheta}$ and performing numerical integration $\widehat{\mu}(t)\equiv\mu(t;\widehat{\brho},\widehat{\bvartheta})=\int_{\mathbb{R}}\psi(u;\widehat{\brho})f_{W}\{t;w(u,\widehat{\bvartheta})\}\ud u$
\end{algorithmic}
\end{algorithm}

%%%%%Figures~\ref{fig:BodDensity} and~\ref{fig:MatDensity}
%%%%%
%%%%%\begin{figure}[!ht]
%%%%%\centering \includegraphics[width=0.8\textwidth]{DensityBOD.pdf}
%%%%%\caption{Prediction of distribution of the future payments (next year for the bodily injury claims) using different models for the individual payment amounts from Subsection~\ref{subsec:payment_amounts}. The true value of cumulative next year payments is depicted by the solid vertical line.\label{fig:BodDensity}}
%%%%%\end{figure}

%%%%%\begin{figure}[!ht]
%%%%%\centering \includegraphics[width=0.8\textwidth]{DensityMAT.pdf}
%%%%%\caption{Prediction of distribution of the future payments (next year for the material damage claims) using different models for the individual payment amounts from Subsection~\ref{subsec:payment_amounts}. The true value of cumulative next year payments is depicted by the solid vertical line.\label{fig:MatDensity}}
%%%%%\end{figure}

The counts of the back-predicted accident dates for the latest year as well as the counts of the predicted accident dates for the next year are visualized in Figure~\ref{fig:AccidentDates}.

\begin{figure}[!ht]
\begin{center}
\includegraphics[width=0.5\textwidth]{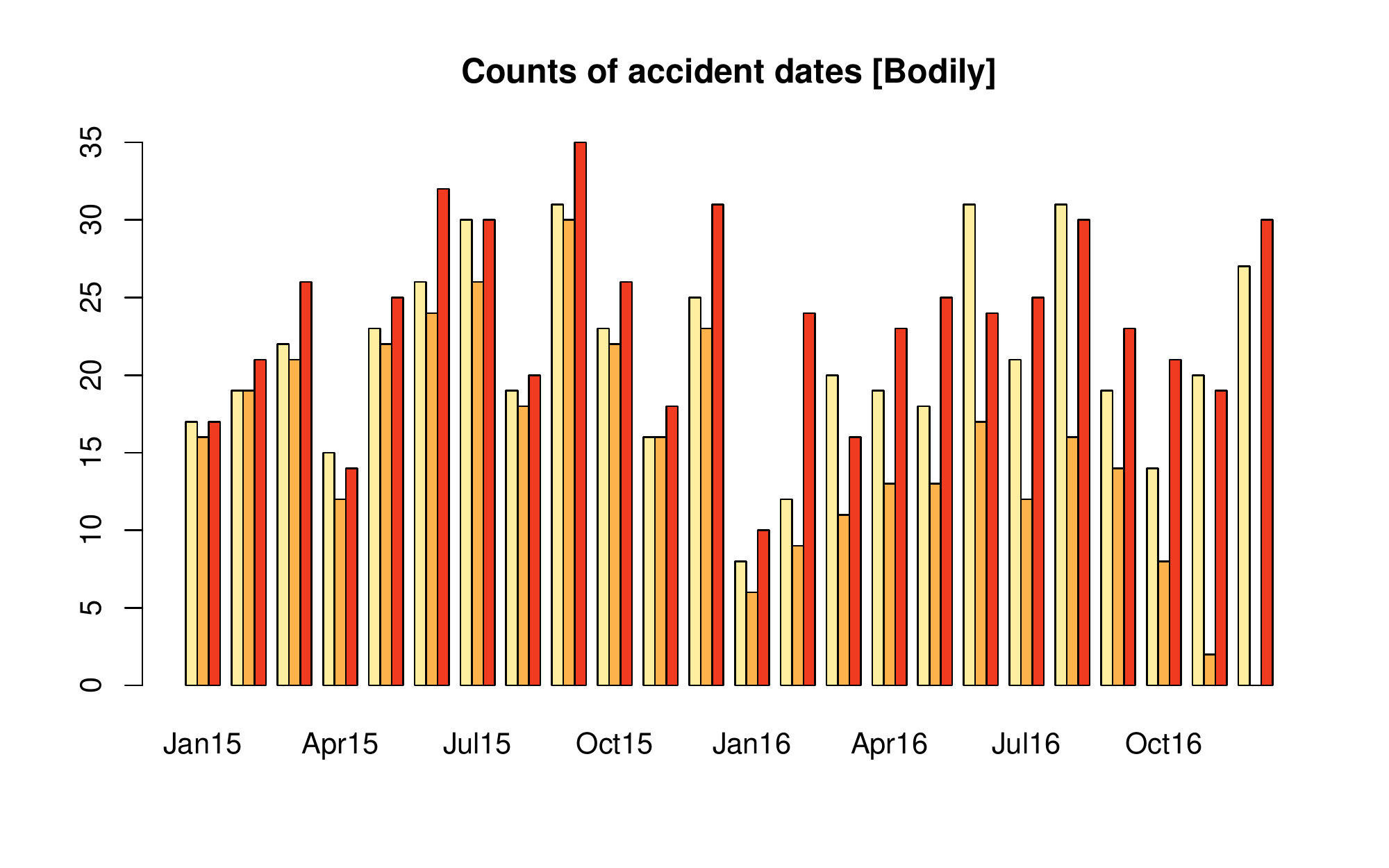}\includegraphics[width=0.5\textwidth]{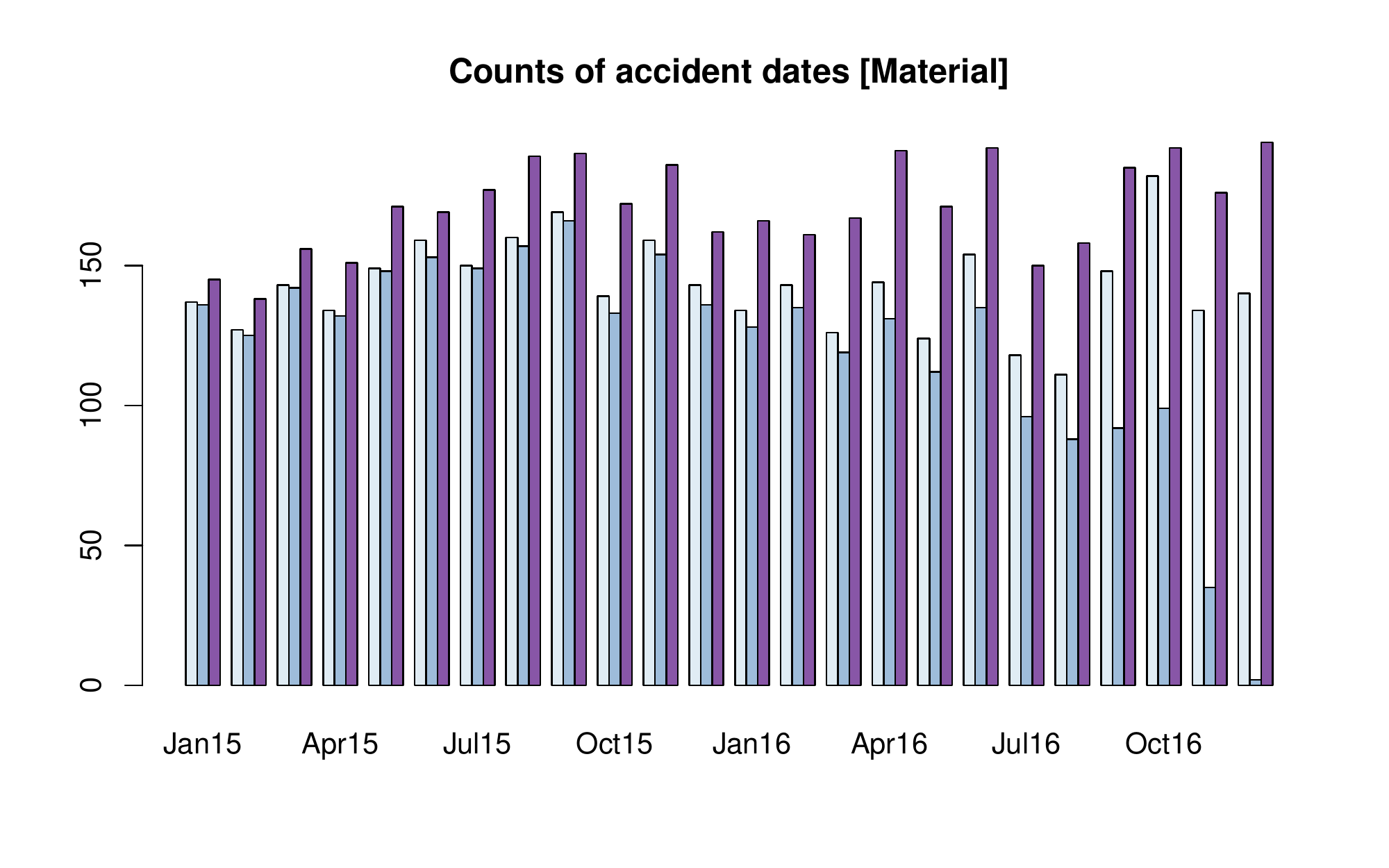}
\caption{Predicted truncated accident dates (secondary aim)---triplets of bars represent: observed counts of the accident dates from the database up to the end of 2016 (in yellow/light blue); observed counts of the accident dates from the database up to the end of 2015 (in orange/darker blue); and predicted counts of the accident dates based on the data till the end of 2015 (in red/violet).}
\label{fig:AccidentDates}
\end{center}
\end{figure}

It is of utmost importance for the insurance company to have information regarding the accidents, which have not been reported yet. Figure~\ref{fig:AccidentDates} reveals triplets of bars, that nicely accommodate the problem of truncated data in our setup. The middle horizontal bar in each triplet corresponds to the known number of accident dates based on the database up to year 2015. The left bar in the triplet stands again for the known number of accident dates, although, coming from the database up to year 2016. Therefore, the left bar has to be higher than the middle one, because additional claims can be reported within the calendar year 2016 (and they can occur in 2016 or even in previous years). The right bar in the triplet represents our prediction. It is supposed to be slightly higher even than the left bar, because there can occur additional accidents even before the end of year 2016 that are not going to be reported till the end of year 2016.

\section{Conclusions}\label{sec:conc}
\emph{Micro forecasting} is in general a~stochastic prediction method for future losses/costs relying on the individual developments of the recorded historical events. Our prediction approach is capable to model the probabilistic behavior of the future losses' occurrences, the occurrences of the incurred but not reported losses, the lengths of the reporting delays, and the frequency and severity of the loss payments in time. This is indeed sufficient for prediction of the future cash-flows for a~predetermined time horizon.

We employ the micro prediction technique in claims reserving. To meet all future insurance claims rising from policies, it is requisite to quantify the outstanding loss liabilities. Here, \emph{utility for solvency of the insurance company} is developed. And, clearly, valuation of the reserving risk in insurance is not the only area of empirical economics, where the proposed methodology can be applied as documented by several case examples.

Quantifying reserving risk in non-life insurance inadvertently yields to a~\emph{theoretical framework of the marked non-homogeneous Poisson process with non-homogeneous Poisson processes as marks}. It can be viewed as an~infinitely stochastic Poisson process and, consequently, a~proper statistical inference relying on simple and verifiable assumptions is derived.

%theoretical by-product is to prove consistency and asymptotic normality of parameter estimates for models handling n.i.n.i.d.~data with simple and verifiable assumptions on the utility function

%developed statistical inference for a~marked non-homogeneous Poisson process with non-homogeneous Poisson processes as marks ... infinitely stochastic Poisson process

%\subsection{Utility for solvency of the insurance company}
%CHANGE The \emph{traditional actuarial view} of reserve risk looks at the uncertainty in the outstanding liabilities over their lifetime. Therefore, the aim is to estimate the \emph{ultimate claims} amount. Although, \emph{Solvency~II} directive view takes a~one year perspective into account, requiring a~distribution of the liabilities after one year.

%CHANGE To summarize all practical issues discussed in previous subsection, assume that we are in time $T=t$ and we want to meet the Solvency~II requirements. If we can model the stochastic behavior of future claims' occurrences, occurrences of incurred but not reported claims, lengths of reporting delay, and the frequency and severity of the loss payments in time, then we possess everything to predict the future cash-flows in time window $(t,t+1\mbox{year}]$ by simulating from the granular loss reserving model. Moreover, one can simulate from the model using the estimated parameters and functionals many times (Monte Carlo based style) in order to obtain simulated distributions of the predictions. Loosely speaking, pushing one button will provide us stochastic predictions for all the future claim payments.

\subsection*{Acknowledgements}
The research of Mat\'{u}\v{s} Maciak and Michal Pe\v{s}ta was supported by the Czech Science Foundation project GA\v{C}R No.~18-01781Y.

%\bibliography{mop-micro}

\begin{thebibliography}{}
	
	\bibitem[Aigner et~al., 1977]{AKS1977}
	Aigner, D., Knox-Lovell, C., and Schmidt, P. (1977).
	\newblock Formulation and estimation of stochastic frontier production function
	models.
	\newblock {\em J. Econometrics}, 6(1):21--37.
	
	\bibitem[Antonio and Plat, 2014]{AntonioPlat}
	Antonio, K. and Plat, R. (2014).
	\newblock Micro-level stochastic loss reserving for general insurance.
	\newblock {\em Scand. Actuar. J.}, 2014(7):649--669.
	
	\bibitem[Arjas, 1989]{Arjas}
	Arjas, E. (1989).
	\newblock The claims reserving problem in non-life insurance: Some structural
	ideas.
	\newblock {\em {ASTIN} Bull.}, 19(2):139--152.
	
	\bibitem[Arnold, 2019]{Arnold2019}
	Arnold, C. (2019).
	\newblock Death, statistics and a~disaster zone: {T}he struggle to count the
	dead after {H}urricane {M}aria.
	\newblock {\em Nature}, 566(7742):22--25.
	
	\bibitem[Azar, 1980]{Azar1980}
	Azar, E.~E. (1980).
	\newblock The conflict and peace databank ({COPDAB}) project.
	\newblock {\em J. Conflict Resolut.}, 24(1):143--152.
	
	\bibitem[Badescu et~al., 2016]{BLT2016}
	Badescu, A.~L., Lin, X.~S., and Tang, D. (2016).
	\newblock A marked {C}ox model for the number of {IBNR} claims: {T}heory.
	\newblock {\em Insur. Math. Econ.}, 69(1):29--37.
	
	\bibitem[Basrak et~al., 2018]{BWZ2018}
	Basrak, B., Wintenberger, O., and \v{Z}ugec, P. (2018).
	\newblock On total claim amount for marked {P}oisson cluster models.
	\newblock \url{https://hal.archives-ouvertes.fr/hal-01788339}.
	
	\bibitem[Benito and L\'{o}pez-Mart\'{i}n, 2018]{BLM2018}
	Benito, S. and L\'{o}pez-Mart\'{i}n, C. (2018).
	\newblock A review of the state of the art in quantifying operational risk.
	\newblock {\em J. Oper. Risk}, 13(4):89--129.
	
	\bibitem[Billingsley, 2008]{Billingsley2008}
	Billingsley, P. (2008).
	\newblock {\em Probability and Measure}.
	\newblock Wiley, New York, NY, 3rd edition.
	
	\bibitem[Bobashev et~al., 2007]{bobashev2007}
	Bobashev, G., Goedecke, D., Yu, F., and Epstein, J. (2007).
	\newblock A hybric epidemic model: {C}ombining advantages of agent-based and
	equation-based approaches.
	\newblock In {\em Proceedings -- 2007 Winter Simulation Conference. IEEE},
	pages 1532--1537.
	
	\bibitem[Bosma et~al., 2004]{BPTW2004}
	Bosma, N., {Van~Praag}, M., Thurik, R., and {De~Witt}, G. (2004).
	\newblock The value of human and social capital investments for the business
	performance of startups.
	\newblock {\em Small Bus. Econ.}, 23(3):227--236.
	
	\bibitem[Braithwaite, 2010]{Braithwaite2010}
	Braithwaite, A. (2010).
	\newblock {MIDLOC}: Introducing the militarized interstate dispute location
	dataset.
	\newblock {\em J. Peace Res.}, 47(1):91--98.
	
	\bibitem[Burda et~al., 2012]{BURDA2012}
	Burda, M., Harding, M., and Hausman, J. (2012).
	\newblock A~{P}oisson mixture model of discrete choice.
	\newblock {\em J. Econometrics}, 166(2):184--203.
	
	\bibitem[Caldbick et~al., 2015]{Caldbick2015}
	Caldbick, S., Wu, X., Lynch, T., Al-Khatib, N., Andkhoie, M., and Farag, M.
	(2015).
	\newblock The financial burden of out of pocket prescription drug expenses in
	canada.
	\newblock {\em Int. J. Health Econ. Ma.}, 15(3):329--338.
	
	\bibitem[Chernobai et~al., 2007]{OR2007}
	Chernobai, A.~S., Rachev, S.~T., and Fabozzi, F.~J. (2007).
	\newblock {\em Operational Risk: {A} guide to {B}asel {II} Capital
		Requirements, Models and Analysis}.
	\newblock Wiley finance, New York, NY.
	
	\bibitem[Clauset, 2018]{Clauset2014}
	Clauset, A. (2018).
	\newblock Trends and fluctuations in the severity of interstate wars.
	\newblock {\em Science Advances}, 4:1--9.
	
	\bibitem[Coeurjolly and M{\o}ller, 2014]{coeurjolly2014}
	Coeurjolly, J.-F. and M{\o}ller, J. (2014).
	\newblock Variational approach for spatial point process intensity estimation.
	\newblock {\em Bernoulli}, 20(3):1097--1125.
	
	\bibitem[Cohen, 2018]{Cohen2018}
	Cohen, R.~D. (2018).
	\newblock An operational risk capital model based on the loss distribution
	approach.
	\newblock {\em J. Oper. Risk}, 13(2):69--81.
	
	\bibitem[Collier et~al., 2004]{CollierHoefflerSoderbom2004}
	Collier, P., Hoeffler, A., and S\"oderbom, M. (2004).
	\newblock On the duration of civil war.
	\newblock {\em J. Peace Res.}, 41(3):253--273.
	
	\bibitem[Colombo and Grilli, 2008]{CG2008}
	Colombo, M. and Grilli, L. (2008).
	\newblock Start-up size: {T}he role of external financing.
	\newblock {\em Econ. Lett.}, 88(1):243--250.
	
	\bibitem[Cressey, 2008]{Cressey2008}
	Cressey, D. (2008).
	\newblock War survey points to millions more dead.
	\newblock {\em Nature News}.
	\newblock doi:10.1038/news.2008.901.
	
	\bibitem[Crowder et~al., 1991]{CKSS1991}
	Crowder, M.~J., Kimber, A.~C., Smith, R.~L., and Sweeting, T.~J. (1991).
	\newblock {\em Statistical Analysis of Reliability Data}.
	\newblock Chapman and Hall, Malta, 1st edition.
	
	\bibitem[Ekedahl et~al., 1995]{ekedahl}
	Ekedahl, A., Andersson, S.~I., Hovelius, B., M\"{o}lstad, S., Liedholm, H., and
	Melander, A. (1995).
	\newblock Drug prescription attitudes and behaviour of general practitioners.
	\newblock {\em Eur. J. Clinical Pharmacol.}, 47(5):381--387.
	
	\bibitem[England and Verrall, 2002]{EV2002}
	England, P. and Verrall, R. (2002).
	\newblock Stochastic claims reserving in general insurance (with discussion).
	\newblock {\em British Actuarial Journal}, 8(3):443--518.
	
	\bibitem[England and Verrall, 1999]{EV1999}
	England, P.~D. and Verrall, R.~J. (1999).
	\newblock Analytic and bootstrap estimates of prediction errors in claims
	reserving.
	\newblock {\em Insur. Math. Econ.}, 25(3):281--293.
	
	\bibitem[Engle, 2000]{Engle2000}
	Engle, R.~F. (2000).
	\newblock The econometrics of ultra-high-frequency data.
	\newblock {\em Econometrica}, 68(1):1--22.
	
	\bibitem[Giesecke and Schwenkler, 2018]{GIESECKE2018}
	Giesecke, K. and Schwenkler, G. (2018).
	\newblock Filtered likelihood for point processes.
	\newblock {\em J. Econometrics}, 204(1):33--53.
	
	\bibitem[Gleditsch et~al., 2014]{GleditschMetternichRuggeri2014}
	Gleditsch, K.~S., Metternich, N.~W., and Ruggeri, A. (2014).
	\newblock Data and progress in peace and conflict research.
	\newblock {\em J. Peace Res.}, 51(2):301--314.
	
	\bibitem[Godecharle and Antonio, 2015]{GodecharleAntonio}
	Godecharle, E. and Antonio, K. (2015).
	\newblock Reserving by conditioning on markers of individual claims: {A}~case
	study using historical simulation.
	\newblock {\em North American Actuarial Journal}, 19(4):273--288.
	
	\bibitem[Goldfarb and Tucker, 2011]{GT2011}
	Goldfarb, A. and Tucker, C. (2011).
	\newblock Online display advertising: Targeting and obtrusiveness.
	\newblock {\em Market. Sci.}, 30(3):389--404.
	
	\bibitem[G\"on\"ul et~al., 2001]{gonul}
	G\"on\"ul, F.~F., Carter, F., Peterova, E., and Srinivasan, K. (2001).
	\newblock Promotion of prescription drugs and its behavior on physicians'
	choice behavior.
	\newblock {\em J. Marketing}, 65:79--90.
	
	\bibitem[Haastrup and Arjas, 1996]{HaastrupArjas}
	Haastrup, S. and Arjas, E. (1996).
	\newblock Claims reserving in continuous time: A nonparametric {B}ayesian
	approach.
	\newblock {\em {ASTIN} Bull.}, 26(2):139--164.
	
	\bibitem[Hallberg, 2012]{Hallberg2012}
	Hallberg, J.~D. (2012).
	\newblock {PRIO} conflict site 1989--2008: {A}~geo-referenced dataset on armed
	conflict.
	\newblock {\em Conflict Manag. Peace}, 29(2):219--232.
	
	\bibitem[Hansen and Scheinkman, 2009]{HS2009}
	Hansen, L.~P. and Scheinkman, J.~A. (2009).
	\newblock Long term risk: {A}n operator approach.
	\newblock {\em Econometrica}, 77(1):177--234.
	
	\bibitem[Harrison and Wolf, 2012]{HarrisonWolf2012}
	Harrison, M. and Wolf, N. (2012).
	\newblock The frequency of wars.
	\newblock {\em Econ. Hist. Rev.}, 65(3):1055--1076.
	
	\bibitem[Hesselager, 1994]{Hesselager1994}
	Hesselager, O. (1994).
	\newblock A {M}arkov model for loss reserving.
	\newblock {\em ASTIN Bull.}, 24(2):183--193.
	
	\bibitem[Hjort and Pollard, 2011]{HjortPollard2011}
	Hjort, N.~L. and Pollard, D. (2011).
	\newblock Asymptotics for minimisers of convex processes.
	\newblock \url{https://arxiv.org/abs/1107.3806}.
	
	\bibitem[Holtz-Eakin et~al., 1994]{HJR1994}
	Holtz-Eakin, D., Joulfaian, D., and Rosen, H. (1994).
	\newblock Entrepreneurial decisions and liquidity constraints.
	\newblock {\em Rand J. Econ.}, 25:334--347.
	
	\bibitem[Hudecov\'{a} and Pe\v{s}ta, 2013]{PH2013}
	Hudecov\'{a}, {\v{S}}. and Pe\v{s}ta, M. (2013).
	\newblock Modeling dependencies in claims reserving with {GEE}.
	\newblock {\em Insur. Math. Econ.}, 53(3):786--794.
	
	\bibitem[Hudecov\'{a} et~al., 2017]{HPH2017}
	Hudecov\'{a}, {\v{S}}., Pe\v{s}ta, M., and Hlubinka, D. (2017).
	\newblock Modelling prescription behaviour of general practitioners.
	\newblock {\em Math. Slovaca}, 67(1):1--17.
	
	\bibitem[Jewell, 1989]{Jewell1989}
	Jewell, W. (1989).
	\newblock Predicting {IBNYR} events and delays, part {I} continuous time.
	\newblock {\em {ASTIN} Bull.}, 19:25--56.
	
	\bibitem[Jewell, 1990]{Jewell1990}
	Jewell, W. (1990).
	\newblock Predicting {IBNYR} events and delays, part {II} discrete time.
	\newblock {\em {ASTIN} Bull.}, 20:93--111.
	
	\bibitem[Kazan, 2015]{Kazan2015}
	Kazan, E. (2015).
	\newblock The innovative capabilities of digital payment platforms:
	{A}~comparative study of {A}pple {P}ay \& {G}oogle {W}allet.
	\newblock In {\em 2015 International Conference on Mobile Business}.
	\newblock Paper~4, \url{http://aisel.aisnet.org/icmb2015/4}.
	
	\bibitem[Kermack and McKendrick, 1927]{Kermack27}
	Kermack, W. and McKendrick, A. (1927).
	\newblock A~contribution to the mathematical theory of epidemics.
	\newblock {\em Proceedings of the Royal Society~A: Mathematical, Physical and
		Engineering Sciences}, 115(772):700--721.
	
	\bibitem[Kingman, 1993]{Kingman1993}
	Kingman, J. F.~C. (1993).
	\newblock {\em Poisson Processes}.
	\newblock Oxford University Press, New York, NY.
	
	\bibitem[Konecny, 1987]{Konecny1987}
	Konecny, F. (1987).
	\newblock The asymptotic properties of maximum likelihood estimators for marked
	{P}oisson processes with a cyclic intensity measure.
	\newblock {\em Metrika}, 34:143--155.
	
	\bibitem[Larsen, 2007]{Larsen}
	Larsen, C. (2007).
	\newblock An individual claims reserving model.
	\newblock {\em {ASTIN} Bull.}, 37(1):113--132.
	
	\bibitem[Lawless, 1987]{Fawless1987}
	Lawless, J.~F. (1987).
	\newblock Regression methods for {P}oisson process data.
	\newblock {\em J. Am. Stat. Assoc.}, 82(399):808--815.
	
	\bibitem[Lewis and Shedler, 1979]{LS1979}
	Lewis, P. A.~W. and Shedler, G.~S. (1979).
	\newblock Simulation of nonhomogeneous {P}oisson processes by thinning.
	\newblock {\em Naval Research Logistics}, 26(3):403--413.
	
	\bibitem[Liaukonyte et~al., 2015]{LTW2015}
	Liaukonyte, J., Teixeira, T., and Wilbur, K.~C. (2015).
	\newblock Television advertising and online shopping.
	\newblock {\em Market. Sci.}, 34(3):311--330.
	
	\bibitem[Linda and Allen, 2008]{LindaAllen2008}
	Linda, J. and Allen, S. (2008).
	\newblock An introduction to stochastic epidemic models.
	\newblock In {\em Mathematical Epidemiology}, chapter~3, pages 81--120.
	Springer-Verlag.
	
	\bibitem[Liu et~al., 2015]{LIU2015372}
	Liu, J., Kauffman, R.~J., and Ma, D. (2015).
	\newblock Competition, cooperation, and regulation: {U}nderstanding the
	evolution of the mobile payments technology ecosystem.
	\newblock {\em Electron. Commer. R. A.}, 14(5):372--391.
	
	\bibitem[Liu and Mattila, 2019]{LIU2019268}
	Liu, S.~Q. and Mattila, A.~S. (2019).
	\newblock Apple {P}ay: {C}oolness and embarrassment in the service encounter.
	\newblock {\em Int. J. Hosp. Manag.}, 78:268--275.
	
	\bibitem[Manchanda and Chintangunta, 2004]{manchanda}
	Manchanda, P. and Chintangunta, P.~K. (2004).
	\newblock Responsiveness of physician prescription behavior of salesforce
	effort: {A}n individual level analysis.
	\newblock {\em Market. Lett.}, 15:129--145.
	
	\bibitem[Miller and Lufi, 1994]{miller}
	Miller, R.~H. and Lufi, H.~S. (1994).
	\newblock Managed care plan performance since 1980: {A} literature analysis.
	\newblock {\em JAMA -- J. Am. Med. Assoc.}, 271(19):1512--1519.
	
	\bibitem[Norberg, 1993]{Norberg}
	Norberg, R. (1993).
	\newblock Prediction of outstanding liabilities in non-life insurance.
	\newblock {\em {ASTIN} Bull.}, 23(1):95--115.
	
	\bibitem[Norberg, 1999]{Norberg1999}
	Norberg, R. (1999).
	\newblock Prediction of outstanding liabilities {II}. {M}odel variations and
	extensions.
	\newblock {\em {ASTIN} Bull.}, 29(1):5--25.
	
	\bibitem[Pe\v{s}ta and Okhrin, 2014]{PO2014}
	Pe\v{s}ta, M. and Okhrin, O. (2014).
	\newblock Conditional least squares and copulae in claims reserving for
	a~single line of business.
	\newblock {\em Insur. Math. Econ.}, 56(1):28--37.
	
	\bibitem[Pigeon et~al., 2014]{PigeonAntonioDenuit}
	Pigeon, M., Antonio, K., and Denuit, M. (2014).
	\newblock Individual loss reserving using paid-incurred data.
	\newblock {\em Insur. Math. Econ.}, 58:121--131.
	
	\bibitem[Proke{\v{s}}ov{\'a} et~al., 2017]{PDV2017}
	Proke{\v{s}}ov{\'a}, M., Dvo{\v{r}}{\'a}k, J., and Jensen, E. B.~V. (2017).
	\newblock Two-step estimation procedures for inhomogeneous shot-noise {C}ox
	processes.
	\newblock {\em Ann. I. Stat. Math.}, 69(3):513--542.
	
	\bibitem[Rizoiu et~al., 2018]{SIRHawkes2018}
	Rizoiu, M.-A., Mishra, S., Kong, Q., Carman, M., and Xie, L. (2018).
	\newblock {SIW-Hawkes}: {L}inking epidemic models and {H}awkes processes to
	model diffusions in finite populations.
	\newblock Technical report, arXiv:1711.01679v3.
	
	\bibitem[Rokstad et~al., 1997]{rokstad}
	Rokstad, K., Straand, J., and Fugelli, P. (1997).
	\newblock General practitioners' drug prescribing practice and diagnoses for
	prescribing: {T}he {M}{\o}re \& {R}omsdal prescription study.
	\newblock {\em J. Clin. Epidemiol.}, 50(4):485--494.
	
	\bibitem[Schoenberg, 2005]{SCHOENBERG2005}
	Schoenberg, F.~P. (2005).
	\newblock Consistent parametric estimation of the intensity of a
	spatial-temporal point process.
	\newblock {\em J. Stat. Plan. Infer.}, 128(1):79--93.
	
	\bibitem[Schrodt, 2014]{Schrodt2014}
	Schrodt, P. (2014).
	\newblock Seven deadly sins of contemporary quantitative political analysis.
	\newblock {\em J. Peace Res.}, 51(2):287--300.
	
	\bibitem[Taylor et~al., 2008]{TaylorMcGuireSullivan}
	Taylor, G., Mc{G}uire, G., and Sullivan, J. (2008).
	\newblock Individual claim loss reserving conditioned by case estimates.
	\newblock {\em Annals of Actuarial Science}, 3(1--2):215--256.
	
	\bibitem[Verrall and {W\"uthrich}, 2016]{VW2016}
	Verrall, R.~J. and {W\"uthrich}, M.~V. (2016).
	\newblock Understanding reporting delay in general insurance.
	\newblock {\em Risks}, 4(3):25.
	
	\bibitem[Waagepetersen and Guan, 2009]{WG2009}
	Waagepetersen, R. and Guan, Y. (2009).
	\newblock Two-step estimation for inhomogeneous spatial point processes.
	\newblock {\em J. Roy. Stat. Soc. B. Met.}, 71(3):685--702.
	
	\bibitem[Waagepetersen, 2007]{Waag2007}
	Waagepetersen, R.~P. (2007).
	\newblock An estimating function approach to inference for inhomogeneous
	{N}eyman--{S}cott processes.
	\newblock {\em Biometrics}, 63(1):252--258.
	
	\bibitem[Ward et~al., 2010]{BakkeGreenhillWard2010}
	Ward, M.~D., Greenhill, B.~D., and Bakke, K.~M. (2010).
	\newblock The perils of policy by p-value: {P}redicting civil conflict.
	\newblock {\em J. Peace Res.}, 45(5):363--375.
	
	\bibitem[Watkins et~al., 2003]{watkins}
	Watkins, C., Harvey, I., Carthy, P., Moore, L., Robinson, E., and Brawn, R.
	(2003).
	\newblock Attitudes and behaviour of general practitioners and their
	prescribing costs: {A}~national cross sectional survey.
	\newblock {\em Qual. and Saf. Health Care}, 12:29--34.
	
	\bibitem[Weisberg et~al., 1984]{WTC1984}
	Weisberg, H.~I., Tomberlin, T.~J., and Chatterjee, S. (1984).
	\newblock Predicting insurance losses under cross-classification: A~comparison
	of alternative approaches.
	\newblock {\em J. Bus. Econ. Stat.}, 2(2):170--178.
	
	\bibitem[White, 1982]{White1982}
	White, H. (1982).
	\newblock Maximum likelihood estimation of misspecified models.
	\newblock {\em Econometrica}, 50(1):1--25.
	
	\bibitem[W{\"u}thrich, 2016]{W2016}
	W{\"u}thrich, M. (2016).
	\newblock Machine learning in individual claims reserving.
	\newblock Swiss Finance Institute Research Paper No.~16--67.
	\url{https://ssrn.com/abstract=2867897}.
	
	\bibitem[W{\"u}thrich and Merz, 2008]{wutrich_kniha}
	W{\"u}thrich, M. and Merz, M. (2008).
	\newblock {\em Stochastic claims reserving methods in insurance}.
	\newblock Wiley finance series. John Wiley \& Sons.
	
	\bibitem[Xiao, 2018]{XIAO2018}
	Xiao, R. (2018).
	\newblock Identification and estimation of incomplete information games with
	multiple equilibria.
	\newblock {\em J. Econometrics}, 203(2):328--343.
	
	\bibitem[Yan, 2008]{PingYan2008}
	Yan, P. (2008).
	\newblock Distribution theory, stochastic processes and infectious disease
	modelling.
	\newblock In {\em Mathematical Epidemiology}, chapter~10, pages 229--293.
	Springer-Verlag.
	
	\bibitem[Zhao and Zhou, 2010]{ZhaoZhou}
	Zhao, X. and Zhou, X. (2010).
	\newblock Applying copula models to individual claim loss reserving methods.
	\newblock {\em Insur. Math. Econ.}, 46:290--299.
	
	\bibitem[Zhao et~al., 2009]{ZhaoZhouWang}
	Zhao, X., Zhou, X., and Wang, J. (2009).
	\newblock Semiparametric model for prediction of individual claim loss
	reserving.
	\newblock {\em Insur. Math. Econ.}, 45:1--8.
	
\end{thebibliography}

%\newpage

\appendix

\section{Proofs}
%{\small
\begin{proof}[Proof of Theorem~\ref{thm:consistency}]
Let us choose $t>0$. With respect to Assumption~\ref{ass:convex}, consider the convex function
\begin{equation*}
H_t(\bss):=\sum_{i=1}^{M(t)}\left[h\left\{Z_i;\brho_0+\mathcal{I}^{-1/2}(t,\brho_0)\bss,t\right\}-h(Z_i;\brho_0,t)\right]
\end{equation*}
in $\bss\in\mathbb{R}^{q}$. It is minimized by $\mathcal{I}^{1/2}(t,\brho_0)\left(\widehat{\brho}-\brho_0\right)$. The Taylor series expansion gives
\begin{multline}\label{eq:repre}
H_t(\bss)=\bss^{\top}\underbrace{\mathcal{I}^{-1/2}(t,\brho_0)\sum_{i=1}^{M(t)}\partial_{\brho_0}h\left(Z_i;\brho,t\right)}_{=:\bU(t)}\\
+\frac{1}{2}\bss^{\top}\underbrace{\mathcal{I}^{-1/2}(t,\brho_0)\left\{\sum_{i=1}^{M(t)}\partial^2_{\brho_0}h(Z_i;\brho,t)\right\}\mathcal{I}^{-1/2}(t,\brho_0)}_{=:\bV(t)}\bss+r_t(\bss)
\end{multline}
almost surely, where $r_t(\bss)=M(t)o\!\left\{\bss^{\top}\mathcal{I}^{-1}(t,\brho_0)\bss\right\}\to 0$ in probability for $t\to\infty$, because of Assumption~\ref{ass:infmat}~{(i)}.

Assumption~\ref{ass:interchange1} assures that there is a~closed sub-neighborhood of~$\brho_0$ denoted by~$\mathscr{U}_t(\brho_0)$ such that for all $\brho\in\mathscr{U}_t(\brho_0)$, it holds that $\int_0^t\left|\partial_{\brho_0,j}\psi(z;\brho)\right|\ud z<\infty$ for all $j=1,\ldots,q$ and one can interchange derivative and integral, i.e., $\partial_{\brho_0}\int_0^t\psi(z;\brho)\ud z=\int_0^{t}\partial_{\brho_0}\psi(z;\brho)\ud z$.

Let us realize that the sequence $\{Z_i\}_{i\in\mathbb{N}}$ forms arrival times of the Poisson counting process $\{M(t)\}_{t\geq 0}$ and, hence,
\[
\sum_{i=1}^{M(t)}\frac{\partial_{\brho_0}\psi(Z_i;\brho)}{\psi(Z_i;\brho_0)}=\int_0^{t}\frac{\partial_{\brho_0}\psi(z;\brho)}{\psi(z;\brho_0)}\ud M(z).
\]
Recall that $h(Z_i;\brho,t)=\frac{1}{M(t)}\Psi(t;\brho)-\log \psi(Z_i;\brho)$. Since $\frac{\partial_{\brho_0}\psi(\cdot;\brho)}{\psi(\cdot;\brho_0)}$ is continuous, we obtain
\begin{align*}
\E \bU(t)&=\E\left[\mathcal{I}^{-1/2}(t,\brho_0)\sum_{i=1}^{M(t)}\partial_{\brho_0}\left\{\frac{1}{M(t)}\Psi(t;\brho)-\log \psi(Z_i;\brho)\right\}\right]\\
&=\mathcal{I}^{-1/2}(t,\brho_0)\E\left\{\partial_{\brho_0}\int_0^t\psi(z;\brho)\ud z-\sum_{i=1}^{M(t)}\frac{\partial_{\brho_0}\psi(Z_i;\brho)}{\psi(Z_i;\brho_0)}\right\}\\
&=\mathcal{I}^{-1/2}(t,\brho_0)\E\left\{\partial_{\brho_0}\int_0^t\psi(z;\brho)\ud z-\int_0^{t}\frac{\partial_{\brho_0}\psi(z;\brho)}{\psi(z;\brho_0)}\ud M(z)\right\}\\
&=\mathcal{I}^{-1/2}(t,\brho_0)\left\{\partial_{\brho_0}\int_0^t\psi(z;\brho)\ud z-\int_0^{t}\frac{\partial_{\brho_0}\psi(z;\brho)}{\psi(z;\brho_0)}\psi(z;\brho_0)\ud z\right\}=\zero.
\end{align*}
One can apply the It\^{o} isometry for jump processes
\begin{align}
\Var\bU(t)&=\E \{\bU(t)\}^{\otimes 2}=\E\left[\mathcal{I}^{-1/2}(t,\brho_0)\sum_{i=1}^{M(t)}\partial_{\brho_0}\left\{\frac{1}{M(t)}\Psi(t;\brho)-\log \psi(Z_i;\brho)\right\}\right]^{\otimes 2}\nonumber\\
&=\mathcal{I}^{-1/2}(t,\brho_0)\E\left\{\partial_{\brho_0}\int_0^t\psi(z;\brho)\ud z-\int_0^{t}\frac{\partial_{\brho_0}\psi(z;\brho)}{\psi(z;\brho_0)}\ud M(z)\right\}^{\otimes 2}\mathcal{I}^{-1/2}(t,\brho_0)\nonumber\\
&=\mathcal{I}^{-1/2}(t,\brho_0)\E\Bigg[\left\{\partial_{\brho_0}\int_0^t\psi(z;\brho)\ud z\right\}^{\otimes 2}+\left\{\int_0^{t}\frac{\partial_{\brho_0}\psi(z;\brho)}{\psi(z;\brho_0)}\ud M(z)\right\}^{\otimes 2}\Bigg.\nonumber\\
&\qquad\Bigg.-\left\{\partial_{\brho_0}\int_0^t\psi(z;\brho)\ud z\right\}\left\{\int_0^{t}\frac{\partial_{\brho_0}\psi(z;\brho)}{\psi(z;\brho_0)}\ud M(z)\right\}^{\top}\Bigg.\nonumber\\
&\qquad\Bigg.-\left\{\int_0^{t}\frac{\partial_{\brho_0}\psi(z;\brho)}{\psi(z;\brho_0)}\ud M(z)\right\}\left\{\partial_{\brho_0}\int_0^t\psi(z;\brho)\ud z\right\}^{\top}\Bigg]\mathcal{I}^{-1/2}(t,\brho_0)\nonumber\\
&=\mathcal{I}^{-1/2}(t,\brho_0)\Bigg[\int_0^{t}\frac{\left\{\partial_{\brho_0}\psi(z;\brho)\right\}^{\otimes 2}}{\psi^2(z;\brho_0)}\psi(z;\brho_0)\ud z\Bigg]\mathcal{I}^{-1/2}(t,\brho_0)\nonumber\\
&=\mathcal{I}^{-1/2}(t,\brho_0)\mathcal{I}(t;\brho_0)\mathcal{I}^{-1/2}(t,\brho_0)={\boldsymbol I},\label{eq:VarUt}
\end{align}
due to Assumptions~\ref{ass:interchange1}. Moreover,
\begin{align}
\E \bV(t)&=\E\left[\mathcal{I}^{-1/2}(t,\brho_0)\sum_{i=1}^{M(t)}\partial_{\brho_0}^2\left\{\frac{1}{M(t)}\Psi(t;\brho)-\log \psi(Z_i;\brho)\right\}\mathcal{I}^{-1/2}(t,\brho_0)\right]\nonumber\\
&=\mathcal{I}^{-1/2}(t,\brho_0)\E\Bigg[\partial_{\brho_0}^2\int_0^t\psi(z;\brho)\ud z\Bigg.\nonumber\\
&\quad\Bigg.-\sum_{i=1}^{M(t)}\left(\frac{\partial_{\brho_0}^2\psi(Z_i;\brho)}{\psi(Z_i;\brho_0)}-\frac{\left\{\partial_{\brho_0}\psi(Z_i;\brho)\right\}^{\otimes 2}}{\psi^2(Z_i;\brho_0)}\right)\Bigg]\mathcal{I}^{-1/2}(t,\brho_0)\nonumber\\
&=\mathcal{I}^{-1/2}(t,\brho_0)\E\Bigg[\partial_{\brho_0}^2\int_0^t\psi(z;\brho)\ud z\Bigg.\nonumber\\
&\quad\Bigg.-\int_0^{t}\left(\frac{\partial_{\brho_0}^2\psi(z;\brho)}{\psi(z;\brho_0)}-\frac{\left\{\partial_{\brho_0}\psi(z;\brho)\right\}^{\otimes 2}}{\psi^2(z;\brho_0)}\right)\ud M(z)\Bigg]\mathcal{I}^{-1/2}(t,\brho_0)\nonumber\\
&=\mathcal{I}^{-1/2}(t,\brho_0)\Bigg[\partial_{\brho_0}^2\int_0^t\psi(z;\brho)\ud z\Bigg.\nonumber\\
&\quad\Bigg.-\int_0^{t}\left(\frac{\partial_{\brho_0}^2\psi(z;\brho)}{\psi(z;\brho_0)}-\frac{\left\{\partial_{\brho_0}\psi(z;\brho)\right\}^{\otimes 2}}{\psi^2(z;\brho_0)}\right)\psi(z;\brho_0)\ud z\Bigg]\mathcal{I}^{-1/2}(t,\brho_0)\nonumber\\
&=\mathcal{I}^{-1/2}(t,\brho_0)\int_0^t\frac{\left\{\partial_{\brho_0}\psi(z;\brho)\right\}^{\otimes 2}}{\psi(z;\brho_0)}\ud z\,\mathcal{I}^{-1/2}(t,\brho_0)\nonumber\\
&=\mathcal{I}^{-1/2}(t,\brho_0)\mathcal{I}(t;\brho_0)\mathcal{I}^{-1/2}(t,\brho_0)={\boldsymbol I},\label{eq:EVt}
\end{align}
because of Assumption~\ref{ass:interchange1}. Furthermore for every $\bss\in\mathbb{R}^q$, it holds that
\[
\Var\left\{\bss^{\top}\bV(t)\bss\right\}=\bss^{\top}\Var\left\{\bV(t)\bss\right\}\bss=\bss^{\top}\left[\E\left\{\bV(t)\bss\right\}^{\otimes 2}-\left\{\E\bV(t)\bss\right\}^{\otimes 2}\right]\bss.
\]
The $(j,k)$-element of $\E\left\{\bV(t)\bss\right\}^{\otimes 2}\equiv\E\left\{\bV(t)\bss\bss^{\top}\bV(t)\right\}$ has a~form of
\begin{multline}\label{eq:jkelement}
\E\sum_{\ell=1}^q\sum_{m=1}^q s_{\ell}s_m\left(\bV(t)\right)_{j,\ell}\left(\bV(t)\right)_{m,k}=\E\sum_{\ell=1}^q\sum_{m=1}^q s_{\ell}s_m\\
\times\sum_{\tilde{\ell}=1}^q\sum_{\breve{\ell}=1}^q\kappa_{j,\tilde{\ell}}(t)\left(\sum_{i=1}^{M(t)}\partial^2_{\brho_0}h(Z_i;\brho,t)\right)_{\tilde{\ell},\breve{\ell}}\kappa_{\breve{\ell},\ell}(t)
\sum_{\tilde{m}=1}^q\sum_{\breve{m}=1}^q\kappa_{m,\tilde{m}}(t)\left(\sum_{i=1}^{M(t)}\partial^2_{\brho_0}h(Z_i;\brho,t)\right)_{\tilde{m},\breve{m}}\kappa_{\breve{m},k}(t),
\end{multline}
where $\bss\bss^{\top}=\left(s_{\ell}s_m\right)_{\ell=1,m=1}^{q,q}$ and $\mathcal{I}^{-1/2}(t,\brho_0)=:\left(\kappa_{\ell,m}(t)\right)_{\ell=1,m=1}^{q,q}$. Let us calculate
\begin{align}
&\E\left\{\left(\sum_{i=1}^{M(t)}\partial^2_{\brho_0}h(Z_i;\brho,t)\right)_{j,\ell}\left(\sum_{i=1}^{M(t)}\partial^2_{\brho_0}h(Z_i;\brho,t)\right)_{m,k}\right\}\nonumber\\
&=\E\Bigg[\Bigg(\partial_{\brho_0}^2\int_0^t\psi(z;\brho)\ud z-\sum_{i=1}^{M(t)}\left\{\frac{\partial_{\brho_0}^2\psi(Z_i;\brho)}{\psi(Z_i;\brho_0)}-\frac{\left\{\partial_{\brho_0}\psi(Z_i;\brho)\right\}^{\otimes 2}}{\psi^2(Z_i;\brho_0)}\right\}\Bigg)_{j,\ell}\Bigg.\nonumber\\
&\qquad\Bigg.\times\Bigg(\partial_{\brho_0}^2\int_0^t\psi(z;\brho)\ud z-\sum_{i=1}^{M(t)}\left\{\frac{\partial_{\brho_0}^2\psi(Z_i;\brho)}{\psi(Z_i;\brho_0)}-\frac{\left\{\partial_{\brho_0}\psi(Z_i;\brho)\right\}^{\otimes 2}}{\psi^2(Z_i;\brho_0)}\right\}\Bigg)_{m,k}\Bigg]\nonumber\\
&=\E\Bigg[\Bigg\{\partial^2_{\brho_0,j,\ell}\int_0^t\psi(z;\brho)\ud z\Bigg\}\Bigg\{\partial^2_{\brho_0,m,k}\int_0^t\psi(z;\brho)\ud z\Bigg\}\Bigg.\Bigg.\nonumber\\
&\quad\Bigg.\Bigg. + \Bigg\{\int_0^{t}\left(\frac{\partial_{\brho_0,j,\ell}^2\psi(z;\brho)}{\psi(z;\brho_0)}-\frac{\partial_{\brho_0,j}\psi(z;\brho)\partial_{\brho_0,\ell}\psi(z;\brho)}{\psi^2(z;\brho_0)}\right)\ud M(z)\Bigg\}\Bigg.\Bigg.\nonumber\\
&\qquad\Bigg.\Bigg. \times\Bigg\{\int_0^{t}\left(\frac{\partial_{\brho_0,m,k}^2\psi(z;\brho)}{\psi(z;\brho_0)}-\frac{\partial_{\brho_0,m}\psi(z;\brho)\partial_{\brho_0,k}\psi(z;\brho)}{\psi^2(z;\brho_0)}\right)\ud M(z)\Bigg\}\Bigg.\Bigg.\nonumber\\
&\quad\Bigg.\Bigg. -\Bigg\{\partial^2_{\brho_0,j,\ell}\int_0^t\psi(z;\brho)\ud z\Bigg\}\Bigg\{\int_0^{t}\left(\frac{\partial_{\brho_0,m,k}^2\psi(z;\brho)}{\psi(z;\brho_0)}-\frac{\partial_{\brho_0,m}\psi(z;\brho)\partial_{\brho_0,k}\psi(z;\brho)}{\psi^2(z;\brho_0)}\right)\ud M(z)\Bigg\}\Bigg.\Bigg.\nonumber\\
&\quad\Bigg.\Bigg. -\Bigg\{\partial^2_{\brho_0,m,k}\int_0^t\psi(z;\brho)\ud z\Bigg\}\Bigg\{\int_0^{t}\left(\frac{\partial_{\brho_0,j,\ell}^2\psi(z;\brho)}{\psi(z;\brho_0)}-\frac{\partial_{\brho_0,j}\psi(z;\brho)\partial_{\brho_0,\ell}\psi(z;\brho)}{\psi^2(z;\brho_0)}\right)\ud M(z)\Bigg\}\Bigg]\nonumber\\
&=\Bigg\{\partial^2_{\brho_0,j,\ell}\int_0^t\psi(z;\brho)\ud z\Bigg\}\Bigg\{\partial^2_{\brho_0,m,k}\int_0^t\psi(z;\brho)\ud z\Bigg\}\Bigg.\Bigg.\nonumber\\
&\quad\Bigg. + \int_0^{t}\left\{\frac{\partial_{\brho_0,j,\ell}^2\psi(z;\brho)}{\psi(z;\brho_0)}-\frac{\partial_{\brho_0,j}\psi(z;\brho)\partial_{\brho_0,\ell}\psi(z;\brho)}{\psi^2(z;\brho_0)}\right\}\Bigg.\Bigg.\Bigg.\nonumber\\
&\qquad\Bigg.\Bigg.\times\left\{\frac{\partial_{\brho_0,m,k}^2\psi(z;\brho)}{\psi(z;\brho_0)}-\frac{\partial_{\brho_0,m}\psi(z;\brho)\partial_{\brho_0,k}\psi(z;\brho)}{\psi^2(z;\brho_0)}\right\}\psi(z;\brho_0)\ud z\Bigg.\Bigg.\nonumber\\
&\quad\Bigg. + \Bigg[\int_0^{t}\left\{\frac{\partial_{\brho_0,j,\ell}^2\psi(z;\brho)}{\psi(z;\brho_0)}-\frac{\partial_{\brho_0,j}\psi(z;\brho)\partial_{\brho_0,\ell}\psi(z;\brho)}{\psi^2(z;\brho_0)}\right\}\psi(z;\brho_0)\ud z\Bigg]\Bigg.\Bigg.\nonumber\\
&\qquad\Bigg. \times\Bigg[\int_0^{t}\left\{\frac{\partial_{\brho_0,m,k}^2\psi(z;\brho)}{\psi(z;\brho_0)}-\frac{\partial_{\brho_0,m}\psi(z;\brho)\partial_{\brho_0,k}\psi(z;\brho)}{\psi^2(z;\brho_0)}\right\}\psi(z;\brho_0)\ud z\Bigg]\Bigg.\nonumber\\
&\quad\Bigg. -\Bigg\{\partial^2_{\brho_0,j,\ell}\int_0^t\psi(z;\brho)\ud z\Bigg\}\Bigg[\int_0^{t}\left\{\frac{\partial_{\brho_0,m,k}^2\psi(z;\brho)}{\psi(z;\brho_0)}-\frac{\partial_{\brho_0,m}\psi(z;\brho)\partial_{\brho_0,k}\psi(z;\brho)}{\psi^2(z;\brho_0)}\right\}\psi(z;\brho_0)\ud z\Bigg]\Bigg.\nonumber\\
&\quad\Bigg. -\Bigg\{\partial^2_{\brho_0,m,k}\int_0^t\psi(z;\brho)\ud z\Bigg\}\Bigg[\int_0^{t}\left\{\frac{\partial_{\brho_0,j,\ell}^2\psi(z;\brho)}{\psi(z;\brho_0)}-\frac{\partial_{\brho_0,j}\psi(z;\brho)\partial_{\brho_0,\ell}\psi(z;\brho)}{\psi^2(z;\brho_0)}\right\}\psi(z;\brho_0)\ud z\Bigg]\nonumber\\
&=\int_0^{t}\frac{1}{\psi(z;\brho_0)}\left\{\partial_{\brho_0,j,\ell}^2\psi(z;\brho)-\frac{\partial_{\brho_0,j}\psi(z;\brho)\partial_{\brho_0,\ell}\psi(z;\brho)}{\psi(z;\brho_0)}\right\}\nonumber\\
&\qquad\times\left\{\partial_{\brho_0,m,k}^2\psi(z;\brho)-\frac{\partial_{\brho_0,m}\psi(z;\brho)\partial_{\brho_0,k}\psi(z;\brho)}{\psi(z;\brho_0)}\right\}\ud z\nonumber\\
&\quad+\int_0^t\frac{\partial_{\brho_0,j}\psi(z;\brho)\partial_{\brho_0,\ell}\psi(z;\brho)}{\psi(z;\brho_0)}\ud z\int_0^t\frac{\partial_{\brho_0,m}\psi(z;\brho)\partial_{\brho_0,k}\psi(z;\brho)}{\psi(z;\brho_0)}\ud z.\label{eq:Ejlkm}
\end{align}
Moreover,
\begin{align}
&\E\left(\sum_{i=1}^{M(t)}\partial^2_{\brho_0}h(Z_i;\brho,t)\right)_{j,\ell}\E\left(\sum_{i=1}^{M(t)}\partial^2_{\brho_0}h(Z_i;\brho,t)\right)_{m,k}\nonumber\\
&=\E\Bigg(\partial_{\brho_0}^2\int_0^t\psi(z;\brho)\ud z-\sum_{i=1}^{M(t)}\left[\frac{\partial_{\brho_0}^2\psi(Z_i;\brho)}{\psi(Z_i;\brho_0)}-\frac{\left\{\partial_{\brho_0}\psi(Z_i;\brho)\right\}^{\otimes 2}}{\psi^2(Z_i;\brho_0)}\right]\Bigg)_{j,\ell}\nonumber\\
&\qquad\times\E\Bigg(\partial_{\brho_0}^2\int_0^t\psi(z;\brho)\ud z-\sum_{i=1}^{M(t)}\left[\frac{\partial_{\brho_0}^2\psi(Z_i;\brho)}{\psi(Z_i;\brho_0)}-\frac{\left\{\partial_{\brho_0}\psi(Z_i;\brho)\right\}^{\otimes 2}}{\psi^2(Z_i;\brho_0)}\right]\Bigg)_{m,k}\nonumber\\
&=\int_0^t\frac{\partial_{\brho_0,j}\psi(z;\brho)\partial_{\brho_0,\ell}\psi(z;\brho)}{\psi(z;\brho_0)}\ud z\int_0^t\frac{\partial_{\brho_0,m}\psi(z;\brho)\partial_{\brho_0,k}\psi(z;\brho)}{\psi(z;\brho_0)}\ud z.\label{eq:EjlEkm}
\end{align}
Thus,
\begin{align*}
\tr\Var\left\{\bV(t)\bss\right\}&=\tr\left[\E\left\{\bV(t)\bss\right\}^{\otimes 2}-\left\{\E\bV(t)\bss\right\}^{\otimes 2}\right]\\
&=\tr\int_0^{t}\mathcal{I}^{-1/2}(t,\brho_0)\mathcal{K}(z,\brho_0)\mathcal{I}^{-1/2}(t,\brho_0)\bss\bss^{\top}\mathcal{I}^{-1/2}(t,\brho_0)\mathcal{K}(z,\brho_0)\mathcal{I}^{-1/2}(t,\brho_0)\ud z\\
&=\bss^{\top}\left[\int_0^{t}\left\{\mathcal{I}^{-1/2}(t,\brho_0)\mathcal{K}(z,\brho_0)\mathcal{I}^{-1/2}(t,\brho_0)\right\}^2\ud z\right]\bss\to 0
\end{align*}
and, consequently, $\Var\left\{\bss^{\top}\bV(t)\bss\right\}\to 0$ as $t\to\infty$, because of Assumption~\ref{ass:infmat}~(ii).

Basic Corollary from~\cite{HjortPollard2011} can be now applied on representation~\eqref{eq:repre}, which directly provides the assertion of this theorem.
\end{proof}

\begin{proof}[Proof of Corollary~\ref{cor:asnorm}]
Let us choose $a\in\mathbb{N}$. One can observe that
\begin{multline*}
\mathcal{I}^{-1/2}(a,\brho_0)\sum_{i=1}^{M(a)}\partial_{\brho_0}h\left(Z_i;\brho_0,a\right)\\
=-\sum_{i=1}^{a}\int_{i-1}^{i}\mathcal{I}^{-1/2}(a,\brho_0)\left\{\partial_{\brho_0}\log\psi(z;\brho)\right\}\left(\ud M(z)-\psi(z;\brho_0)\ud z\right)=-\sum_{i=1}^a\bY_i
\end{multline*}
is a~sum of independent random vectors. Since
\[
\E\bY_i=\E\int_{i-1}^{i}\mathcal{I}^{-1/2}(a,\brho_0)\left\{\partial_{\brho_0}\log\psi(z;\brho)\right\}\left(\ud M(z)-\psi(z;\brho_0)\ud z\right)=\zero
\]
and
\begin{multline*}
\E\bY_i^{\otimes 2}=\E\left[\int_{i-1}^{i}\mathcal{I}^{-1/2}(a,\brho_0)\left\{\partial_{\brho_0}\log\psi(z;\brho)\right\}\left(\ud M(z)-\psi(z;\brho_0)\ud z\right)\right]^{\otimes 2}\\
=\mathcal{I}^{-1/2}(a,\brho_0)\int_{i-1}^i\frac{\left\{\partial_{\brho_0}\psi(z;\brho)\right\}^{\otimes 2}}{\psi^2(z;\brho_0)}\psi(z;\brho_0)\ud z\,\mathcal{I}^{-1/2}(a,\brho_0),
\end{multline*}
we have
\begin{align*}
&\Var\sum_{i=1}^a\bY_i=\Var\left\{\mathcal{I}^{-1/2}(a,\brho_0)\sum_{i=1}^{M(a)}\partial_{\brho_0}h\left(Z_i;\brho_0,a\right)\right\}\\
&=\mathcal{I}^{-1/2}(a,\brho_0)\int_{0}^a\frac{\left\{\partial_{\brho_0}\psi(z;\brho)\right\}^{\otimes 2}}{\psi(z;\brho_0)}\ud z\,\mathcal{I}^{-1/2}(a,\brho_0)=\mathcal{I}^{-1/2}(a,\brho_0)\mathcal{I}(a,\brho_0)\mathcal{I}^{-1/2}(a,\brho_0)={\boldsymbol I}.
\end{align*}

%Assumption~\ref{ass:Lind} allows us to apply the multivariate CLT for the martingale difference sequences---the univariate version~\citet[Theorem~24.4]{Davidson1994} together with the Cram\'{e}r-Wold device~\citep[Theorem~29.4]{Billingsley2008}---on $\frac{1}{\sqrt{n}}\bD(Z_i)$. Then,
The Lindeberg condition in Assumption~\ref{ass:Lind} and the Cram\'{e}r-Wold device~\citep[Theorem~29.4]{Billingsley2008} allow us to apply the multivariate central limit theorem (CLT) for the zero mean independent random vectors $\bY_i$'s. Thus,
\[
\mathcal{I}^{-1/2}(a,\brho_0)\sum_{i=1}^{M(a)}\partial_{\brho_0}h\left(Z_i;\brho_0,a\right)\xrightarrow[a\to\infty]{\dist}\mathsf{N}_{q}\left(\zero,{\boldsymbol I}\right).
\]
Hence, the desired convergence in distribution follows from the asymptotic representation in Theorem~\ref{thm:consistency}.
\end{proof}

\begin{proof}[Proof of Proposition~\ref{prop:infty}]
For the non-homogeneous Poisson process $\{M(t)\}_{t\geq 0}$ holds that $M(t)=\widetilde{M}(\Psi(t;\brho))$, where $\{\widetilde{M}(t)\}_{t\geq 0}$ is a~standard Poisson process (i.e., homogeneous Poisson process with intensity equal one). Suppose that $\widetilde{Z}_n$'s are arrival times of the standard Poisson process $\{\widetilde{M}(t)\}_{t\geq 0}$. Since
\begin{multline*}
\prob\left\{\lim_{t\to\infty}M(t)<\infty\right\}=\prob\left[\lim_{t\to\infty}\widetilde{M}\{\Psi(t;\brho)\}<\infty\right]=\prob\left(\widetilde{Z}_n-\widetilde{Z}_{n-1}=\infty\,\,\mbox{for some}\,\,n\right)\\
=\prob\left[\bigcup_{n=1}^{\infty} \left\{\widetilde{Z}_n-\widetilde{Z}_{n-1}=\infty\right\}\right]\leq\sum_{n=1}^{\infty}\prob\left(\widetilde{Z}_n-\widetilde{Z}_{n-1}=\infty\right)=0,
\end{multline*}
we have that $\lim_{t\to\infty}M(t)=\infty$ with probability one.
\end{proof}

\begin{proof}[Proof of Theorem~\ref{thm:consistencyN}]
Let us choose $t>0$. With respect to Assumption~\ref{ass:convexN}, consider the convex function
\begin{equation*}
G_t(\bss):=\sum_{i=1}^{M(t)}\left[g_i\left\{\bU_i;\bth_0+\mathcal{J}^{-1/2}(t,\bth_0)\bss,t\right\}-g_i(\bU_i;\bth_0,t)\right]
\end{equation*}
in $\bss\in\mathbb{R}^{p}$. It is minimized by $\mathcal{J}^{1/2}(t,\bth_0)\left(\widehat{\bth}-\bth_0\right)$. The Taylor series expansion gives
\begin{multline}\label{eq:repreN}
G_t(\bss)=\bss^{\top}\underbrace{\mathcal{J}^{-1/2}(t,\bth_0)\sum_{i=1}^{M(t)}\partial_{\bth_0}g_i\left(\bU_i;\bth,t\right)}_{=:\mathcal{U}(t)}\\
+\frac{1}{2}\bss^{\top}\underbrace{\mathcal{J}^{-1/2}(t,\bth_0)\left\{\sum_{i=1}^{M(t)}\partial^2_{\bth_0}g_i(\bU_i;\bth,t)\right\}\mathcal{J}^{-1/2}(t,\bth_0)}_{=:\mathcal{V}(t)}\bss+\mathcal{R}_t(\bss)
\end{multline}
almost surely, where $\mathcal{R}_t(\bss)=M(t)o\!\left(\bss^{\top}\mathcal{J}^{-1}(t,\bth_0)\bss\right)\to 0$ in probability for $t\to\infty$, because of Assumption~\ref{ass:infmatN}~{(i)}.

Assumption~\ref{ass:interchange1N} assures that there has to exist a~closed sub-neighborhood of $\bth_0$ denoted by~$\mathscr{V}_t(\brho_0)$ such that for all $\bth\in\mathscr{V}_t(\bth_0)$, it holds that $\int_{Z_i}^t\left|\partial_{\bth_0,j}\lambda(\tau,Z_i;\bth)\right|\ud \tau<\infty$ almost surely for all $j=1,\ldots,p$ and one can interchange derivative and integral, i.e., $\partial_{\bth_0}\int_{Z_i}^t\lambda(\tau,Z_i;\bth)\ud \tau=\int_{Z_i}^{t}\partial_{\bth_0}\lambda(\tau,Z_i;\bth)\ud \tau$.

Let us realize that for every $i\in\mathbb{N}$ the sequence $\{U_{i,k}\}_{k\in\mathbb{N}}$ forms arrival times of the Poisson counting process $\{N_i(t)\}_{t\geq 0}$ and, hence,
\begin{equation}\label{eq:intNi}
\sum_{k=1}^{N_i(t)}\frac{\partial_{\bth_0}\lambda(U_{i,k},Z_i;\bth)}{\lambda(U_{i,k},Z_i;\bth_0)}=\int_{Z_i}^{t}\frac{\partial_{\bth_0}\lambda(\tau,Z_i;\bth)}{\lambda(\tau,Z_i;\bth_0)}\ud N_i(\tau).
\end{equation}
In sequel, we use conditioning on $\{M(z)\}_{z\in[0,t]}$ (i.e., information contained in the process~$M$ up to time~$t$), which corresponds to conditioning on $Z_1,\ldots,Z_{M(t)}$. Recall that $g_i(\bU_i;\bth,t)=\int_{Z_i}^{t}\lambda(\tau,Z_i;\bth)\ud \tau-\sum_{k=1}^{N_i(t)}\log\lambda(U_{i,k},Z_i;\bth)$. Since $\frac{\partial_{\bth_0}\lambda(\cdot,Z_i;\bth)}{\lambda(\cdot,Z_i;\bth_0)}$ is continuous, we obtain
\begin{align*}
&\E \{\mathcal{U}(t)|\{M(z)\}_{z\in[0,t]}\}\\
&=\E\Bigg[\mathcal{J}^{-1/2}(t,\bth_0)\sum_{i=1}^{M(t)}\partial_{\bth_0}\Bigg\{\Lambda(t,Z_i;\bth)-\sum_{k=1}^{N_i(t)}\log\lambda(U_{i,k},Z_i;\bth)\Bigg\}\Bigg|\{M(z)\}_{z\in[0,t]}\Bigg]\\
&=\mathcal{J}^{-1/2}(t,\bth_0)\sum_{i=1}^{M(t)}\E\Bigg\{\partial_{\bth_0}\int_{Z_i}^t\lambda(\tau,Z_i;\bth)\ud \tau-\sum_{k=1}^{N_i(t)}\frac{\partial_{\bth_0}\lambda(U_{i,k},Z_i;\bth)}{\lambda(U_{i,k},Z_i;\bth_0)}\Bigg|\{M(z)\}_{z\in[0,t]}\Bigg\}\\
&=\mathcal{J}^{-1/2}(t,\bth_0)\sum_{i=1}^{M(t)}\E\Bigg\{\partial_{\bth_0}\int_{Z_i}^t\lambda(\tau,Z_i;\bth)\ud \tau-\int_{Z_i}^t\frac{\partial_{\bth_0}\lambda(\tau,Z_i;\bth)}{\lambda(\tau,Z_i;\bth_0)}\ud N_i(\tau)\Bigg|\{M(z)\}_{z\in[0,t]}\Bigg\}\\
&=\mathcal{J}^{-1/2}(t,\bth_0)\sum_{i=1}^{M(t)}\Bigg\{\partial_{\bth_0}\int_{Z_i}^t\lambda(\tau,Z_i;\bth)\ud \tau-\int_{Z_i}^t\frac{\partial_{\bth_0}\lambda(\tau,Z_i;\bth)}{\lambda(\tau,Z_i;\bth_0)}\lambda(\tau,Z_i;\bth_0)\ud \tau\Bigg\}=\zero.
\end{align*}
One can apply the It\^{o} isometry for jump processes in a~similar way as in~\eqref{eq:VarUt} and utilize that the processes $N_i$'s are independent
\begin{align*}
&\Var\{\mathcal{U}(t)|\{M(z)\}_{z\in[0,t]}\}\\
&=\sum_{i=1}^{M(t)}\E\Bigg[\mathcal{J}^{-1/2}(t,\bth_0)\partial_{\bth_0}\Bigg\{\Lambda(t,Z_i;\bth)-\sum_{k=1}^{N_i(t)}\log\lambda(U_{i,k},Z_i;\bth)\Bigg\}\Bigg|\{M(z)\}_{z\in[0,t]}\Bigg]^{\otimes 2}\\
&=\sum_{i=1}^{M(t)}\mathcal{J}^{-1/2}(t,\bth_0)\E\Bigg\{\partial_{\bth_0}\int_{Z_i}^t\lambda(\tau,Z_i;\bth)\ud \tau\Bigg.\\
&\quad\Bigg.-\int_{Z_i}^t\frac{\partial_{\bth_0}\lambda(\tau,Z_i;\bth)}{\lambda(\tau,Z_i;\bth_0)}\ud N_i(\tau)\Bigg|\{M(z)\}_{z\in[0,t]}\Bigg\}^{\otimes 2}\mathcal{J}^{-1/2}(t,\bth_0)\\
&=\mathcal{J}^{-1/2}(t,\bth_0)\sum_{i=1}^{M(t)}\Bigg\{\int_{Z_i}^{t}\frac{\left\{\partial_{\bth_0}\lambda(\tau,Z_i;\bth)\right\}^{\otimes 2}}{\lambda^2(\tau,Z_i;\bth_0)}\lambda(\tau,Z_i;\bth_0)\ud \tau\Bigg\}\mathcal{J}^{-1/2}(t,\bth_0),
\end{align*}
due to Assumption~\ref{ass:interchange1N}. Then, we get $\E\mathcal{U}(t)=\E[\E\{\mathcal{U}(t)|\{M(z)\}_{z\in[0,t]}\}]=\zero$ and
\begin{align*}
\Var\mathcal{U}(t)&=\Var[\E\{\mathcal{U}(t)|\{M(z)\}_{z\in[0,t]}\}]+\E[\Var\{\mathcal{U}(t)|\{M(z)\}_{z\in[0,t]}\}]\\
&=\mathcal{J}^{-1/2}(t,\bth_0)\E\sum_{i=1}^{M(t)}\mathcal{J}_i(t;\bth_0)\mathcal{J}^{-1/2}(t,\bth_0)={\boldsymbol I}.
\end{align*}
Moreover analogously as in~\eqref{eq:EVt},
\begin{align*}
\E \mathcal{V}(t)&=\E[\E\{\mathcal{V}(t)|\{M(z)\}_{z\in[0,t]}\}]\\
%&=\E\left\{\E\left[\mathscr{K}^{-1/2}(t,\bth_0)\sum_{i=1}^{M(t)}\partial_{\bth_0}^2\left\{\Lambda(t,Z_i;\bth)-\sum_{k=1}^{N_i(t)}\log\lambda(U_{i,k},Z_i;\bth)\right\}\mathscr{K}^{-1/2}(t,\bth_0)\Bigg|\{M(z)\}_{z\in[0,t]}\right]\right\}\\
&=\E\Bigg[\mathcal{J}^{-1/2}(t,\bth_0)\sum_{i=1}^{M(t)}\E\Bigg\{\partial_{\bth_0}^2\int_{Z_i}^t\lambda(\tau,Z_i;\bth)\ud \tau\Bigg.\Bigg.\\
&\quad\Bigg.\Bigg.-\int_{Z_i}^{t}\left(\frac{\partial_{\bth_0}^2\lambda(\tau,Z_i;\bth)}{\lambda(\tau,Z_i;\bth_0)}-\frac{\left\{\partial_{\bth_0}\lambda(\tau,Z_i;\bth)\right\}^{\otimes 2}}{\lambda^2(\tau,Z_i;\bth_0)}\right)\ud N_i(\tau)\Bigg|\{M(z)\}_{z\in[0,t]}\Bigg\}\mathcal{J}^{-1/2}(t,\bth_0)\Bigg]\\
&=\mathcal{J}^{-1/2}(t,\bth_0)\E\left[\sum_{i=1}^{M(t)}\int_{Z_i}^t\frac{\left\{\partial_{\bth_0}\lambda(\tau,Z_i;\bth)\right\}^{\otimes 2}}{\lambda(\tau,Z_i;\bth_0)}\ud \tau\right]\mathcal{J}^{-1/2}(t,\bth_0)\\
&=\mathcal{J}^{-1/2}(t,\bth_0)\mathcal{J}(t;\bth_0)\mathcal{J}^{-1/2}(t,\bth_0)={\boldsymbol I},
\end{align*}
because of Assumption~\ref{ass:interchange1N}. Furthermore for every $\bss\in\mathbb{R}^p$, it holds that
\[
\Var\left\{\bss^{\top}\mathcal{V}(t)\bss\right\}=\bss^{\top}\Var\left\{\mathcal{V}(t)\bss\right\}\bss=\bss^{\top}\left[\E\left\{\mathcal{V}(t)\bss\right\}^{\otimes 2}-\left\{\E\mathcal{V}(t)\bss\right\}^{\otimes 2}\right]\bss.
\]
The $(j,k)$-element of $\E\left\{\mathcal{V}(t)\bss\right\}^{\otimes 2}\equiv\E\left\{\mathcal{V}(t)\bss\bss^{\top}\mathcal{V}(t)\right\}$ has a~form of
\begin{multline}\label{eq:jkelementN}
\E\sum_{\ell=1}^p\sum_{m=1}^p s_{\ell}s_m\left(\mathcal{V}(t)\right)_{j,\ell}\left(\mathcal{V}(t)\right)_{m,k}=\E\sum_{\ell=1}^p\sum_{m=1}^p s_{\ell}s_m\\
\times\sum_{\tilde{\ell}=1}^p\sum_{\breve{\ell}=1}^p\kappa_{j,\tilde{\ell}}(t)\left(\sum_{i=1}^{n}\partial^2_{\bth_0}g_i(\bU_i;\bth,t)\right)_{\tilde{\ell},\breve{\ell}}\kappa_{\breve{\ell},\ell}(t)\sum_{\tilde{m}=1}^p\sum_{\breve{m}=1}^p\kappa_{m,\tilde{m}}(t)\left(\sum_{i=1}^{n}\partial^2_{\bth_0}g_i(\bU_i;\bth,t)\right)_{\tilde{m},\breve{m}}\kappa_{\breve{m},k}(t),
\end{multline}
where $\bss\bss^{\top}=\left(s_{\ell}s_m\right)_{\ell=1,m=1}^{p,p}$ and $\mathcal{J}^{-1/2}(t,\bth_0)=:\left(\kappa_{\ell,m}(t)\right)_{\ell=1,m=1}^{p,p}$. In a~similar fashion as in~\eqref{eq:Ejlkm}--\eqref{eq:EjlEkm} together with the independence of~$N_i$'s, we get
\begin{align}
&\E\left[\left(\sum_{i=1}^{M(t)}\partial^2_{\bth_0}g_i(\bU_i;\bth,t)\right)_{j,\ell}\left(\sum_{i=1}^{M(t)}\partial^2_{\bth_0}g_i(\bU_i;\bth,t)\right)_{m,k}\Bigg|\{M(z)\}_{z\in[0,t]}\right]\nonumber\\
&=\sum_{i=1}^{M(t)}\E\Bigg[\Bigg(\partial_{\bth_0}^2\int_{Z_i}^t\lambda(\tau,Z_i;\bth)\ud \tau\nonumber\\
&\quad\quad-\sum_{\tilde{i}=1}^{N_i(t)}\left\{\frac{\partial_{\bth_0}^2\lambda(U_{i,\tilde{i}},Z_i;\bth)}{\lambda(U_{i,\tilde{i}},Z_i;\bth_0)}-\frac{\left\{\partial_{\bth_0}\lambda(U_{i,\tilde{i}},Z_i;\bth)\right\}^{\otimes 2}}{\lambda^2(U_{i,\tilde{i}},Z_i;\bth_0)}\right\}\Bigg)_{j,\ell}\Bigg.\nonumber\\
&\quad\Bigg.\times\Bigg(\partial_{\bth_0}^2\int_{Z_i}^t\lambda(\tau,Z_i;\bth)\ud \tau\nonumber\\
&\quad\quad-\sum_{\tilde{i}=1}^{N_i(t)}\left\{\frac{\partial_{\bth_0}^2\lambda(U_{i,\tilde{i}},Z_i;\bth)}{\lambda(U_{i,\tilde{i}},Z_i;\bth_0)}-\frac{\left\{\partial_{\bth_0}\lambda(U_{i,\tilde{i}},Z_i;\bth)\right\}^{\otimes 2}}{\lambda^2(U_{i,\tilde{i}},Z_i;\bth_0)}\right\}\Bigg)_{m,k}\Bigg|\{M(z)\}_{z\in[0,t]}\Bigg]\nonumber\\
&\quad+\mathop{\sum_{i=1}^{M(t)}\!\sum_{\iota=1}^{M(t)}}_{i\neq\iota}\E\Bigg(\partial_{\bth_0}^2\int_{Z_i}^t\lambda(\tau,Z_i;\bth)\ud \tau\nonumber\\
&\quad\quad-\sum_{\tilde{i}=1}^{N_i(t)}\left\{\frac{\partial_{\bth_0}^2\lambda(U_{i,\tilde{i}},Z_i;\bth)}{\lambda(U_{i,\tilde{i}},Z_i;\bth_0)}-\frac{\left\{\partial_{\bth_0}\lambda(U_{i,\tilde{i}},Z_i;\bth)\right\}^{\otimes 2}}{\lambda^2(U_{i,\tilde{i}},Z_i;\bth_0)}\right\}\Bigg|\{M(z)\}_{z\in[0,t]}\Bigg)_{j,\ell}\nonumber\\
&\quad\times\E\Bigg(\partial_{\bth_0}^2\int_{Z_{\iota}}^t\lambda(\tau,Z_{\iota};\bth)\ud \tau\nonumber\\
&\quad\quad-\sum_{\tilde{\iota}=1}^{N_{\iota}(t)}\left\{\frac{\partial_{\bth_0}^2\lambda(U_{\iota,\tilde{\iota}},Z_{\iota};\bth)}{\lambda(U_{\iota,\tilde{\iota}},Z_{\iota};\bth_0)}-\frac{\left\{\partial_{\bth_0}\lambda(U_{\iota,\tilde{\iota}},Z_{\iota};\bth)\right\}^{\otimes 2}}{\lambda^2(U_{\iota,\tilde{\iota}},Z_{\iota};\bth_0)}\right\}\Bigg|\{M(z)\}_{z\in[0,t]}\Bigg)_{m,k}\nonumber\\
&=\sum_{i=1}^{M(t)}\int_{Z_i}^{t}\frac{1}{\lambda(\tau,Z_i;\bth_0)}\left\{\partial_{\bth_0,j,\ell}^2\lambda(\tau,Z_i;\bth)-\frac{\partial_{\bth_0,j}\lambda(\tau,Z_i;\bth)\partial_{\bth_0,\ell}\lambda(\tau,Z_i;\bth)}{\lambda(\tau,Z_i;\bth_0)}\right\}\nonumber\\
&\quad\times\left\{\partial_{\bth_0,m,k}^2\lambda(\tau,Z_i;\bth)-\frac{\partial_{\bth_0,m}\lambda(\tau,Z_i;\bth)\partial_{\bth_0,k}\lambda(\tau,Z_i;\bth)}{\lambda(\tau,Z_i;\bth_0)}\right\}\ud \tau\nonumber\\
&\quad +\left\{\sum_{i=1}^{M(t)}\int_{Z_i}^{t}\frac{\partial_{\bth_0,j}\lambda(\tau,Z_i;\bth)\partial_{\bth_0,\ell}\lambda(\tau,Z_i;\bth)}{\lambda(\tau,Z_i;\bth_0)}\ud\tau\right\}\left\{\sum_{i=1}^{M(t)}\int_{Z_i}^{t}\frac{\partial_{\bth_0,m}\lambda(\tau,Z_i;\bth)\partial_{\bth_0,k}\lambda(\tau,Z_i;\bth)}{\lambda(\tau,Z_i;\bth_0)}\ud\tau\right\}.\label{eq:EjlkmN}
\end{align}
Moreover,
\begin{align}
&\E\left(\sum_{i=1}^{M(t)}\partial^2_{\bth_0}g_i(\bU_i;\bth,t)\Bigg|\{M(z)\}_{z\in[0,t]}\right)_{j,\ell}\E\left(\sum_{i=1}^{M(t)}\partial^2_{\bth_0}g_i(\bU_i;\bth,t)\Bigg|\{M(z)\}_{z\in[0,t]}\right)_{m,k}\nonumber\\
&=\left\{\sum_{i=1}^{M(t)}\int_{Z_i}^{t}\frac{\partial_{\bth_0,j}\lambda(\tau,Z_i;\bth)\partial_{\bth_0,\ell}\lambda(\tau,Z_i;\bth)}{\lambda(\tau,Z_i;\bth_0)}\ud\tau\right\}\left\{\sum_{i=1}^{M(t)}\int_{Z_i}^{t}\frac{\partial_{\bth_0,m}\lambda(\tau,Z_i;\bth)\partial_{\bth_0,k}\lambda(\tau,Z_i;\bth)}{\lambda(\tau,Z_i;\bth_0)}\ud\tau\right\}.\label{eq:EjlEkmN}
\end{align}
Thus,
\begin{align*}
&\tr\Var\left\{\mathcal{V}(t)\bss\right\}=\tr\E\left[\Var\left\{\mathcal{V}(t)\bss|\{M(z)\}_{z\in[0,t]}\right\}\right]\\
&\quad+\tr\Var\left[\E\left\{\mathcal{V}(t)\bss|\{M(z)\}_{z\in[0,t]}\right\}\right]\\
&=\tr\E\left[\E\left\{\mathcal{V}(t)\bss|\{M(z)\}_{z\in[0,t]}\right\}^{\otimes 2}-\left\{\E\left(\mathcal{V}(t)\bss|\{M(z)\}_{z\in[0,t]}\right)\right\}^{\otimes 2}\right]\\
&\quad+\tr\E\left[\E\left\{\mathcal{V}(t)\bss|\{M(z)\}_{z\in[0,t]}\right\}\right]^{\otimes 2}-\tr\left[\E\left\{\E\left(\mathcal{V}(t)\bss|\{M(z)\}_{z\in[0,t]}\right)\right\}\right]^{\otimes 2}\\
&=\tr\E\sum_{i=1}^{M(t)}\int_{Z_i}^{t}\mathcal{J}^{-1/2}(t,\bth_0)\mathcal{L}_i(\tau,\bth_0)\mathcal{J}^{-1/2}(t,\bth_0)\bss\bss^{\top}\mathcal{J}^{-1/2}(t,\bth_0)\mathcal{L}_i(\tau,\bth_0)\mathcal{J}^{-1/2}(t,\bth_0)\ud \tau\\
&=\bss^{\top}\left[\E\sum_{i=1}^{M(t)}\int_{Z_i}^{t}\left\{\mathcal{J}^{-1/2}(t,\bth_0)\mathcal{L}_i(\tau,\bth_0)\mathcal{J}^{-1/2}(t,\bth_0)\right\}^2\ud \tau\right]\bss\to 0,\quad t\to 0
\end{align*}
and, consequently, $\Var\left\{\bss^{\top}\mathcal{V}(t)\bss\right\}\to 0$ as $t\to\infty$, because of Assumption~\ref{ass:infmatN}~(ii).

Finally, Basic Corollary from~\cite{HjortPollard2011} can be applied on representation~\eqref{eq:repreN}, which directly provides the assertion of this theorem.
\end{proof}

\begin{proof}[Proof of Corollary~\ref{cor:asnormN}]
Let us choose $a\in\mathbb{N}$. Recall that $\tilde{N}_{z}(\tau)=\sum_{k=1}^{\infty}\mathbbm{1}\{U_{i,k}-z\leq\tau\}$ is the restarted process of~$N_i$. One can observe that
\begin{align*}
&\mathcal{J}^{-1/2}(a,\bth_0)\sum_{i=1}^{M(a)}\partial_{\bth_0}g_i\left(\bU_i;\bth_0,a\right)\\
&=\sum_{i=1}^{M(a)}\mathcal{J}^{-1/2}(a,\bth_0)\partial_{\bth_0}\left\{\int_{Z_i}^{a}\lambda(\tau,Z_i;\bth)\ud \tau-\sum_{k=1}^{N_i(a)}\log\lambda(U_{i,k},Z_i;\bth)\right\}\\
&=-\sum_{i=1}^{M(a)}\mathcal{J}^{-1/2}(a,\bth_0)\left[\int_{Z_i}^{a}\left\{\partial_{\bth_0}\log\lambda(\tau,Z_i;\bth)\right\}\left(\ud N_i(\tau)-\lambda(\tau,Z_i;\bth_0)\ud\tau\right)\right]\\
&=-\sum_{j=1}^{a}\int_{j-1}^{j}\mathcal{J}^{-1/2}(a,\bth_0)\left[\int_{z}^{a}\left\{\partial_{\bth_0}\log\lambda(\tau,z;\bth)\right\}\left(\ud\tilde{N}_z(\tau-z)-\lambda(\tau,z;\bth_0)\ud\tau\right)\right]\ud M(z)\\
&=-\sum_{i=1}^a\mathcal{Y}_i
\end{align*}
is a~sum of independent random vectors. Since
\begin{align*}
\E\mathcal{Y}_i&=\E\int_{j-1}^{j}\mathcal{J}^{-1/2}(a,\bth_0)
\left[\int_{z}^{a}\left\{\partial_{\bth_0}\log\lambda(\tau,z;\bth)\right\}\left(\ud\tilde{N}_z(\tau-z)-\lambda(\tau,z;\bth_0)\ud\tau\right)\right]\ud M(z)\\
&=\zero,
\end{align*}
we have
\[
\Var\sum_{i=1}^a\mathcal{Y}_i=\Var\left\{\mathcal{J}^{-1/2}(a,\bth_0)\sum_{i=1}^{M(a)}\partial_{\bth_0}g_i\left(\bU_i;\bth_0,a\right)\right\}=\mathcal{J}^{-1/2}(a,\bth_0)\mathcal{J}(a,\bth_0)\mathcal{J}^{-1/2}(a,\bth_0)={\boldsymbol I}.
\]

The Lindeberg condition in Assumption~\ref{ass:lindN} and the Cram\'{e}r-Wold device~\citep[Theorem~29.4]{Billingsley2008} allow us to apply the multivariate CLT for the zero mean independent random vectors $\mathcal{Y}_i$'s. Then,
\[
\mathcal{J}^{-1/2}(a,\bth_0)\sum_{i=1}^{M(t)}\partial_{\bth_0}g_i\left(\bU_i;\bth_0,a\right)\xrightarrow[a\to\infty]{\dist}\mathsf{N}_{p}\left(\zero,{\boldsymbol I}\right).
\]
To conclude, the desired convergence in distribution follows from the asymptotic representation in Theorem~\ref{thm:consistencyN}.
\end{proof}
%}

\end{document}